\documentclass{article}
\usepackage{subcaption}
\usepackage{graphicx} % Required for inserting images
\usepackage{authblk}
\usepackage{soul}

\usepackage[thmtools-compat]{keytheorems}
\usepackage{amsmath,amssymb,mathtools}
\usepackage{hyperref}

\usepackage{zref-clever}
\newcommand{\cref}[1]{\zcref{#1}}
\newcommand{\Cref}[1]{\zcref[S]{#1}}

\usepackage{color}
\usepackage{geometry}
\usepackage{braket}
\usepackage{todonotes}
\usepackage{pdfpages}
\usepackage{comment}
\usepackage{circuitikz}
\usetikzlibrary{arrows.meta}
\usetikzlibrary{shapes.geometric}
\usepackage{algorithm,algpseudocode}

\newcommand{\R}{\mathbb R}
\newcommand{\E}{\mathbb E}

\newcommand{\Z}{\mathbb Z}
\newcommand{\T}{\top}

\newcommand{\sfe}{\mathbb{S}^{d-1}}
\newcommand{\dd}{\operatorname{d}}

\newcommand{\vol}{\operatorname{vol}}
\newcommand{\feastol}{\mathrm{feasTol}}
\newcommand{\opttol}{\mathrm{optTol}}
\newcommand{\ball}{\operatorname{\mathbb{B}_2}}

\DeclarePairedDelimiter\abs{\lvert}{\rvert}

\def\ve#1{\mathchoice{\mbox{\boldmath$\displaystyle\bf#1$}}
{\mbox{\boldmath$\textstyle\bf#1$}}
{\mbox{\boldmath$\scriptstyle\bf#1$}}
{\mbox{\boldmath$\scriptscriptstyle\bf#1$}}}

\newcommand\vea{{\ve a}}
\newcommand\veb{{\ve b}}
\newcommand\vecc{{\ve c}}

\newcommand\vece{{\ve e}}
\newcommand\vef{{\ve f}}

\newcommand\vep{{\ve p}}

\newcommand\ves{{\ve s}}
\newcommand\vet{{\ve t}}

\newcommand\vev{{\ve v}}

\newcommand\vex{{\ve x}}
\newcommand\vey{{\ve y}}
\newcommand\vez{{\ve z}}

\newcommand\veo{{\ve 0}}
\newcommand\vetheta{\ve\theta}
\newcommand\velambda{\ve\lambda}

\newcommand{\norm}[1]{\|#1\|}

\newcommand{\eps}{\varepsilon} % bestsilon
\renewcommand{\epsilon}{\varepsilon} % bestsilon

\newtheorem{theorem}{Theorem}
\newtheorem{remark}[theorem]{Remark}

\newtheorem{lemma}[theorem]{Lemma}
\newtheorem{definition}[theorem]{Definition}

\definecolor{eleoncolor}{rgb}{0.8, 0.8, 1}
\definecolor{sophcolor}{rgb}{0.8, 1, 0.95}

\definecolor{alexcolor}{rgb}{0.3, .5, 0.7}

\usepackage[style=alphabetic,maxnames=100,minnames=100,minalphanames=6,maxalphanames=6]{biblatex}
\title{Beyond Smoothed Analysis: Analyzing the Simplex Method by-the-book}

\author[1]{Eleon Bach}
\author[2]{Alexander E. Black}
\author[3]{Sophie Huiberts}
\author[4]{Sean Kafer}

\affil[1]{Technische Universit\"{a}t M\"{u}nchen}
\affil[2]{Bowdoin College}
\affil[3]{LIMOS, CNRS, University Clermont Auvergne}
\affil[4]{Illinois State University}
\date{}

\begin{document}

\maketitle

\begin{abstract}
Narrowing the gap between theory and practice is a longstanding goal of the algorithm analysis community.
To further progress our understanding of how algorithms work in practice, we propose a new algorithm analysis framework that we call \emph{by-the-book analysis}.
In contrast to earlier frameworks, by-the-book analysis not only models an algorithm's input data, but also the algorithm itself.
Results from by-the-book analysis are meant to correspond well with established knowledge of an algorithm's practical behavior, as they are meant to be grounded in observations from implementations, input modeling best practices, and measurements on practical benchmark instances. 
We apply our framework to the simplex method, an algorithm which is beloved for its excellent performance in practice and notorious for its high running time under worst-case analysis.
The simplex method similarly showcased the previous state of the art framework \emph{smoothed analysis} (Spielman and Teng, STOC'01).
We explain how our framework overcomes several weaknesses of smoothed analysis and we 
prove that under input scaling assumptions, feasibility tolerances and other design principles used by simplex method implementations, the simplex method indeed attains a polynomial running time. 
Our results provide analytical justification for these features which are common to all high-quality simplex method implementations.
\end{abstract}

\maketitle

\section{Introduction}
There is a longstanding need for frameworks explaining the performance of algorithms whose users---whether they are engineers, operations researchers or mathematicians---know that they perform very well for their applications, but for which traditional worst-case analysis is too pessimistic.
The simplex method for linear programming (LP) problems is an example of such an algorithm: it was a driving force in the development of electronic computers \cite{marosbook} as it was observed to perform efficiently;
requiring a linear number of pivot steps in practice \cite{randibmmanual, dantzigbook,sha87, andrei2004complexity,makhoringlpk, xpress}.
A long line of research attempts to explain this behavior theoretically by a suitable algorithm analysis framework.
The project of explaining the performance of the simplex method is not only interesting from a mathematical point of view.
There is also interest from practical users who are constantly facing new areas of application. They might fear that the algorithm will have worse performance in their new application domains.
As such, there is a need for theoretical models showing that the simplex method is guaranteed to be efficient. 
%could lead to a non-polynomial worse performance leading to the question if practical simplex methods indeed run in polynomial time. 

The simplex algorithm serves well as a showcase for the developments in the algorithm analysis community. 
Following the traditional line of algorithm analysis research, in 1972, the first exponential worst-case analysis lower bounds were shown for the number of pivot steps the simplex method may take with certain pivot rules, and many analogous results followed \cite{km72, jer73, AC78, GS79, Mur80, g83, k92,msw96, jour/cm/AZ98, conf/stoc/FHZ11, conf/ipco/Friedmann11,hz15, disser2020exponential, simplexzeroone,black22, dissermosis, black2024exponentiallowerboundspivot}. This discovery opened up the search for algorithm analysis frameworks which could more properly capture the behavior of the simplex algorithm and algorithms which showed a similar behavior.

For algorithms like these, average-case analysis is the theoretician's next attempt to provide a better understanding of the running time.
In the case of the simplex method, average-case analysis saw much activity during the 1970s and 1980s, and as a result, we know tight polynomial upper and lower bounds in the average case setting \cite{thesis/Borgwardt77, b82,b87, report/Haimovich83, jour/jacm/AM85, jour/mapr/Megiddo86, jour/mapr/Todd86, jour/jc/AKS87, b87,b99,bdghl21}. 
However, linear programs seen in the average case setting greatly differ on a structural level from linear programs seen in practice. As such, average case analysis offered an incomplete explanation. Despite the development of algorithms for linear programming which run in polynomial time under a worst-case analysis, the simplex method is still an essential part of all linear programming software.
As such, there remains a need to understand its performance in practice.

In their seminal work \cite{ST04}, Spielman and Teng proposed the smoothed analysis framework as an explanation as to ``why the simplex method usually takes polynomial time," and which sought to incorporate the benefits of randomness in average case analysis while retaining some of the structure found in practical linear programs.
Their framework was subsequently applied to many more algorithms (see \cite{bwcabook} for an overview).
In \Cref{sec:smoothed_intro}, we review smoothed analysis of the simplex method.
In \Cref{sec:smoothed_limitations}, we discuss the limitations of the smoothed analysis framework as a model for the simplex method.
These limitations undermine the ability of smoothed analysis to explain why the simplex method usually runs in polynomial time.

In \Cref{sec:BTB_intro}, we describe our new algorithm analysis framework and how our framework is overcoming these limitations.
 
 Using our by-the-book-analysis framework, 
 we explain how a simplex method with bound perturbations is guaranteed to solve LPs in expected polynomial time.

\subsection{Smoothed Analysis and the Simplex Method}\label{sec:smoothed_intro}
%In this section we review the smoothed analysis of the simplex method.
We first recall the basic principles of the simplex method. The simplex method is best thought of as a class of algorithms. A simplex method first determines a basic feasible solution using a \emph{Phase I} procedure and then iteratively moves to new basic feasible solutions until an optimal solution has been found. The individual moves are called \emph{pivot steps}, and the number of pivot steps is a proxy for the running time. The choice for which basic feasible solution to move to is governed by a \emph{pivot rule}.
Some popular examples of pivot rules are the most negative reduced cost rule \cite{dan51}, the steepest edge rule and its approximations \cite{harris, gol76, jour/mapr/FG92}, and the shadow vertex rule, also known as the parametric rule \cite{gas55, thesis/Borgwardt77}.
It is this last pivot rule that is the foundational tool for most probabilistic analyses, including all results in smoothed analysis.

In the smoothed analysis of linear programming, we assume that an adversary specifies an~LP
\begin{align*}
    \operatorname{maximize} \quad & \vecc^\T \vex \\
    \operatorname{subject~to} \quad & \bar A\vex \leq \bar \veb,
\end{align*}
and that subsequently a perturbation $\hat A \in \R^{n \times d}, \hat \veb \in \R^n$ is sampled.
We assume that the entries of $(\hat A, \hat \veb)$ are independently sampled from the Gaussian distribution with mean $0$ and standard deviation $\sigma > 0$,
and we assume that the rows of the extended matrix $(\bar A, \bar \veb)$ each have Euclidean norm at most $1$.
Under these assumptions, smoothed analysis aims to develop simplex methods for solving
\begin{align*}
    \operatorname{maximize} \quad & \vecc^\T \vex \\
    \operatorname{subject~to} \quad & (\bar A + \hat A)\vex \leq \bar \veb + \hat \veb,
\end{align*}
for which the expected running time can be bounded by a polynomial function in $\sigma^{-1}, d$ and $n$.
After a line of work \cite{ST04, ds05, ver09, thesis/Schnalzger14, DH18, hlz}, Bach and Huiberts \cite{optimal_smoothed_analysis} found a simplex method that requires no more than $$O(\sigma^{-1/2} d^{11/4} \log(n)^{7/4})$$ pivot steps.
They also described a lower bound which applies to all pivot rules, stating that $$\Omega(\sigma^{-1/2} d^{1/2}\ln(1/\sigma)^{-1/4})$$ pivot steps are necessary when $n = (4/\sigma)^d$.

Despite these upper and lower bounds agreeing on the exponent of $\sigma^{-1/2}$, smoothed analysis can hardly be said to be a complete explanation of the simplex method's running time. The LPs in smoothed analysis are unlike LPs in practice in a number of important ways, of which we describe three in the next section.
These three differences undermine the three leading interpretations given for the smoothed analysis of the simplex method.

\subsection{Limitations of Smoothed Analysis}\label{sec:smoothed_limitations}

\paragraph{Sparse perturbations}
Most linear programs seen in practice are very \emph{sparse}; less than 0.1\% of their entries are non-zero (see, e.g., \cite{marosbook, Hall2005, miplib,gay1985electronic}).
The same holds true for known worst-case inputs.
This stands in stark contrast with smoothed analysis, where 100\% of the entries are non-zero since they are randomly perturbed according to a continuous probability distribution.
As such, in the smoothed probability model, the set of sparse linear programs has measure $0$.

Since both practical linear programs and theoretical worst-case inputs are sparse, it remains unclear whether smoothed analysis provides a viable synthetic model for describing the ``brittleness'' of worst-case inputs and whether smoothed analysis indeed helps to explain why the simplex method is usually fast on inputs observed in practice.

In their foundational work introducing smoothed analysis \cite{ST04}, Spielman and Teng highlighted the lack of sparsity as a weakness of their model.
As a possible fix for this problem, they proposed zero-preserving and multiplicative perturbations.
However, we note that exponential running times for simplex methods still occur in this zero-preserving model, even for constant-size perturbations. Consider the Klee-Minty instance \cite{km72}:
\begin{align*}
    \operatorname{maximize} \quad &\vex_d \\
    \operatorname{subject~to} \quad & 0 \leq \vex_1 \leq 1 \\
    & \eps \vex_{j-1} \leq \vex_j \leq 1-\eps \vex_{j_1} \qquad \forall j \in \{2,\dots,d\}.
\end{align*}
This construction gives exponential-length simplex paths for Bland's least-index rule for any value of $0 < \eps < 1/2$.
The different occurrences of $\eps$ need not be equal for their proof to work.
As such, if we take an instance of the KM-cube with $\eps=1/4$ and perturb all non-zero entries with small but constant-size noise, the same simplex path will remain a simplex path for Bland's rule.
This demonstrates that the simplex method with Bland's rule has exponential smoothed running time for both zero-preserving and multiplicative perturbations.
No other sparse smoothed analysis model has been proposed.

\paragraph{Matrix entry precision}
When LP constraint data is noisy or is rounded to low precision, this affects LP solver behavior negatively. For example, a nominally singular set of constraints might become a basis (potentially an ill-conditioned one) due to improper rounding, and this can lead to performance degradation.
For this reason it is recommended to specify all constraints with maximum numerical precision (see, e.g., \cite{gurobi,klotzseminar}): when possible use all-integer inputs, otherwise prefer double-precision over single-precision floating point numbers.
Higher-precision matrix entries tend to give better performance.

In contrast, smoothed analysis produces stronger performance guarantees if \textit{more} noise is added to the constraint matrix, thereby seeming to predict that lowering the precision of constraint data yields performance improvement.
This appears to be a disagreement between practical advice and smoothed analysis' theoretical prediction.
This disagreement indicates that smoothed analysis might not be an 
appropriate model for the effects of round-off error or floating-point imprecision.

\paragraph{Independence of matrix entries}
A small perturbation to the constraint matrix can lead to a large change in the feasible region.
This fact has long been observed on practical LPs in the robust optimization literature.
For example, \cite{robustoptimization} describes that the instance \texttt{PILOT4} has constraints whose entries are seemingly random with 7 digits significant.
One example they give is constraint \texttt{372} which reads
\begin{align*}
    \vea^\T \vex \equiv &- 15.79081\vex_{826} - 8.598819\vex_{827} - 1.88789\vex_{828} - 1.362417\vex_{829} - 1.526049\vex_{830} \\
&-0.031883\vex_{849} - 28.725555\vex_{850} - 10.792065\vex_{851} - 0.19004\vex_{852} - 2.757176\vex_{853} \\
&-12.290832\vex_{854} + 717.562256\vex_{855} - 0.057865\vex_{856} - 3.785417\vex_{857} - 78.30661\vex_{858} \\
&-122.163055\vex_{859} - 6.46609\vex_{860} - 0.48371\vex_{861} - 0.615264\vex_{862} - 1.353783\vex_{863} \\
&-84.644257\vex_{864} - 122.459045\vex_{865} - 43.15593\vex_{866} - 1.712592\vex_{870} - 0.401597\vex_{871} \\
&+\vex_{880} - 0.946049\vex_{898} - 0.946049\vex_{916} \geq b \equiv 23.387405.
\end{align*}
When the ``ugly reals''\footnote{We use the terminology used in \cite{robustoptimization}.} entries are each independently randomly changed by as little as $0.1\%$, then the original LP solution has a mean violation $\max(0,\frac{b - \vea^\T \vex^*}{b})$ of $150\%$.
Surely we must assume that even these LPs, which are sensitive to small \emph{independent} random perturbations, were formulated appropriately by the modeler for their purpose.
Despite looking random, these entries may have non-trivial relations between them which are destroyed when perturbing independently.

We interpret this to imply that the constraint data is highly accurate, in the sense that the constraint matrix entries cannot be independently perturbed by small amounts.
In this manner, we argue that the independence assumption makes smoothed analysis an inappropriate model for measurement error or other forms of limited precision input.

\paragraph{What is modelled}
The variable nature of $\sigma$ leaves smoothed analysis agnostic to what it models.
Although different interpretations have been given in the literature, they largely fall within one of three categories:
smoothed analysis models the ``brittleness'' or ``isolation'' of worst-case instances,
smoothed analysis models the effects of numerical imprecision in floating-point arithmetic,
or smoothed analysis models the effects of measurement error or other arbitrary influences.
As we have seen above, each of these three interpretations has notable weaknesses.
We must note that, when considered as a \emph{scientific theory}, the flexibility of interpretation in regards to the meaning of $\sigma$ is not a strength of the smoothed analysis model.

We conclude that smoothed analysis is a breakthrough framework that led to sound mathematical theorems.
As a scientific theory, there are several interpretations about what smoothed analysis purports to model. This ambiguity results from the fact that its inherent assumptions are made for mathematical simplicity, without being traceable to direct observations.
In this work we will propose an alternative analytical framework designed to\textemdash among other things\textemdash address precisely this weakness of smoothed analysis.
We offer what we argue is a much improved explanation for the practical behavior of the simplex method and we expect this new framework to offer improved explanations on the behavior of other algorithms as well.  Although we seek specifically to deviate from those aspects of smoothed analysis which we have identified above as weaknesses, the proof techniques which were developed in the smoothed analysis literature will play a crucial role in the proposed \textit{by-the-book analysis} presented in this paper; a framework we detail in the next section.

\subsection{Introducing By-The-Book Analysis}\label{sec:BTB_intro}
In this paper we propose a new algorithm analysis framework that we call by-the-book analysis.
Unlike smoothed analysis, by-the-book analysis is explicitly empirical by being intentional about what effect it models. 
As part of this, we open up the possibility of modeling not only the input data but also aspects of an algorithm's implementation.
It is well-known that theoretical descriptions of algorithms differ from their software implementations.
These differences are often thought to be incidental and to be theoretically uninteresting, even though they have repeatedly proven to be beneficial in practice.
We posit that these differences, if they are found in state of the art implementations, can be essential to a full theoretical understanding.
In doing so, by-the-book analysis answers the call of \cite{Hooker1994} for an empirically-based explanatory theory capturing ``the all-important tricks that are engineered into commercial codes.''

With by-the-book analysis, we strive to prove performance guarantees while using assumptions that are maximally grounded in computational experience.
For this reason, our research process has occurred in three phases as depicted in \Cref{fig:schematic}.  We will first describe what we propose as the general framework of by-the-book analysis, and then we will exemplify this framework by outlining the specifics of our own by-the-book analysis of the simplex method.

%\begin{figure}[!ht]
%\centering
%\includegraphics[width=\columnwidth]{images/btbadiagram.pdf}
%\caption{Schematic depiction of by-the-book analysis}
%\label{fig:schematic}
%\end{figure}
\begin{figure}[!ht]
\centering
\includegraphics[width=0.55\textwidth]{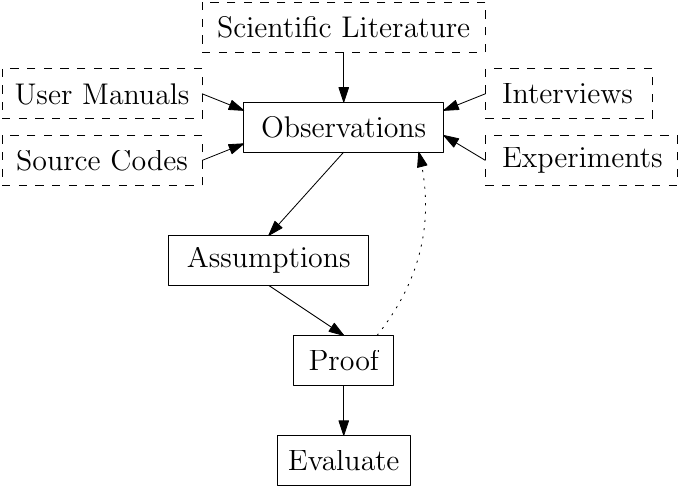}
\caption{Schematic depiction of by-the-book analysis}
\label{fig:schematic}
\end{figure}

% 7: HiGHS, OMP, Google, Porrier, FICO, SoPlex, Gurobi. The other 3 we spoke to: MOSEK, Hexaly and COPT were interviewed only later in the process

In general, by-the-book analysis proceeds in three phases. In the first phase, one establishes an understanding of true and relevant features of real-world implementations of an algorithm and features of real-world inputs.  This process can involve reading of the scientific literature, open-source code, user manuals, running experiments etc. in order to establish how the algorithm is implemented in practice, how users of the algorithm are instructed to formulate their inputs, and what features are common among real problem instances.

The second phase is a process which converts parts of this understanding to a mathematical model of the algorithm on which some form of traditional theoretical analysis can be performed.   Based on the observations made in the first phase, one chooses an algorithm description and a collection of assumptions satisfied by the algorithm and/or its inputs.  The chosen algorithm description may deviate from traditional theoretical descriptions of the algorithm if that deviation is warranted by observations. One might assume, for example, that certain relevant parameters of the input are bounded.  Although one must inevitably make assumptions that make the upcoming mathematical analysis tractable, one should, to the greatest extent possible, choose a mathematical model and mathematical assumptions that comport with the observations made in the first phase. This is what gives a by-the-book analysis its empirical basis.

The third phase is a traditional mathematical analysis of the chosen model under the chosen assumptions.  In this phase, by-the-book analysis is agnostic to the mathematical techniques used provided that any further assumptions inherent to those techniques are likewise grounded (to the greatest extent possible) in the observations of the first phase.  For example, in this phase the mathematical analysis could resemble parts of a worst-case, average-case, smoothed analysis, etc. Ultimately, what makes something a by-the-book analysis is not the choice of mathematical tools used, but that the model, assumptions, and tools are grounded in observation.

Note that the order of the phases is important.  In particular, the determination of mathematical assumptions occurs \textit{before} the production of theoretical running time bounds. It is contrary to by-the-book analysis to produce a theoretical bound parameterizing by whichever parameters turn out to be mathematically fruitful and to only then search for explanations as to why these parameters may be bounded in practice.  By \textit{basing} the theoretical approach on observations, we argue that this methodology is better positioned to produce accurate explanations.  This is the ultimate objective of performing a by-the-book analysis.  In this spirit, we emphasize again that what makes something a by-the-book analysis is not which type of mathematical tools or proof techniques are used, but whether or not the research methodology follows the general framework outlined above.  To call something a by-the-book analysis (or not) is not a mathematical distinction, but a methodological one; it is not a question of \textit{what} mathematical assumptions are made, but of \textit{how} they are produced in the first place.

We will now briefly outline how we applied the above procedure to the simplex method in the present analysis.
In the initial phase we studied the scientific literature, source-code, and user manuals of simplex implementations, and we interviewed developers working on multiple different LP software packages.
A primary role in this phase was played by the titular textbooks \cite{orchardhaysbook, marosbook} which describe how to implement a simplex method.
After observing the prominent role that numerical tolerances are given in both books (on which we will elaborate later), we verified their importance by way of inspecting both the source code of open-source LP solvers and the user manuals of closed-source LP solvers.
Because the simplex method is a popular and well-documented algorithm, we were able to draw on a diverse array of sources and identify points of broad agreement.
The outcomes of this phase are described in detail in \Cref{sec:bythebooksimplexfindings}.

In the second phase, we balanced verisimilitude with our observations against the tractability of the upcoming mathematical analysis.
We aimed for a traditional, theoretical performance guarantee which, when all parameters are filled in with plausible values, gives the strongest possible conclusion in terms of how well this bound is reflecting an algorithm's behavior in practice (subject, of course, to adherence to our observations in the first phase). We observed that every major modern simplex solver perturbs its right-hand side vector with random exponentially distributed noise (whose magnitude is grounded in our observations). 
For this analysis, we hence chose a model of the simplex method that perturbs the right-hand side vector and which solves the input LP up to optimality and feasibility tolerances (whose magnitudes are grounded in our observations) using the shadow vertex pivot rule with an auxiliary objective function chosen uniformly at random.
To maintain mathematical tractability, we assumed the algorithm to be implemented using exact real valued arithmetic, Dantzig's ratio test, and the shadow vertex pivot rule.
Relevant instance-specific parameters are assumed to be bounded in magnitude (again, grounded in observation). Detailed descriptions can be found in \Cref{sub:introbythebookproperties} and \Cref{sec:bythebooksimplexfindings}.

In the third and final phase, we used our mathematical assumptions on the data and the algorithm to prove running time bounds using a mathematical analysis similar to those found in smoothed analysis, using the randomness in our perturbations and our auxiliary objective function.

As is evident from the specificity of some choices detailed above, different by-the-book analyses for the same algorithm are possible. They can vary in how the algorithm and the data are modeled, i.e., which assumptions and parameters are chosen, and they can vary in the types of mathematical techniques used to prove running time bounds.
We can distinguish at least two paths for improvements on a particular by-the-book analysis. The first occurs in the third phase and is one that we are used to in theoretical computer science: for the same set of mathematical assumptions, stronger mathematical analysis techniques can be developed, leading to stronger theorems.
The second is to formulate \emph{different} mathematical assumptions to model the same or different observations performed in the first phase. 

One is unlikely to produce a singular, holistic explanation for the performance of an algorithm, parameterizing by all possible factors and perfectly distilling the impact of these parameters into a bound. Theoretical analysis requires simplified models on which theoretical analyses can be performed.
The goal is not to produce a perfect model of reality, but instead to produce a mathematical model which is better informed by, and more reflective of, what we can observe.
As such, a by-the-book analysis can be improved in the first phase by making more or better observations, and second phases by better modelling.
Some axes on which such improvement can be evaluated include, but are not limited to, evaluability and verisimilitude.
It is preferable to express bounds in terms of parameters which can be effectively observed. 
It is also preferable to use mathematical assumptions which more closely resemble computational practice, such as to have a theoretical description of an algorithm which comports with the way in which that algorithm is implemented in code.

To exemplify this last point, we wield it against our present analysis in \Cref{sec:scorecard}.
%There we perform what we consider to be the final crucial step in any by-the-book analysis. 
We evaluate how well the chosen model, assumptions, and mathematical techniques comport with our observations, and evaluate how well our research process adhered to the framework of by-the-book analysis.
For the sake of good scientific practice, such a final evaluation phase is crucial. 
We recommend the evaluation process to be not only conducted as an isolated final last step. Instead, we recommend to evaluate the observations and decisions made as a steady part of the process throughout all phases. 

\subsection{Observations of Practical LP}\label{sub:introbythebookproperties}
Having described the framework of by-the-book analysis in general, we can now present a by-the-book analysis of the simplex method.  We begin here with an overview of our findings about the properties held by LPs in practice. As described in \Cref{sec:BTB_intro}, these findings will inform our mathematical assumptions.  As noted already, two main contributing factors to our findings are the software implementations of the simplex method and the observed characteristics of (well-formulated) inputs.
A given problem can be modeled in various ways, and a given LP has a wide array of equivalent formulations.  When LPs are solved in practice, they are typically modeled by domain experts who are constructing their formulations according to best practice recommendations explicitly detailed in software user manuals.

We review a number of software packages and describe our observations.
For the open-source solvers Glop and HiGHS, we report our observations from reading their respective source codes.
For the closed-source solvers Gurobi and MOSEK, we read their user manuals in detail and spoke extensively with developers.
We ran an experiment on the MIPLIB 2017 benchmark set to support our assumption that feasible sets of practical LPs have small mean width.
The features of solvers and data described here are mostly limited to those which are directly relevant to our models of scaling, tolerances, and bound/cost perturbations.
We first describe the general principles of what we observe, and in \Cref{sec:bythebooksimplexfindings}, we describe what observations we have with regard to specific software packages.
There are many other features we do not remark upon, although we believe that some can be potentially of interest to theoretical computer scientists and may lead to improved by-the-book analyses of the simplex method.

\paragraph{Condition number}
In an IEEE 754 double precision floating point number, there is 1 bit for the sign, 11 bits for the exponent, and 53 bits for the significand (of which 52 are explicitly stored).
This means that any such number can have a relative precision of at most $2^{-53} \approx 10^{-16}$.
Any arithmetic performed with numbers whose relative magnitudes differ by many orders of magnitude will yield imprecise results.
This effect is particularly visible when solving a system of linear equations.
Here, the backward error up to which it can be solved increases when the linear system has a high \emph{condition number}. 
The condition number $\kappa$ of a matrix $A$ is defined as $\kappa \coloneqq \norm{A}\norm{A^{-1}}$.
The relation to the backward error is the following. When solving the linear system $Ax = b$, a condition number $\kappa = 10^k$ indicates that one may loose up to $k$ digits of accuracy in $x$ from the accuracy in $b$.
Typically, in linear programming software, the condition number of linear systems is assumed to be no greater than $10^{12}$, see the discussion pertaining to Gurobi in \Cref{sec:bythebooksimplexfindings} for example. 

\paragraph{Tolerances}\label{par:tolerances}
With $53$ bits of accuracy and a condition number of order $10^{10}$, solving a linear system $A_B \vex = \veb_B$ can be done up to backward error around $10^{-6}$. That is, the software is able to find $\hat \vex$ such that $\norm{A_B^{-1}\veb_B - \hat \vex} \leq 10^{-6}\cdot \norm{A_B^{-1}\veb_B}$.
Because of this inevitable inaccuracy, every LP solver based on floating-point arithmetic includes a number of different tolerances and thresholds.
Here we discuss only the user-facing \emph{absolute feasibility tolerances}.
There are two of these, a primal feasibility tolerance and a dual feasibility tolerance, also called the \emph{optimality tolerance}.
These two numbers, both commonly of order $10^{-6}$ (e.g. the default value in Glop and Gurobi, as noted in \Cref{sec:bythebooksimplexfindings}), indicate how large the violations of primal or dual constraints are allowed to be for the solver to still report a solution as being `feasible and optimal'.
In contrast to the relative error described in the preceding paragraph, solver tolerances are absolute.

In this work we model the primal feasibility and complementary slackness conditions as follows. For more information on the deduction of this statement of complementary slackness, see \Cref{sec:algorithm}. There, we solve linear programs of the form 
\begin{align*}
    \operatorname{maximize} \quad & \vecc^\T \vex \\
    \operatorname{subject~to} \quad & A\vex \leq \veb \\
    &\vex \geq \veo,
\end{align*}
\noindent where we assume that every row of $A$ has norm $1$.
The solver may return (as optimal and feasible up to tolerances) any primal-dual solution pair $\vex^* \in \R^d, \vey^* \in \R^n_{\geq 0}$ which satisfies primal feasibility up to primal feasibility tolerance $\feastol>0$ and complementary slackness up to dual feasibility tolerance $\opttol>0$.  Formally, they satisfy the primal feasibility conditions
\begin{align}
    A\vex^* &\leq \veb + \feastol\cdot \ve1_n \label{eq:apx-feas}\\
    \vex &\geq -\feastol\cdot\ve1_d, \nonumber
\end{align}
and dual feasibility conditions
\begin{align*}
    A^\T \vey^* &\geq \vecc - \opttol\ve1_d \\
    \vey^* &\geq \veo.
\end{align*}
To ensure approximate optimality, we have a complementary slackness condition
\begin{align}
    \text{~for all~} i \in [n] &\text{,~if~} \vey^*_i > 0 \text{~then~} (A\vex^*)_i \geq \veb_i. \label{eq:apx-opt} \\
    \text{~for all~} j \in [d] &\text{,~if~} (A^\T \vey^*)_j > c_j + \opttol \text{~then~} \vex_j^* \leq 0. \nonumber
\end{align}

In our analysis we will take $\vey^* \geq \veo$, but solvers may assume a bit more flexibility and take $\vey^* \geq -\opttol\ve1$ due to their need to obtain $\vey^*$ as the solution to a system of linear equations.

\paragraph{Scaling}

A well-formed LP should be scaled appropriately: the input data, the output solution, and intermediate values should be not too large and not too small.
This is important to the functioning of real-world software.

The primary way in which this scaling requirement has been explained is through the effects of floating-point arithmetic.
A poorly scaled matrix is likely to have a high condition number.
The numerical concerns described above are an important reason for why proper scaling of the input is explicitly recommended.

The tolerances provide a second important anchor point for scaling.
When some of the optimal values are small relative to the tolerances, this can negatively affect the quality of the solution for its intended purpose.
When some of the solution values are too large relative to the tolerance, the solver may be unable to satisfy all feasibility and optimality conditions within the small (absolute) tolerance.

It is best practice to choose the measurement units for the rows and columns such that the non-zero values of variables across basic feasible solutions lie in some common, limited range of magnitudes.
Most LP solvers also include automatic scaling algorithms in case the user's scaling is deemed insufficient.
This practice speaks directly to the geometry of the feasible set.

Overall, the role of scaling for the performance of the simplex method is poorly understood \cite[p.~110]{marosbook}.

\paragraph{Perturbations}
When feasibility tolerances are permitted, solvers can be made faster by intentionally changing the right-hand side $\veb$, variable bounds, or the objective $\vecc$.
A number of mechanisms make use of this allowance.
We focus on one specific such mechanism, which is random perturbation before starting the simplex method.
Generally these perturbations are done in such a manner as to make the feasible region larger.
The perturbations are of a similar order of magnitude to the feasibility tolerance as discussed in \Cref{sec:bythebooksimplexfindings}. The primary use of a priori perturbation in practice is in perturbing the objective vector $\vecc$ before starting the dual simplex method.
During both the primal and dual simplex method, more perturbations and shifts can be performed at later iterations if needed. 
One consequence of bound perturbations is that degeneracy is avoided and consequently the simplex method cannot cycle. We furthermore analyze their implications on the running time of the simplex method.

As noted by \cite[p.~61]{qi-huangfu}, perturbations are ubiquitous among LP solvers but ``despite its wide application, the discussion of why perturbation works is relatively rare."
He highlights that in the survey by Bixby \cite{BixbySurvey} it is noted that more ``aggressive'' perturbations are experimentally viewed as more effective and that, in personal communication, it became clear there was no theoretical understanding as to why.
Our results provide the first theoretical explanation of this phenomenon.

\subsection{Overview}
\label{sec:model_overview}
The body of the paper is split into three parts.
In \Cref{sec:bythebooksimplexfindings} we describe in detail our observations from the user manuals of Gurobi and MOSEK and the source code of HiGHS and Glop, as well as the results of experiments on the LP relaxations of instances in MIPLIB 2017.
The observations described are related to the three properties described above: scaling, tolerances, and perturbations, as well as the condition number of the constraint matrix. It is common 
for a simplex method in state of the art implementations to add perturbations. 
We choose these perturbations proportional to the feasibility tolerances, mirroring their use in practice as found in \Cref{sec:bythebooksimplexfindings}.

After mathematical preliminaries in \Cref{sec:prelims}, in \Cref{sec:algorithm} we describe our full \emph{two-phase} simplex method whose running time we analyze, and connect the required parameters with the observations from \Cref{sec:bythebooksimplexfindings}.
Its analysis relies on \Cref{cor:easytouse}, which we devote 
\Cref{sec:mathysection} to proving. On the perturbed linear program, we assume our algorithm stays exactly feasible at all times without considering the tolerances, i.e., without accounting for floating point arithmetic errors.
These observations will yield an argument that certain parameters should be assumed to be bounded, show that bound perturbations can be part of a complete simplex method, and provide estimates for the order of magnitude that different quantities have in practice.
We end with a discussion of the results in \Cref{sec:discussion}.
In the remainder of this overview, we summarize \Cref{sec:algorithm} and \Cref{sec:mathysection}, which contain all proofs.

\paragraph{Two-Phase Simplex Method}
We describe our algorithm in detail in \Cref{sec:algorithm}.
In phase 1, roughly speaking, we start from the basis of non-negativity constraints
and add an auxiliary variable $0 \leq s \leq \norm{\veb}_\infty$ to make the corresponding point $\vex = \veo$ feasible as $A\vex + s\frac{\veb}{\norm{\veb}_\infty} \leq \veb$ with $\vex = \veo$ and $s=\norm{\veb}_\infty$.
We pivot to find a basic feasible solution with $s=0$, which yields a basic feasible solution to the original LP.
Then Phase II works as expected, pivoting to the optimal solution.
Assuming non-degeneracy, this strategy can be executed with the semi-random shadow vertex method with no unexpected complications: we move from an initial objective to a random objective on all the coordinates $(\vex,s)$, then we minimize $s$ to find a basic feasible solution which is optimal for a random objective on $\vex$, and finally we move from the random objective to the input objective.

The main challenge in our phase 1 is to make this strategy work when the right-hand side and bound vectors are perturbed: the basis of non-negativity constraints should remain feasible.
To achieve this, we slightly change the coefficients on $s$ to obtain
$A\vex + s\frac{\veb -\feastol\sqrt{\ln n}\ve1}{\norm{\veb}_\infty} \leq \veb$.
Now the point $(\veo, \norm{\veb}_\infty)$ has slack $\feastol\sqrt{\ln n}$ on each constraint.
Using Hoeffding's inequality, we can show that the solution $(\hat \veo - \frac{\feastol}{2}\ve1, \norm{\veb}_\infty)$ is feasible with good probability.
Note that this requires that solution to have mean $(\veo,\norm{\veb}_\infty)$, which we achieve by adding a coefficient on the auxiliary variables into the perturbed non-negativity bounds $\vex \geq \hat \veo + \frac{\feastol}{2\norm{\veb}_\infty}\ve1 s$. 
Perturbing the right-hand side $\veb$ to make $\hat \veb$ does not change this feasibility because $\hat\veb \geq \veb$.
The bounds $0 \leq s \leq \norm{\veb}_\infty$ remain unperturbed.
The analysis can accommodate for a single variable with fixed bounds, incurring a constant factor increase in the expected number of pivot steps.

We end \Cref{sec:algorithm} with a proof that the resulting solution is primal feasible up to primal feasibility tolerance, dual feasible up to the dual feasibility tolerance, and that it satisfies the complementary slackness up to optimality tolerance, indicating that the returned solution satisfies all requirements.

\paragraph{Running time bound}
\Cref{sec:mathysection} folds the variable bounds into the constraint matrix and considers linear programs of the form
\begin{align*}
    \operatorname{maximize} \quad & \vecc^\T \vex \\
    \operatorname{subject~to} \quad & A\vex \leq \veb,
\end{align*}
where again every row of the matrix $A$ is assumed to have norm equal to $1$.
Assuming that the right-hand side vector $\veb$ is perturbed with two-sided exponentially distributed perturbations with rate $\eta > 0$, the expected  number of pivot steps taken by the shadow vertex method to travel from the vertex maximizing a uniformly random objective vector $\vetheta \in \sfe$ to the vertex maximizing an objective $\vecc + \eps\vetheta$ is, by \Cref{cor:easytouse}, at most
\[O\left(d \sqrt{\frac{d \ln(n) M}{\eta} \ln\left(\frac{dN\ln (n)}{\eta \cdot \eps} \right)}\right),\]
where $M > 0$ indicates the expected mean width of the feasible set and $N$ is a proxy for the maximum absolute objective value.
The ratio $N/\eps$ is essentially scale-invariant in $\norm{\vecc}$, but our analysis does require the linear program to be ``well-scaled'' in two ways:
the feasible set must be scaled such that the mean width $M$ is bounded, and the constraint matrix must be scaled such that the rows of $A$ each have norm at most $1$.
These assumptions are to be expected.
Scaling assumptions are required for the feasibility tolerances to be well-defined, and are common for state of the art software implementations as well, as described in \Cref{sec:bythebooksimplexfindings}.
%There, we also explain a bound on $\kappa$ of $\kappa = 10^{12}$.
The perturbation size is chosen such that $\eta \cdot \ln(n) \approx 10^{-6}$ is the feasibility tolerance.

In order to learn the mean width of typical linear programming problems, we sampled the widths for LP relaxations of the MIPLIB 2017 benchmark set \cite{miplib} and for LP problems found in NETLIB \cite{gay1985electronic}.
For the MIPLIB instances, we found that the majority of LP relaxations has a feasible set with mean width smaller than $\sqrt{\frac{2d}{\pi}}$, the limit mean width of the unit cube.
On NETLIB, the mean widths were often larger, but only seldom above $10^6$.
We expect that if the columns of the NETLIB instances are rescaled properly, that the mean widths will be smaller.

All parameters which appear in our running time bound can be effectively estimated.
That is, after treating our chosen parameters (as identified during the first and second phases of by-the-book analysis) as having realistic and constant values, we obtain a polynomially bounded running time. 
Importantly, we have produced a \textit{theoretical} treatment of the simplex method's running time that is principally grounded in observable traits of real-world codes and problem instances.

We note that the explanatory power of this result strongly exceeds that of smoothed analysis, since the result presented here follows from modeling practical implementations of the simplex method applied to real-world data. In particular, the constraint matrix is not perturbed, so our analysis applies to sparse matrices, and all parameters involved are based on measures computable in practice. 
Our analysis is built to know ahead of time exactly what real-world properties our mathematical assumptions are modeling.  In contrast (as discussed in \Cref{sec:smoothed_intro}), there is not agreement on what, if anything, is actually modeled in smoothed analysis by perturbing the constraint matrix. We argue that the result presented here narrows the gap between what we can explain using theoretical frameworks and what is observed in practice.

\paragraph{Proof overview}

The proof of our main technical result \Cref{cor:easytouse} proceeds in three phases. First, in \Cref{sec:goodslacks} we describe the effect of the right-hand side perturbations.
We show that, conditional on a basis being feasible, the corresponding basic solution will have all its non-zero slack values on the inequalities be large. 
Specifically, the smallest non-zero slack will be at least $\frac{\eta}{1240 d \ln(n)}$ with probability at least $0.9$.
As in the previous paragraph, $\eta$ is a measure of the magnitude of the perturbations.
The lower bound on the slacks is derived using a Chernoff bound.
This argument is derived from a similar one used in smoothed analyses of the simplex method \cite{hlz,optimal_smoothed_analysis}.

Second, in \Cref{sub:largeredcosts} we show a dual analog to the above, based on the shadow vertex simplex path from a random auxiliary objective $\vez$ to a largely fixed objective $\eps^{-1} \vecc + \vez$.
We show that, conditional on a basis lying on this shadow path, with probability at least $0.9$ there exists an intermediate objective $t \vecc + \vez$ with $t \in [0,\eps^{-1}]$ for which this basis is optimized and for which the reduced costs on the tight constraints are large. Note that the existence of $t \in [0,\varepsilon^{-1}]$ for which the basis is optimal for $t \vecc + \vez$ is guaranteed by virtue of being on the shadow path. The novelty is that we can guarantee a choice of $t$ yielding large reduced costs on tight constraints.
In order to simplify this part of the analysis, we will throughout \Cref{sec:mathysection} assume that the random objective $\vez$ is sampled from an $L$-log-Lipschitz probability distribution.
This will not affect the final conclusion.

Finally, in \Cref{sec:shadowpath} we show our upper bound on the length of the semi-random shadow vertex path.
As a consequence of our results in \Cref{sec:goodslacks} and \Cref{sub:largeredcosts}, we show that a large fraction of the feasible bases on this path must have both (in the above explained senses) large non-zero slacks and admit an intermediate objective for which all $d$ reduced costs are large.
Afterwards, we argue that whenever two subsequent bases on this path both have both properties, then the pivot step leading from the first to the second of them must make large progress.
This progress is measured through four different potential functions.
Two of these potential functions measure the progress of the intermediate objective vector $t \vecc + \vez$ on the line segment from $\vez$ to $\eps^{-1}\vecc + \vez$.
The third and fourth potentials encode progress in the objective value and auxiliary objective value respectively.
A majority of pivot steps will consume a certain amount of at least one of these four potential functions.
The total amount of potential available is bounded, which leads to our upper bound on the expected number of pivot steps.

\subsection{Related Work}
The simplex method has been studied in the worst case complexity framework under various additional assumptions.
One line of work has assumed that the constraint matrix $A \in \Z^{n \times d}$ is integer and has every square submatrix have a determinant which is bounded from above in absolute value.
Under this assumption, upper bounds on the running time of the simplex method have been proven \cite{Eisenbrand2016, conf/icalp/BR13, dadushhahnle, dyerfrieze1994, Bonifas2012}. The semi-random shadow vertex pivot rule of \cite{dadushhahnle} strongly informed our own.
%This extends earlier work on the combinatorial diameter of polyhedra with bounded subdeterminants \cite{dyerfrieze1994, Bonifas2012}.
Approximating the maximum subdeterminant is known to be NP-hard \cite{Summa2014}.
Similar results are available for a closely related condition measure, the circuit imbalance \cite{farbod-bento-laci-survey}.
A study estimating the circuit imbalance on the instances in MIPLIB 2017 \cite{jakub-martin-experimental} found it to require much higher numerical precision than is available with IEEE 754 double precision floating point arithmetic.

A parallel line of work investigates the worst-case performance of the simplex method when the slack values and reduced costs are bounded away from zero \cite{kitaharamizuno2011primal,kitaharamizuno2012dual,tanomiyashirokitahara2019steepest}.
These papers extend an earlier result about the policy iteration algorithm for Markov Decision Processes with bounded discount rate \cite{Ye2005}, and are similar to a calculation done by Dantzig \cite{Dantzig1990origins}.
The condition number underlying this line of work is NP-hard to compute as shown in \cite{kunosanotsuruda2018computing}.
The same paper also estimated this condition number on the LP instances in the NETLIB test problem set \cite{gay1985electronic}, though the measurement approach they used severely limits what conclusions can be drawn based on these estimates.
The use of maximal and minimal non-zero slacks was a direct inspiration for our analysis, although we have the minimal non-zero slack be probabilistically bounded using perturbations instead of deterministically by assuming structure in the input data.
 
In our by-the-book analysis, we use a simplex method with bound perturbations to obtain provable running time guarantees.
Bound perturbations are often incorporated in codes with the stated purpose of preventing issues of degeneracy and stalling \cite{marosbook, koberstein}.
Other methods to prevent stalling are possible, including the use of expanding tolerances or minimum step sizes \cite{Gill1989}.
In a theoretical regime where bounds may not be changed, other approaches are available \cite{Bland1977, murtydegeneracy, Kukharenko2024}.

A technique akin to, but developed independently of, bound perturbations was previously used in \cite{kelnerspielman} to find a weakly polynomial-time algorithm for linear programming. Although aspects of what they do are similar to what we do here, their arguments adapted to our setting would produce exponentially worse running time dependencies on the parameters we consider.
%the arguments in their proof were too weak for our by-the-book analysis.
%For their algorithm, they construct a series of auxiliary linear programs for which boundedness must be proven.
%Repeatedly their algorithm finds feasible solutions of large norm using a variant of the semi-random shadow vertex simplex method, which is then used to rescale the auxiliary linear program.
%A volume decrease argument, similar to the analysis of the ellipsoid method, eventually allows them to prove whether the auxiliary linear program is bounded or unbounded.
%That capability is then extended to solving general LP in weakly polynomial time.
%While sufficient for their own purposes, 

\section{Observations}\label{sec:bythebooksimplexfindings}
In the following we discuss our findings in the several open source codes and user manuals which provide evidence and specificity in support of our claims in  \Cref{sub:introbythebookproperties}. The order of the solvers is alphabetical.

\paragraph{Glop}
As part of Google's software suite OR-Tools \cite{ortools}, Glop is an open-source LP solver. For any finite values in the input LP, Glop by default\footnote{\url{https://github.com/google/or-tools/blob/9b770/ortools/glop/parameters.proto\#L481}} allows numbers with absolute values as low as $10^{-30}$ and as high as $10^{30}$.
However, in order to increase sparsity, by default any floating-point numbers with absolute value below $10^{-14}$ are set to zero\footnote{\url{https://github.com/google/or-tools/blob/9b770/ortools/glop/parameters.proto\#L183}}. Thus, some restrictions on the scaling of the LP are inherent in the solver.

Glop has configurable primal and dual feasibility tolerances, which have a value of $10^{-6}$ by default\footnote{\url{https://github.com/google/or-tools/blob/9b770/ortools/glop/parameters.proto\#L251}}.
In the dual simplex method, Glop can perturb costs before starting.
When this is done, the default perturbation on the $i$'th element of the cost vector $\vecc$ is given by\footnote{\url{https://github.com/google/or-tools/blob/9b770/ortools/glop/reduced_costs.cc\#L255}}
\begin{align*}
r_i \cdot (10^{-5} \vecc_i + 10^{-7} \norm{\vecc}_\infty),
\end{align*}
where $r_i$ is sampled uniformly from the interval $[1,2]$. This cost perturbation
is turned off by default\footnote{\url{https://github.com/google/or-tools/blob/9b770/ortools/glop/parameters.proto\#L415}}.

\paragraph{Gurobi}
The Gurobi user manual recommends that any right-hand side coefficients are scaled such that their absolute value is $10^4$ or less \cite[Section 7.3.4]{gurobi}.
The same recommendation is given with regards to the objective coefficients and variable bounds.
Similarly, they recommend that every LP has an optimal value less than $10^4$.
The latest version of Gurobi includes 3 different automatic scaling modes, but the documentation does not specify what algorithms are used or when they are used.

In Gurobi, any constraint coefficient whose absolute value is smaller than $10^{-13}$ is treated as zero \cite[Section 7.3.6]{gurobi}.
The user manual recommends that the user scales the input problem such that all non-zero matrix coefficients have their absolute value between $10^{-3}$ and $10^{6}$, and that they span no more than $6$ orders of magnitude.
The default big-$M$ value (essentially, the finite number that stands in for `infinity' in ILP formulations) is $10^6$, further indicating that coefficients larger than this are undesirable.

The primal feasibility tolerance defaults to $10^{-6}$ but is configurable to any value between $10^{-2}$ and $10^{-9}$. The same holds for the optimality tolerance.
The primal tolerance is described as follows \cite[Section 7.3]{gurobi}: ``$\vea \cdot \vex \leq b$ will be considered to hold if $(\vea \ast \vex) - b \leq$ FeasibilityTol.''  Here, `$\cdot$' is the exact inner product, whereas `$\ast$' is the computed product (subject to errors).

The Gurobi user manual further indicates that a condition number of $10^{12}$ for a basic submatrix is considered large \cite[Section 7.4]{gurobi}.
The Gurobi solver warns the user of large bounds when a variable has a bound with absolute value that is $10^{10}$ or larger. It warns the user of large objective coefficients when there is a coefficient with absolute value $10^{10}$ or larger. 

Conversation with the developers of Gurobi confirmed that perturbations are actively used in all their simplex codes.
Initial perturbations have magnitude proportional to the feasibility tolerance and perturbation size may grow when the simplex method makes insufficient progress.

\paragraph{HiGHS}
In the open-source MIP solver HiGHS \cite{highs}, the simplex method incorporates two scaling methods (one using equilibration and one based on maximum value\footnote{\url{https://ergo-code.github.io/HiGHS/dev/options/definitions/\#option-simplex-scale-strategy}}). If all non-zero elements of the constraint matrix are between $0.2$ and $5.0$ then no automatic scaling will take place\footnote{\url{https://github.com/ERGO-Code/HiGHS/blob/c67f06df422ce1faff4089033513874b8/highs/lp_data/HighsLpUtils.cpp\#L861}}. Notably, this automatic scaling is not performed for the interior-point method.
%This automatic scaling is not performed for the interior-point method.
%Notably, if all non-zero elements of the constraint matrix are between $0.2$ and $5.0$ then no automatic scaling will take place\footnote{\url{https://github.com/ERGO-Code/HiGHS/blob/c67f06df422ce1faff4089033513874b8/src/lp_data/HighsLpUtils.cpp\#L861}}.

HiGHS has configurable parameters for the smallest non-zero and largest finite entries that may be present in the constraint matrix\footnote{\url{https://ergo-code.github.io/HiGHS/dev/options/definitions/\#option-small-matrix-value}}.
By default, any matrix entry whose absolute value is smaller than $10^{-9}$ is treated as $0$,
and any matrix entry whose absolute value is larger than $10^{15}$ is treated as infinity.

There are configurable primal and dual feasibility tolerances which equal $10^{-7}$ by default and which can be tuned to be as low as $10^{-10}$.
In the primal simplex method, when a bound is found to be infeasible, it is shifted to make it feasible by at least the feasibility tolerance, plus a uniformly random quantity between $0$ and the feasibility tolerance\footnote{\url{https://github.com/ERGO-Code/HiGHS/blob/c67f06df422ce1faff4089033513874b8/highs/simplex/HEkkPrimal.cpp\#L2793}}.
In the dual simplex method, the solver perturbs the cost vector already during initialization\footnote{\url{https://github.com/ERGO-Code/HiGHS/blob/c67f06df422ce1faff4089033513874b8/highs/simplex/HEkkDual.cpp\#L126}}.
These perturbations are uniformly distributed in an interval.

The PhD thesis of Qi Huangfu \cite{qi-huangfu} documents the development of the HiGHS simplex method solver. In his work, he suggests that feasibility and optimality  tolerances are typically on the order of $10^{-6}$ or $10^{-7}$. For HiGHS he describes\footnote{See also \url{https://github.com/ERGO-Code/HiGHS/blob/c67f06df422ce1faff4089/highs/simplex/HEkk.cpp\#L2488}} an explicit choice of value for cost perturbation of
\[(10^{-5}|\vecc_{j}| + 10^{-7}\norm{\vecc}_\infty +10 \cdot \opttol) (r+1), \]
where $\opttol > 0$ is the dual feasibility tolerance and $r$ is a uniformly random number in $(0,1]$. In particular, the above cost perturbation incorporates a term whose magnitude is proportional to the objective and a term with absolute magnitude. In our analysis, we model these terms as being the same by assuming our LP is sufficiently well-scaled.

The HiGHS documentation recommends users to scale their primal variables such that the optimal solution has its entries be of order unity.
The documentation further states that ``this typically occurs if all objective and constraint matrix coefficients, as well as finite variables and constraint bounds, are of order unity.''

\paragraph{MOSEK}

The MOSEK instruction manual \cite{mosekmanual} mentions that problems containing large or small coefficients, outside the range $[10^{-7}, 10^9]$, are often hard to solve\footnote{\url{https://docs.mosek.com/latest/pythonapi/presolver.html\#index-5}}.
Although it includes a method to automatically scale variables, no specifics are reported about how this scaling is performed or under what circumstances.
The solver will print an error\footnote{\url{https://docs.mosek.com/latest/pythonapi/parameters.html\#mosek.dparam.data_tol_aij_large}} to warn the user if any matrix element is larger than $10^{10}$, if a cost entry is larger than $10^8$ or if a bound is larger than $10^8$.
If an upper and a lower bound for a variable differ by less than $10^{-8}$, then MOSEK will consider the two numbers to be equal\footnote{\url{https://docs.mosek.com/latest/pythonapi/parameters.html\#mosek.dparam.data_tol_x}}. The MOSEK Modeling Cookbook \cite[Section 7.3]{mosekmodeling} considers ``a model to be badly scaled if variables are measured in very different scales, or constraints or bounds are measured on very different scales.''

The absolute tolerance\footnote{\url{https://docs.mosek.com/latest/pythonapi/parameters.html\#mosek.dparam.basis_tol_x}} for the primal and dual infeasibility both default to $10^{-6}$.
A technical report \cite{mosekfarkas} mentions that, when troubleshooting infeasible problems, the most important constraints  are those whose relative values in the Farkas certificate are at least $10^{-8}$.

Interviews with the developers of MOSEK confirm that perturbations are actively used in all their simplex codes, and that adding larger perturbations are an effective tool when the simplex method stalls.
Initial perturbations have magnitude proportional to the feasibility tolerance.

\paragraph{Experiments on MIPLIB 2017}\label{par:MIPLIB}
Our main technical result \Cref{cor:easytouse} is stated in terms of the \emph{mean width} of the feasible set of the linear program (after bounds are perturbed).
The mean width of a polyhedron $P$ is defined as the expected value of
\begin{align*}
    \operatorname{maximize} \quad & \vetheta^\T (\vex - \vex')\\
    \operatorname{subject~to} \quad & \vex, \vex' \in P,
\end{align*}
where $\vetheta \in \sfe$ is uniformly random distributed on the unit sphere.

In order to empirically ground this quantity, we performed experiments to measure it on two benchmark libraries: the NETLIB library of LPs \cite{gay1985electronic} and the LP relaxations of the MIPLIB 2017 benchmark set of Mixed Integer Linear Programs \cite{miplib}.
The experiment is rudimentary in design, foregoing any considerations of automatic scaling or presolve.

For every instance, we sampled the width in 100 random directions as follows.
Due to Milman's concentration of measure phenomenon (see, e.g., \cite{ball}), the width is quite sharply concentrated around the mean.

We sampled an objective vector $\vez$ with entries sampled independently at random from a Gaussian distribution with mean $0$ and standard deviation $1$.
We loaded the \texttt{MPS} files using \texttt{Gurobi 13.0.1} and used it to obtain the LP relaxation.
Over this LP feasible set we maximized and minimized the objective $\vez^\T \vex$.
The difference between the two values was divided by $\norm{\vez}$ to obtain the width in direction $\frac{\vez}{\norm{\vez}}$, which is a uniformly random unit direction.
The code for this experiment can be found at \url{https://github.com/sophiehuiberts/miplib-meanwidth}, and it was run on an Intel i5-14600KF CPU.

Besides \texttt{neos-5114902-kasavu}, every individual minimization and maximization completed within the deterministic work limit of 300.
Any instance finding an unbounded or infeasible status code (\texttt{INFEASIBLE}, \texttt{INF\_OR\_UNBD} or \texttt{UNBOUNDED}) was dropped.
On NETLIB, 58 out of 114 instances remained after this phase.
On MIPLIB, 186 out of 240 instances remained.

All sample means obtained were between $10^{-1}$ and $10^{9}$.
In MIPLIB, the sample means were relatively small. 46 instances have sample means in $[10^{-1},10^1]$, 65 have sample means in $[10^1,10^2]$, 54 in $[10^2,10^4]$, 10 in $[10^4,10^6]$ and the remaining two instances with sample means between $10^6$ and $10^9$.
The mean widths on NETLIB were larger on average, with 7 instances each having sample mean widths in $[10^{-1},10^1]$ and $[10^1,10^2]$, 26 instances in $[10^2,10^4]$, 12 in $[10^4,10^6]$ and the remaining 6 in $[10^6,10^9]$.

The full results are drawn in \Cref{fig:meanwidth}, with the sample mean width on the vertical axis and the total variable count of the instance on the horizontal axis.
Besides the results for NETLIB and MIPLIB benchmark instances, the figure includes sample mean widths for unit cubes $[0,1]^d$ in various dimensions.
The mean width of the unit cube converges to $\sqrt{\frac{2d}{\pi}}$, and most MIPLIB instances are below that line.

\begin{figure}[!ht]
\centering
\includegraphics[width=0.85\textwidth]{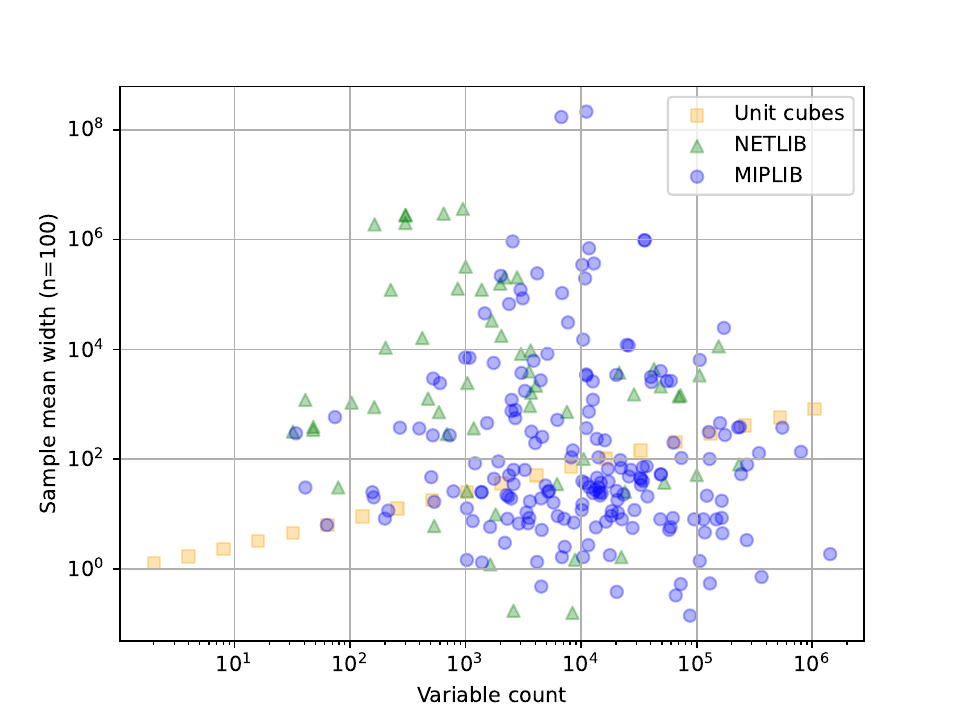}
\caption{Measurement results.}
\label{fig:meanwidth}
\end{figure}

We see that many instances have sample mean widths that are small.
We take these results as preliminary evidence that the mean width of the feasible set is bounded from above for well-scaled linear programs.

Although many instances have small mean widths, there is a sizable fraction with large mean width.
This is especially pronounced for the NETLIB set.
We suspect that the instances the instances with high mean width are poorly scaled.
More work is needed to verify whether indeed the better scaled instances tend to have lower mean widths.

In production software, the simplex method is run on linear programs after they have been presolved and automatically rescaled.
In our experiment we measured the mean width of the feasible set as given by the user, instead of the mean width of the presolved and scaled feasible set.
This may affect the observed results.
An in-depth computational study of the interaction of mean width and presolve/scaling is justified.

\section{Preliminaries}\label{sec:prelims}

We write $[n] = \{1,\dots,n\}$ and $\binom{[n]}{d} = \{B \subseteq [n] : \abs{B}=d\}$.
We take $[0]=\emptyset$ to be the empty set.
The $d \times d$ identity matrix is denoted $I_{d\times d}$.
Given a matrix $A \in \R^{n \times d}$ and a set $S \subseteq [n]$ we write $A_S$ to indicate the submatrix consisting of rows indicated by indexes in $S$, and its associated map $A_S : \R^d \to \R^S$.
A set $B \in \binom{[n]}{d}$ is called a \emph{basis} if the square matrix $A_B$ has full rank.

Vectors are bolded. The all-ones vector is $\ve1_d$, where the dimension of the vector may be omitted and denoted as just $\ve1$ if clear from context. We denote the $i$'th vector in the standard coordinate basis is $\vece^{i}$.
For a vector $\veb \in \R^n$ we write $\veb_S \in \R^S$ to indicate the restriction of $b$ to the entries indexed in $S$.
When $A, \veb$ are clear from context, we define the basic solution for a basis $B$ as $\vex^B = A_B^{-1} \veb_B$.
Scalars are denoted by lowercase letters, vectors by bold lowercase letters, and sets and matrices by uppercase letters.
For two vectors $\vev, \vev' \in \R^d$ we denote by $[\vev,\vev'] \subset \R^d$ the straight line segment connecting them.
The unit sphere is denoted by $\sfe = \{\vex \in \R^d : \norm{\vex}=1\}$.

\subsection{The Shadow Vertex Simplex Method}\label{sec:shadow_vertex_magic}

Our analysis makes use of a particular simplex method called the shadow vertex simplex method originally introduced by Gass and Saaty in \cite{gas55} in the context of parametric linear programming. A textbook introduction that proves the standard properties we rely on may be found in \cite{bwcachapter}. 

The shadow vertex simplex method works as follows. 
We fix a linear program
\begin{align*}
    \operatorname{maximize} \quad & \vecc^\T \vex \\
    \operatorname{subject~to} \quad & A\vex \leq \veb,
\end{align*} with $ A \in \R^{n \times d}$ and $ \veb \in \R^n$ 
as input and start with an initial feasible basis $B$. We obtain an auxiliary objective $\vez \in \mathbb{R}^{d}$ as follows. Let $\vez$ be generically chosen such that $\vez A_{B}^{-1} > \veo$ holds. 
If $\vecc A_{B}^{-1} \geq \veo$, already the current basis is optimal and we are done. Otherwise, there exists a coordinate $i$ such that $(\vecc A_{B}^{-1})_{i} < 0$. Over all those indices $i$ for which $(\vecc A_{B}^{-1})_{i} < 0$, we apply the following ratio test. We define 
\[\lambda = \min_i\left(- \frac{(\vez A_{B}^{-1})_{i}}{(\vecc A_{B}^{-1})_{i}}: (\vecc A_{B}^{-1})_{i} < 0 \right),\]
and let $p$ be the argument minimizing the same function. Then, by the minimality of $\lambda$ and positivity of $\vez A_{B}^{-1}$, $((\vez + \lambda \vecc)A_{B}^{-1})_{i} > 0$ for all $i \neq p$ and $((\vez + \lambda \vecc)A_{B}^{-1})_{p} = 0$. By non-degeneracy, there is a unique basis $B^{1} \neq B$ such that $B^{1} \supseteq B \setminus \{p\}$.   We let $B^{1}$ be the new basis and $\vez+ \lambda\vecc$ as the new auxiliary objective and continue the process from $B^{1}$ following a sequence of adjacent bases $B^{1}, B^{2},\dots, B^{k}$ until reaching an optimal basis. We note that the genericity of $\vez$ ensures that for each basis $B^{i}$ on the path, there exists $t > 0$ such that $(\vez+t \vecc)A_{B^{i}}^{-1} > 0$. The existence of such a $t$ is what allows the loop to continue even though $(\vez + \lambda \vecc)A_{B^{1}}^{-1}$ is only positive in the coordinates of $B^{1} \cap B$ and $0$ otherwise.  

There are two standard geometric interpretations of the shadow vertex simplex method. The first is in terms of a two dimensional linear projection. Let  $\vev = A_{B}^{-1}\veb_{B}$ be the vertex for the initial basis $B$. Then $\vez$ is a generic vector optimized uniquely at $\vev$. Then orthogonally projecting onto the $(\vez,\vecc)$-plane gives a polygon. By construction, $\vev$ and the $\vecc$-maximizer are both projected onto the boundary of that polygon and are the vertices maximizing the first and second coordinate respectively. A path in the polygon between the projection of $\vev$ and the projection of the $\vecc$-maximizer corresponds uniquely to a path on the polytope of the same length. This path is precisely the path chosen by the shadow vertex simplex method with auxiliary vector $\vez$. 

The second geometric interpretation relies on a dual perspective on a polytope. Given any feasible basis $B$, the cone $\R^{d}_{\geq 0} A_{B}$ is called the normal cone of $B$ or equivalently the normal cone of the vertex $A_B^{-1}\veb_{B}$. It is precisely the cone of all objectives optimized at that vertex. Then the shadow vertex simplex method functions by choosing a generic vector $\vez$ in the normal cone of the initial vertex $\vev$. Then the path traversed consists of all vertices with normal cones that intersect the line segment $[\vez, \vecc]$. Equivalently, that is precisely the set of normal cones of feasible bases that intersect the ray $\{\vez +\lambda \vecc: \lambda \geq 0\}$. We call the points on $[\vez, \vecc]$ or the ray $\vez + \lambda \vecc$ intermediate objectives.

When a linear program has no feasible solutions, any shadow path on it will be considered to have length $0$.

These geometric interpretations are useful for understanding the algorithm and explaining its efficiency in average case, smoothed analysis, and in the by-the-book analysis we undertake here.

\section{A Simplex Method with Bound Perturbations}\label{sec:algorithm}

In this section, we overview and describe how to form a complete two-phase simplex method with running time that can be analyzed using \Cref{cor:easytouse}.
For notational simplicity we assume that we are given a linear program of the form
\begin{align}
    \operatorname{maximize} \quad & \vecc^\T \vex \tag{input-LP} \label{eq:input-LP}\\
    \operatorname{subject~to} \quad &  A\vex \leq \veb \nonumber \\
    & \vex \geq \veo \nonumber
\end{align}
with $A \in \R^{n \times d}$. We assume that every row of $A$ has $\ell_2$ norm at most $1$.
Any LP can easily be transformed to this format.
This specific form of LP will aid in a simpler presentation of our Phase I method.
We will solve this LP \emph{up to tolerances} $\feastol > 0$ and $\opttol > 0$.
See \eqref{eq:apx-feas} and \eqref{eq:apx-opt} for definitions of solving up to tolerances.

To start, we perturb our bound vector $\veo$ and the right-hand side vector $\veb$.
Unlike the uniform distribution used in solvers, we perturb according to the following distribution:
\begin{definition}\label{def:ourperturbations}
    Let $\eta, \gamma \in \mathbb{R}_{>0}$, and let $\vev \in \mathbb{R}^{k}$. Then define for each $i \in [k]$, 
    \[f_{i}(t) = \frac{1}{2\eta} e^{-\abs{t - \vev_i - \gamma\eta}/\eta}.\] 
    Define a random vector $\hat \vev$ to be $(\vev, \eta,\gamma)$-exponentially distributed if its entries $\hat \vev_{i}$ are independently distributed with each entry's probability density function given by $f_{i}$.
\end{definition}
We sample $(-\hat \veo, \hat \veb) \in \R^{d+n}$ to be $\left((\veo,\veb), \frac{\feastol}{4\ln(d+n)}, 2\ln(d+n)\right)$-exponentially distributed.
Strong tail bounds will ensure that these perturbed bounds are suitable:
\begin{lemma}\label{lem:perturbationsbounded}
    Let $\eta, \gamma \in \mathbb{R}_{>0}$, and let $\vev \in \R^k$. Let $\hat \vev$ be a $(\vev, \eta,\gamma)$-exponentially distributed random vector. Then 
    \[\Pr\big[\hat \vev_{i} \in [\vev_{i},\vev_{i} + 2\gamma \eta] \text{ for all } i \in[k]\big] \geq 1 - ke^{-\gamma}. \]
\end{lemma}
\begin{proof}
    By applying a union bound we know that
    \[
    1 - \Pr\big[\hat \vev_{i} \in [\vev_{i},\vev_{i} + 2\gamma \eta] \text{ for all } i \in [k]\big] \leq \sum_{i = 1}^{k} \Pr[ \hat \vev_i < \vev_i] + \Pr[ \hat \vev_i > \vev_i+2\gamma\eta].
    \]
    Note that by construction the latter probabilities are symmetrically of the same size, i.e., $\Pr[\hat \vev_i < \vev_i] = \Pr[\hat \vev_i > \vev_i+2\gamma\eta]$
    for every $i \in [k]$.
    Writing down their probability densities, we find
    \begin{align*}
        \Pr[\hat \vev_i < \vev_i] &= \frac{1}{2\eta} \int_{-\infty}^{\vev_i} e^{-\abs{t-\vev_i - \gamma\eta}/\eta}\dd t \\
         &= \frac{1}{2\eta} \int_{-\infty}^{-\gamma\eta} e^{-\abs{t}/\eta}\dd t \\
         &= \frac{1}{2} \int_{-\infty}^{-\gamma} e^{-\abs{t}}\dd t = \frac{1}{2}e^{-\gamma}.
    \end{align*}

    Summing over all $i \in [k]$,
    we see that $1 - \Pr\big[\hat \vev_{i} \in [\vev_{i},\vev_{i} + 2\gamma \eta] \text{ for all } i \in [k]\big] \leq ke^{-\gamma}$, and hence $\Pr\big[\hat \vev_{i} \in [\vev_{i},\vev_{i} + 2\gamma \eta ]\text{ for all } i\in [k]\big] \geq 1 - ke^{-\gamma}.$
\end{proof}
 By plugging in our chosen parameter values, it follows from \Cref{lem:perturbationsbounded} that with probability greater than $1-\frac{1}{d+n}$ we have
\begin{align}
    -\feastol\cdot \ve1_d &\leq \hat \veo \leq \veo \label{eq:perturbationswithintols}\\
    \veb &\leq \hat \veb \leq \veb + \feastol\cdot \ve1_n. \nonumber
\end{align}
In order to ensure that the above inequalities \eqref{eq:perturbationswithintols} are guaranteed to hold for our algorithm's output, we reject any vector $(\hat \veo, \hat \veb)$ which violates these conditions and run the entire algorithm anew with fresh samples $(\hat \veo, \hat \veb)$.
The expected number of samples needed is $(1-\frac{1}{n+d})^{-1}$. Wald's equation tells us that the expected number of pivot steps when this rejection sampling is applied grows by at most a factor $(1-\frac{1}{n+d})^{-1}$ compared to when rejection sampling is inactive.

The simplex method described will, in Phase II, visit vertices of the perturbed LP
\begin{align}
    \operatorname{maximize} \quad & \vecc^\T \vex \tag{perturbed-LP} \label{eq:perturbed-LP}\\
    \operatorname{subject~to} \quad &  A\vex \leq \hat \veb \nonumber \\
    & \vex \geq \hat \veo \nonumber
\end{align}

\paragraph{Sampling Perturbations}
The space of possible perturbation distributions has not been explored systematically.
At the moment of writing it is not clear what options are available when designing a perturbation distribution, nor is it clear how to evaluate the quality of any such distribution.
Nevertheless, we describe here how to sample from the distribution used in this paper as an aid for readers looking to try the above perturbation distribution in their own simplex code.

The general sampling mechanism has four parameters: an offset, a lower bound for rejection sampling, an upper bound for rejection sampling, and a decay rate.
The probability density is then the unique function that has its maximum at the offset, which continuously decays at the prescribed rate when going up or down from the offset to either bound, and which is zero outside the bounds.
Code to sample is given in \Cref{alg:sampling}.

% I don't know how to typeset an algorithm in a way thats pretty.
\begin{algorithm}
\caption{General sampling method for perturbations}
\label{alg:sampling}
\begin{algorithmic}
\Require \texttt{offset, lowerReject, upperReject, logLipschitz} $\in \R$
\Repeat
\State \texttt{sign} $\gets \operatorname{Bernoulli}(1/2)$
\State \texttt{uni} $\gets \operatorname{Uniform}[0,1]$
\State \texttt{exponential} $\gets -\ln(\texttt{uni})$
\If{\texttt{sign == 1}}
    \State \texttt{sample $\gets$ offset + logLipschitz $\cdot$ exponential}
\Else
    \State \texttt{sample $\gets$ offset - logLipschitz $\cdot$ exponential}
\EndIf
\Until{\texttt{lowerReject < sample} and \texttt{sample < upperReject}}
\State return \texttt{sample}
\end{algorithmic}
\end{algorithm}

In the main mathematical analysis of \Cref{sec:mathysection}, we assume that no rejection sampling is performed, i.e., that \texttt{lowerReject}$ = -\infty$ and \texttt{upperReject}$ = \infty$.
In order to bound the expected running time with finite rejection bounds, we choose \texttt{lowerReject}$ = $\texttt{offset}$ - 2\ln(d+n)$\texttt{logLipschitz} and \texttt{upperReject}$ = $\texttt{offset}$ + 2\ln(d+n)$\texttt{logLipschitz}.
Using \Cref{lem:perturbationsbounded}, we see that the total probability of any sample out of $n+d$ being rejected is less than $1-\frac{1}{n+d}$.

In our analysis we assume that we can solve systems of linear equations exactly, and we search for a solution that is feasible up to the feasibility tolerance.
As such, for our analysis we choose $\texttt{lowerReject} = 0$ and $\texttt{upperReject} = \feastol$ as per \eqref{eq:perturbationswithintols}.
Solving for the other parameters gives $\texttt{offset} = \feastol/2$ and $\texttt{logLipschitz} = \frac{\feastol}{4\ln(n+d)}$.

\subsection{Initial Basis}
Traditionally, Phase I LPs have a trivial basic feasible solution from which to start.  However, our tools only allow us to move from the optimum of a random objective to the optimum of a fixed objective, or from the optimum of a fixed objective to the optimum of a random objective.  In order to conclude Phase I, we need to ensure that the introduced artificial variable becomes zero, which is to say, our Phase I procedure concludes at the optimum of a fixed objective.  As such, in order for us to bound the time necessary to solve the Phase I LP, our Phase I procedure must start from the optimum of a random objective.  In this section, we show how to do this by solving the LP (\ref{eq:phase0}) below.

Let $\vetheta \in \sfe$ be uniformly distributed. Consider the linear program
\begin{align}
    \operatorname{maximize} \quad &\vetheta^\T \vex \nonumber \\
    \operatorname{subject~to} \quad & A \vex \leq \feastol\sqrt{\ln n}\ve1 + (\hat \veb - \veb) \nonumber \\
    & \vex \geq \hat \veo + \frac{\feastol}{2}\ve1. \tag{Initial Basis} \label{eq:phase0}
\end{align}
This linear program will be solved in order to find the initial basis from which Phase I will be started.
\begin{lemma}\label{lem:phase0feas}
    The point $\vex = \hat \veo + (\feastol/2)\ve1$ is a basic solution to \eqref{eq:phase0}.
    Conditional on \eqref{eq:perturbationswithintols}, with probability at least $1-1/n$ this basic solution is feasible.
\end{lemma}
\begin{proof}
    We can observe that $\vex \geq \hat \veo + \frac{\feastol}{2}\ve1$ is part of the system defining \eqref{eq:phaseI}.
    The point satisfies these at equality, and they are clearly linearly independent.
    Thus, the point is a basic solution.
    
    The point is feasible if and only if
    $A(\hat \veo + \frac{\feastol}{2}\ve1) \leq (\hat \veb - \veb) + (\feastol\sqrt{\ln n})\ve1$.
    Since we know $\hat \veb$ satisfies $\hat \veb \geq \veb$, it suffices to show that
    $A(\hat \veo + \frac{\feastol}{2}\ve1) \leq (\feastol\sqrt{\ln n})\ve1$ holds with probability at least $1- 1/n$.

    Let $f$ denote the probability density function of $\hat \veo + \frac{\feastol}{2}\ve1$. Recall that we chose parameters for our distribution function $\gamma = \frac{\feastol}{4 \ln(d+n)}$ and $\eta = 2 \ln(d+n)$. By definition, the distribution function we chose for $\hat \veo$ is the Laplace distribution centered at $\veo - \eta \gamma$ and so
    \[\E[\hat\veo + \frac{\feastol}{2}\ve1] = \veo - \eta \gamma \ve1 + \frac{\feastol}{2}\ve1 =-\frac{\feastol}{4 \ln(d+n)}(2 \ln(d+n))\ve1 +   \frac{\feastol}{2}\ve1  = \ve0. \]
    Then, since linear transformations commute with expectations, $\E[A(\hat\veo + \frac{\feastol}{2}\ve1)] = \veo$.

    Our chosen distribution of the right-hand side perturbations is symmetric around the mean.
    In particular, the vector $\hat\veo + \frac{\feastol}{2}\ve1$
    follows a $(\veo, \eta,\gamma)$ distribution.
    Even after conditioning on \eqref{eq:perturbationswithintols}, its distribution is invariant under arbitrary sign flips.
    Hence the probability density function $f$ is invariant under flipping any signs:
    for every $\vey \in \R^d$ we have \[f(\vey_1,\dots,\vey_d) = f(\abs{\vey_1},\dots,\abs{\vey_d}) = f(\pm\vey_1,\dots,\pm\vey_d).\]
    Thus we get strong concentration according to Hoeffding's inequality.
    
    Let $\ve\eps \in \{+1,-1\}^d$ be sampled uniformly at random.
    For any $i \in [n]$ we find
    \begin{align*}
        \Pr_{\hat\veo}[A_i^\T (\hat \veo + \frac{\feastol}{2}\ve1) \geq (\feastol\sqrt{\ln n})\ve1]
        &= \Pr_{\hat\veo}\left[\sum_{j=1}^d A_{ij} (\hat \veo_j + \frac{\feastol}{2}) \geq \feastol\sqrt{\ln n}\right] \\
        &= \E_{\hat\veo}\left[ \Pr_{\eps}\left[\sum_{j=1}^d \eps_j A_{ij} (\hat \veo_j + \frac{\feastol}{2}) \geq \feastol\sqrt{\ln n}\right] \right]\\
        &\leq \E_{\hat\veo}\left[ \exp\left(-\frac{2(\feastol\sqrt{\ln n})^2}{\sum_{j=1}^d (2 A_{ij} \abs{\hat\veo_j + \frac{\feastol}{2}} )^2 }\right) \right] \\
        &\leq \E_{\hat\veo}\left[ \exp\left(-\frac{2(\feastol\sqrt{\ln n})^2}{\sum_{j=1}^d (A_{ij} \feastol )^2 }\right) \right] \\
        &\leq \E_{\hat\veo}\left[ \exp\left(-2(\sqrt{\ln n})^2\right) \right] \\
        &=  1/n^2.
    \end{align*}
    The first inequality is Hoeffding's, the second uses that $\veo \geq \hat\veo \geq -\feastol\ve1$ holds by \eqref{eq:perturbationswithintols}, and the third inequality used that the rows of $A$ each have Euclidean norm at most $1$.
    Hence, any particular inequality constraint has probability at most $1/n^2$ of being violated by $\vex$.
    Taking the union bound over all choices of $i \in [n]$,
    we find that $
    \Pr_{\hat\veo}[A(\hat \veo + \frac{\feastol}{2}\ve1) \leq (\feastol\sqrt{\ln n})\ve1]
     \geq 1-1/n$
    as required.
\end{proof}

Once again, since the failure probability is small enough we may use rejection sampling in order to assume that feasibility holds when necessary.
The increase in expected number of pivot steps is no greater than a factor $(1-\frac{1}{n})^{-1}$ by another application of Wald's equation.

Before Phase I can start, we must pivot from $\vex = \hat\veo+(\feastol/2)\ve1$ to the basis optimizing \eqref{eq:phase0}.
Notice that $\vex=\hat\veo+(\feastol/2)\ve1$ maximizes the objective $-\sum_{i=1}^d \vex_i$ over the same constraints.

\begin{lemma}\label{lem:phase0pivots}
    The optimal basic feasible solution to \eqref{eq:phase0} can be found in
    \[
    O\left(\sqrt{\frac{M d^3 \ln(n)}{\eta} \ln\big(d \ln(n) \cdot \norm{\ve1_d T^{-1}}_1 \big)}\right)
    \]
    pivot steps, where $M$ is the expected mean with of \eqref{eq:phase0} and $T$ is any invertible submatrix of $\begin{pmatrix}
        A \\
        - I_{d\times d}
    \end{pmatrix}$
    for which $\ve1 T^{-1} \geq 0$.
\end{lemma}
\begin{proof}
    We will apply \Cref{cor:easytouse}.
    When we consider \eqref{eq:phase0} in standard inequality form, we
    observe that all rows of the extended constraint matrix $\begin{pmatrix}A\\ -I_{d \times d}\end{pmatrix}$ have Euclidean norm at most $1$.
    The right-hand side vector consists of
    $\feastol\sqrt{\ln n}\ve1 + (\hat \veb - \veb)$
    and $- \hat \veo - \frac{\feastol}{2}\ve1$, which are respectively
    $(\feastol\sqrt{\ln n}\ve1,\eta,\gamma)$-exponentially distributed and $(-\frac{\feastol}{2}\ve1,\eta,\gamma)$-exponentially distributed with $\eta = \frac{\feastol}{4\ln(d+n)}$ and $\gamma=2\ln(d+n)$.

    Choose $\vecc = -\ve1 = -\sum_{i=1}^d \vece_i$ and observe that since $\vecc+\vetheta\leq\veo$, $\vecc + \vetheta$ is still maximized by $\vex=\hat\veo+(\feastol/2)\ve1$.
    Thus we can choose $\eps = 1$ for our initial objective $\vecc + \eps\vetheta$, observing that $\eps^{-1} = 1 \geq \frac{1}{1000d^3\gamma\norm{\vecc}_2}$ as required.
    Thus we apply \Cref{cor:easytouse} to find an upper bound on the expected number of pivot steps from $\vecc + \eps\vetheta$ to $\vetheta$ of
    \[102 + 316\sqrt{\frac{M d^3 \ln(n+d)}{\eta} \ln\left(\frac{6742 d^{3} N \ln^2(n+d)}{\eta} \right)},\]
    where $M$ is the expected mean width of the feasible set \eqref{eq:phase0}, defined in \Cref{def:meanwidth},
    and where $N$ is defined in \Cref{lem:N}.
    To upper bound $N$, since $N$ is defined as a minimizer across all pairs of bases, we must in particular have $N \leq N_{B^0,B^1}$ where for $B^1$ we take the constraints $\vex \geq \hat \veo + \frac{\feastol}{2}\ve1$, i.e. the basis which maximizes the objective $\vecc$, and for $B^0$ we take any basis whose basic submatrix $T$ of $\begin{pmatrix}
        A \\
        - I_{d\times d}
    \end{pmatrix}$
    satisfies $\ve1 T^{-1} \geq \ve0$. This basis will be dual feasible for minimizing objective $\vecc$.
    Let $\vef$ abbreviate the (unperturbed) right-hand side vector of \eqref{eq:phase0}, such that the two sets of basic inequalities are equivalently expressed as
    $A_{B^0}\vex \leq \vef_{B^0}$ and $A_{B^1}\vex \leq \vef_{B^1}$.
    
    By definition,
    \[N_{B^0,B^1} = \vecc^\T(A_{B^1}^{-1}\vef_{B^1} - A_{B^0}^{-1} \vef_{B^0}) + \gamma\eta(\norm{\vecc^\T A_{B^{1}}^{-1}}_{1} + \norm{\vecc^\T A_{B^{0}}^{-1}}_1).\]
    Once again, by our choice of $\gamma$ and $\eta$, we have $\gamma \eta = \frac{\feastol}{2}$.
    The basis $B^1$ indicates the constraints $-\vex \leq -\hat\veo - \frac{\feastol}{2}\ve1$.
    When computing $N$ we disregard the perturbation and get
    $A_{B^{1}}^{-1}\vef_{B^{1}} = \frac{\feastol}{2}\ve1$. Furthermore, $A_{B^{0}} = T$ by definition, and $A_{B^{1}} = -I_{d \times d}$ also by definition. 
    \begin{align*}
        N_{B^0,B^1} &= \vecc^\T(A_{B^1}^{-1}\vef_{B^1} - A_{B^0}^{-1} \vef_{B^0}) + \gamma\eta(\norm{\vecc^\T A_{B^{1}}^{-1}}_{1} + \norm{\vecc^\T A_{B^{0}}^{-1}}_1) \\
        &= \vecc^{\T}(\frac{\feastol}{2}\ve1 -T^{-1}\vef_{B^{0}}) + \frac{\feastol}{2}(\|\vecc\|_{1} + \|\vecc^{\T} T^{-1}\|_{1}). \\ 
        &= \vecc^{\T}(\frac{\feastol}{2}\ve1 - T^{-1}\vef_{B^{0}}) + \frac{\feastol}{2}(\|\vecc^{\T} T^{-1}\|_{1}+d). \\ 
        &\leq \vecc^\T(\frac{\feastol}{2}\ve1) + \|\vecc^{\T}T^{-1}\|_1\cdot\|\vef_{B^{0}}\|_\infty  + \frac{\feastol}{2}(\|\vecc^{\T} T^{-1}\|_{1}+d) \\ 
        &= \frac{\feastol}{2}(\vecc^{\T} \ve1 +d) +\|\vecc^{\T}T^{-1}\|_1\cdot\|\vef_{B^{0}}\|_\infty + \frac{\feastol}{2}\|\vecc^{\T} T^{-1}\|_{1} \\
        &= \|\vecc^{\T}T^{-1}\|_1\cdot\|\vef_{B^{0}}\|_\infty + \frac{\feastol}{2}\|\vecc^{\T} T^{-1}\|_{1},
    \end{align*}
    where to go from the second to last line we use that $\mathbf{c} = -\ve1$, so $\vecc^{\T} \ve1 = -d$. Finally, note that $\vef_{B^0}$ is a subvector of $\begin{pmatrix}\feastol\sqrt{\ln n}\ve1 \\ - \frac{\feastol}{2}\ve1 \end{pmatrix}$.
    Hence $\norm{\vef_{B^{0}}}_\infty \leq \feastol\sqrt{\ln n}$.
    
    Putting this together, we find $N \leq N_{B^0,B^1} \leq \feastol (\sqrt{\ln n}+\frac 1 2) d \norm{\ve1_d T^{-1}}$,
    which we fill in to the bound.
    Finally, we use $n \geq d$ to obtain $\ln(n+d) = O(\ln(n))$ and our choices of $\gamma$ and $\eta$ to finish the proof.
\end{proof}
\begin{remark}
    With $\gamma = 2\ln(d+n)$, $\eta = \frac{\feastol}{4\ln(d+n)}$, and $M$ the expected mean width of \eqref{eq:phase0}, the ratio $\frac{M}{\eta}$ is a measure of $A$, independent of $\feastol$ or $\veb$ or $\vecc$.
\end{remark}

\subsection{Phase I}

The Phase I auxiliary LP we use is stated below.  We extend the objective vector $\vetheta\in\sfe$ of the LP (\ref{eq:phase0}) to uniformly random objective $\bar\vetheta \in \mathbb{S}^d$ by sampling $\bar\vetheta_d$, appending it to $\vetheta$, and re-normalizing to ensure $\bar\vetheta$ has unit length.

\begin{align}
    \operatorname{maximize} \quad &\bar \vetheta^\T (\vex,s) \nonumber \\
    \operatorname{subject~to} \quad & A \vex + \frac{\veb - \feastol\sqrt{\ln n}\ve1}{\norm{\veb}_\infty} s  \leq \hat \veb \nonumber \\
    & \vex \geq \hat \veo + \frac{\feastol}{2\norm{\veb}_\infty}\ve1 s \nonumber \\
    &0 \leq s \leq \norm{\veb}_\infty. \tag{Phase I} \label{eq:phaseI}
\end{align}

Note that the bounds/right-hand sides of the first two sets of constraints are perturbed, but not those of the constraints $0 \leq s \leq \norm{\veb}_\infty$.
By construction, the face with $s=\norm{\veb}_\infty$ has the same solutions as \eqref{eq:phase0}.
The face with $s=0$ is the set of points which are exactly feasible for \eqref{eq:perturbed-LP}.
Hence we will follow two shadow paths: the first from $s + \eps\bar\vetheta^
\T(\vex,s)$ to $\bar\vetheta^
\T(\vex,s)$ and the second from $\bar\vetheta^
\T(\vex,s)$ to $-s + \eps\bar\vetheta^
\T(\vex,s)$.
We now determine our choice of $\eps > 0$ to make sure this path starts and ends where we need it to.

For the next lemma, we add an assumption that $\norm{\veb}_\infty \geq \feastol\sqrt{\ln n}$.
In all reasonable circumstances this should be satisfied, since $\feastol$ is a very small number and since $\norm{\veb}_\infty \geq 1$ should hold for all reasonable scalings of the LP data whenever $\veb \neq \veo$.

\begin{lemma}\label{lem:phase0opt}
    For a generic $\bar\vetheta \in \mathbb{S}^{d}$,
    let $(\vex^*, s^*)$ denote an optimal solution to
\begin{align}
    \operatorname{maximize} \quad &\eps \bar\vetheta^\T (\vex,s) + s \nonumber \\
    \operatorname{subject~to} \quad & A \vex - \frac{\veb - \feastol\sqrt{\ln n}\ve1}{\norm{\veb}_\infty} s  \leq \hat \veb \nonumber \\
    & \vex \geq \hat \veo + \frac{\feastol}{2\norm{\veb}_\infty}\ve1 s \nonumber \\
    & 0 \leq s \leq \norm{\veb}_\infty. \nonumber
\end{align}
    Assume $\norm{\veb}_\infty \geq \feastol\sqrt{\ln n}$.
    Let $s_\mathrm{min}$ denote the minimum singular value among all feasible bases of \eqref{eq:phase0}.
    If $0 \leq \eps \leq \frac{1}{1 + 2 \sqrt{d} s_\mathrm{min}^{-1}}$ then $s^* = \norm{\veb}_\infty$.
\end{lemma}
\begin{proof}
    Let $(\vex',s')$ be the basic solution obtained at the end of the path, where $\eps \to 0$.
    This solution must satisfy $s' \geq 0$ by feasibility.
    We assume that \eqref{eq:perturbed-LP} is feasible, and hence by \Cref{lem:feasibility} this value is achieved.
    Since $(\vex',s')$ maximizes the objective $s$,
    this means that $s' = \norm{\veb}_\infty$ and $\vex^*=\vex'$.
    If $(\vex',s')$ is maximal for every objective $\eps \bar\vetheta^\T (\vex,s) + s$ with every possible value $\eps > 0$ then the conclusion holds directly. Hence, we may assume for the remainder of this proof that at least one basic constraint has negative reduced cost for the objective $\bar\vetheta^\T (\vex,s)$.
    
    When $\eps = 0$, the corresponding dual solution for this basis has weight $1$ on the constraint $s \leq \norm{\veb}_\infty$ and weight $0$ on all other constraints.
    Thus if $\eps$ increases infinitesimally, by genericity of $\vetheta$ and optimality of $(\vex',s')$, the zero weights must go up and the weight on $s \leq \norm{\veb}_\infty$ must go down.
    In order to prove the lemma's conclusion, we require an upper bound on $\eps$ which guarantees that the weight on $s \leq \norm{\veb}_\infty$ remains non-negative.
    
    Let $M$ be the $d \times d$ submatrix of the $n \times d$ matrix $\begin{pmatrix} A \\ I_{d\times d} \end{pmatrix}$,
    $M_{\cdot,d+1}$ be the corresponding subvector of $\begin{pmatrix}-\frac{\veb - \feastol\sqrt{\ln n}\ve1_n}{\norm{\veb}_\infty} \\ \frac{\feastol}{2\norm{\veb}_\infty}\ve1_d\end{pmatrix}$,
    and $\vey$ be the subvector of $\begin{pmatrix} \hat b \\ \hat \veo\end{pmatrix}$
    such that the basic constraints for $(\vex',s')$ are equivalently stated as
    $M \vex + M_{\cdot,d+1} s\leq \vey$ and $s
    \leq \norm{\veb}_\infty$.
    The corresponding dual solution for an objective $\eps\bar\vetheta^\T(\vex,s)+s$ is obtained as
    $\velambda^\T = \eps(\bar\vetheta_1,\dots,\bar\vetheta_d) M^{-1}$ on the constraints $M \vex + M_{\cdot,d+1}s \leq \vey$ and
    a multiplier of
    \begin{equation}
        1 + \eps \bar\vetheta_{d+1} - M_{\cdot,d+1}^\T \velambda \label{eq:multiplierinphase0opt}
    \end{equation}
    on the constraint $s\leq \norm{\veb}_\infty$.
    When $\eps=0$, these multipliers are respectively $\veo_d$ and $1$.
    Since there exist values $\eps > 0$ for which the multipliers are all non-negative, the genericity of $\bar\vetheta$ mandates that $(\bar\vetheta_1,\dots,\bar\vetheta_d) M^{-1}$ is strictly positive.
    As such, the vector of multipliers is non-negative so long as \eqref{eq:multiplierinphase0opt} is non-negative.
    Hence the basic solution $(\vex',s')$ is optimal for any value $\eps \geq 0$ for which
    $\eps \leq \frac{1}{\abs{\bar\vetheta_{d+1} + (\bar\vetheta_1,\dots,\bar\vetheta_d) M^{-1} M_{\cdot,d+1}}}$. 
    Using Cauchy-Schwarz and $\norm{\bar\vetheta} = 1$, we find a sufficient condition that
    \[
        \eps \leq \frac{1}{1 + \norm{M_{\cdot,d+1}}\cdot\norm{M^{-1}}}.
    \]
    By definition we have $s_\mathrm{min} \leq \norm{M^{-1}}$.
    Since $\norm{\veb}_\infty \geq 2\feastol\sqrt{\ln n}$ by assumption, we find that $\norm{M_{\cdot,d+1}} \leq 2\sqrt{d}$. 

    Thus, by hypothesis,  
    \[\varepsilon \leq \frac{1}{1+2\sqrt{d} s_{\min}^{-1}} \leq \frac{1}{1 + \norm{M_{\cdot,d+1}}\cdot\norm{M^{-1}}},\]
     so $\varepsilon \overline{\theta}(\mathbf{x},s) + s$ is optimized at a vertex in the facet maximizing $\mathbf{s}$. It follows that $\mathbf{s}^{\ast}$ is maximal and therefore must be equal to $\|\mathbf{b}\|_{\infty}$ as desired.
\end{proof}

Assuming that \eqref{eq:perturbed-LP} is feasible, we will find a feasible solution when maximizing $-s$.

\begin{lemma}\label{lem:feasibility}
    The Phase I linear program \eqref{eq:phaseI} has a feasible solution $(\vex,s)$ with $s=0$ if and only if \eqref{eq:perturbed-LP} has a feasible solution.
\end{lemma}
\begin{proof}
    Directly by construction.
\end{proof}

Moreover, the solution that we find will maximize a uniformly randomly sampled objective direction.

\begin{lemma}\label{lem:phasetransition}
    Let $\bar\vetheta \in \mathbb{S}^{d}$ be arbitrary.
    Assume that $(\vex^*,s^*)$ is a feasible solution to \eqref{eq:phaseI}
    and that for some $\eps > 0$, $(\vex^*,s^*)$ maximizes the objective $(\vex,s) \mapsto -s + \eps \bar\vetheta^\T(\vex,s)$.
    If $s^*=0$ then $\vex^*$ is a feasible solution to \eqref{eq:perturbed-LP}.
    Moreover, if $(\vex^*,s^*)$ is basic for \eqref{eq:phaseI} then $\vex^*$ is basic for \eqref{eq:perturbed-LP} and maximizes $\vex \mapsto \bar\vetheta^\T(\vex,0)$.
\end{lemma}
\begin{proof}
    The feasibility of $\vex^*$ for \eqref{eq:perturbed-LP} when $s^*=0$ is direct from the construction.
    When $(\vex^*,s^*)$ is basic for \eqref{eq:phaseI} then, since $s^*=0$ and all basic solutions are non-degenerate with probability $1$, the basis describing this solution includes the constraint $s \geq 0$ and exactly $d$ constraints out of $A \vex + \frac{\veb - (\feastol\sqrt{\ln n})\ve1}{\norm{\veb}_\infty} s  \leq \hat \veb$ and $\vex \geq \hat\veo + \frac{\feastol}{2\norm{\veb}_\infty}\ve1 s$.
    Filling in $s=0$ in these constraints yields a set of $d$ constraints which all appear in the description of \eqref{eq:perturbed-LP} and which form a basis describing $\vex^*$.
    
    The same dual solution which certifies optimality of $(\vex^*,s^*)$ for the objective $(\vex,s) \mapsto -s + \eps \bar\vetheta^\T(\vex,s)$ over \eqref{eq:phaseI} also certifies,
    using the same multipliers but leaving out the constraint $0 \leq s$,
    the optimality of $\vex^*$ for the objective $\vex \mapsto \eps\bar\vetheta^\T(\vex,0)$ over \eqref{eq:perturbed-LP}.
\end{proof}

Similar to \Cref{lem:phase0opt}, we can estimate what value of $\eps > 0$ we need to reach.

\begin{lemma}\label{lem:phase1opt}
    For a generic $\bar\vetheta \in \mathbb{S}^{d}$,
    let $(\vex^*, s^*)$ denote an optimal solution to
\begin{align}
    \operatorname{maximize} \quad &\eps \bar\vetheta^\T (\vex,s) - s \nonumber \\
    \operatorname{subject~to} \quad & A \vex - \frac{\veb - \feastol\sqrt{\ln n}\ve1}{\norm{\veb}_\infty} s  \leq \hat \veb \nonumber \\
    & \vex \geq \hat \veo + \frac{\feastol}{2\norm{\veb}_\infty}\ve1 s \nonumber \\
    & 0 \leq s \leq \norm{\veb}_\infty. \nonumber
\end{align}
    Assume $\norm{\veb}_\infty \geq \feastol\sqrt{\ln n}$.
    Let $s_\mathrm{min}$ denote the minimum singular value among all feasible bases of \eqref{eq:perturbed-LP}.
    Suppose that \eqref{eq:perturbed-LP} is feasible.
    If $0 \leq \eps \leq \frac{1}{1 + 2 \sqrt{d} s_\mathrm{min}^{-1}}$ then $s^* = 0$.
\end{lemma}
\begin{proof}
    Let $(\vex',s')$ be the basic solution obtained at the end of the path, where $\eps \to 0$.
    This solution must satisfy $s' \geq 0$ by feasibility.
    We assume that \eqref{eq:perturbed-LP} is feasible, and hence by \Cref{lem:feasibility} this value is achieved.
    Since $(\vex',s')$ maximizes the objective $-s$,
    this means that $s' = 0$ and $\vex^*=\vex'$.
    
    When $\eps = 0$, the corresponding dual solution for this basis has weight $1$ on the constraint $-s \leq 0$ and weight $0$ on all other constraints.
    Thus if $\eps$ increases infinitesimally, by genericity of $\bar\vetheta$ and optimality of $(\vex',s')$, the weight on $-s \leq 0$ must go down while all other weights go up.
    In order to prove the lemma's conclusion, we require an upper bound on $\eps$ which guarantees that the weight on $-s \leq 0$ remains non-negative.
    
    Let $M$ be the $d \times d$ submatrix of the $n \times (d+1)$ matrix $\begin{pmatrix} A \\ -I_{d\times d} \end{pmatrix}$,
    $M_{\cdot,d+1}$ be the corresponding subvector of $\begin{pmatrix}-\frac{\veb - \feastol\sqrt{\ln n}\ve1_n}{\norm{\veb}_\infty} \\ \frac{\feastol}{2\norm{\veb}_\infty}\ve1_d\end{pmatrix}$,
    and $\vey$ be the subvector of $\begin{pmatrix} \hat b \\ \hat \veo\end{pmatrix}$
    such that the basic constraints for $(\vex',s')$ are equivalently stated as
    $M \vex + M_{\cdot,d+1} s\leq \vey$ and $-s
    \leq 0$.
    The corresponding dual solution for an objective $\eps\bar\vetheta^\T(\vex,s)-s$ is obtained as
    $\velambda^\T = \eps(\bar\vetheta_1,\dots,\bar\vetheta_d) M^{-1}$ on the constraints $M \vex \leq \vey$ and
    a multiplier of
    \begin{equation}
        1 - \eps \bar\vetheta_{d+1} - M_{\cdot,d+1}^\T \velambda \label{eq:multiplierinphase1opt}
    \end{equation}
    on the constraint $-s\leq 0$.
    When $\eps=0$, these multipliers are respectively $\veo$ and $1$.
    Since there exist values $\eps > 0$ for which the multipliers are all non-negative, the genericity of $\bar\vetheta$ mandates that $(\bar\vetheta_1,\dots,\bar\vetheta_d) M^{-1}$ is strictly positive.
    As such, the vector of multipliers is non-negative so long as \eqref{eq:multiplierinphase1opt} is non-negative.
    Hence the basic solution $(\vex',s')$ is optimal for any value $\eps \geq 0$ for which
    $\eps \leq \frac{1}{\abs{\bar\vetheta_{d+1} + (\bar\vetheta_1,\dots,\bar\vetheta_d) M^{-1} M_{\cdot,d+1}}}$.
    Using Cauchy-Schwarz and $\norm{\bar\vetheta} = 1$, we find a sufficient condition that
    \[
        \eps \leq \frac{1}{1 + \norm{M_{\cdot,d+1}}\cdot\norm{M^{-1}}}.
    \]
    By definition we have $s_\mathrm{min} \leq \norm{M^{-1}}$.
    Since $\norm{\veb}_\infty \geq \feastol\sqrt{\ln n}$ by assumption, we find that $\norm{M_{\cdot,d+1}} \leq 2\sqrt{d}$.
    This proves the lemma.
\end{proof}

We will bound our total Phase I running time in terms of the expected mean width of a slightly larger set than our Phase I LP.
The difference in size is minor.

\begin{definition}
Let $(-\hat L, \hat U)$ be $((0,U),\eta,\gamma)$-exponentially distributed, independent of $\hat \veb$.
    Let $\mathcal{M}$ denote the expected mean width of the following set
\begin{align}
    \{(\vex, s) \in \R^{d+1} : & A \vex + \frac{\veb - \feastol\sqrt{\ln n}\ve1}{\norm{\veb}_\infty} s  \leq \hat \veb, \nonumber \\
    & \vex \geq \hat \veo + \frac{\feastol}{2\norm{\veb}_\infty}\ve1 s \nonumber, \\
    &\hat L \leq s \leq \hat U \} \label{eq:phaseIextraperturb}
\end{align}
\end{definition}

\begin{theorem}\label{thm:phaseI}
    Let $\kappa$ denote the maximum condition number among all $d \times d$ submatrices of $\begin{pmatrix} A \\ -I_{d\times d} \end{pmatrix}$.
    Assume $\norm{\veb}_\infty \geq \feastol \sqrt{\ln n}$.
    The proposed simplex method satisfies the following properties:
    \begin{enumerate}
        \item $(\vex^0,s^0)=(\hat\veo+(\feastol/2)\ve1,\norm{\veb}_\infty)$ is a basic feasible solution to \eqref{eq:phase0}.
        \item The initial basis from which to start Phase I is found after an expected number of pivot steps no greater than 
        \[
            O\left(d^{1.5} \ln(n) \sqrt{\frac{\mathcal{M}}{\feastol} \ln\left(\frac{d \ln(n) \cdot \kappa}{\feastol} \right)}\right)
        \]
        \item Phase I uses an expected number of pivot steps no greater than
        \[O\left(d^{1.5} \ln(n) \sqrt{\frac{\mathcal{M}}{\feastol}\ln\left(\frac{d \norm{\veb}_\infty \ln(n) \kappa}{\feastol}\right)}\right)\]
        to find a feasible basis for \eqref{eq:perturbed-LP}.
        \item The basis of \eqref{eq:perturbed-LP} returned by Phase I maximizes a uniformly random objective direction $\vetheta \in \sfe$.
    \end{enumerate}
\end{theorem}
\begin{proof}
    The first point holds by construction, because we used rejection sampling to ensure that it was satisfied.
    The second point is established using \Cref{lem:phase0pivots},
    which gives an upper bound of
    \[
    O\left(\sqrt{\frac{M^0 d^3 \ln(n)}{\eta} \ln\left(\frac{d \ln(n) \cdot \gamma \cdot \norm{\ve1_d T^{-1}}_1}{\eta} \right)}\right),
    \]
    where $M^0$ denotes the expected mean width of \eqref{eq:phase0}.
    We now fill in all the quantities used in this bound.
    We have $\eta = \frac{\feastol}{4\ln(d+n)}$ and $\gamma=2\ln(d+n)$.
    The rows of $A$ each have norm at most $1$, which means that any basis matrix from $\begin{pmatrix} A \\ -I_{d\times d} \end{pmatrix}$
    has operator norm at most $d$.
    Since the condition number of a matrix $T$ is $\norm{T}\cdot\norm{T^{-1}}$, we know that the vector
    $\ve1_d T^{-1}$ satisfies 
    \[
    \norm{\ve1_d T^{-1}}_1 \leq d \norm{T^{-1}} \leq d^2 \kappa.
    \]
    Finally, to obtain the desired statement, we must bound the expected mean width $M^0$ of \eqref{eq:phase0}.
    To achieve that, notice that the feasible set of \eqref{eq:phase0} is a lower-dimensional slice of \eqref{eq:phaseIextraperturb}, namely the slice where $s=0$.
    As such, the expected mean width $M^0$ is at most the mean width of \eqref{eq:phaseIextraperturb} divided by a factor
    $1 - O(\frac{1}{\sqrt{d}}) \leq \E[\norm{(\bar\vetheta_1,\dots,\bar\vetheta_d)}] \leq 1$. This scaling factor is to account for the fact that the mean width of the slice is being measured in a lower dimension.

    For the third point, we traverse shadow paths on the feasible set \eqref{eq:phaseI}.
    First we follow the shadow path from fixed objective $s$ to random objective $\bar\vetheta^\T(\vex,s)$.
    As stated in \Cref{lem:phase0opt}, the solution we found by solving \eqref{eq:phase0} gives a starting basis for this path which is optimal for $s+\eps\bar\vetheta^\T(\vex,s)$ with $\eps=\frac{1}{1+2\sqrt{d}s_\mathrm{min}^{-1}}$.
    Observe that, when we write the constraints of \eqref{eq:phaseI} as a standard inequality form LP, then the rows of the constraint matrix each have norm at most $3$.
    We can still apply \Cref{thm:TODO} by scaling down all inequalities by a factor $3$, which reduces the effective value of $\eta$ by the same factor $3$.
    This part of the shadow path has expected length at most
\[113 + 344d \sqrt{\frac{d \ln(n+d+2) \mathcal{M}}{\eta/3} \ln\left(\frac{6742d^{3}(\norm{\veb}_\infty + 2\eta\gamma/3)\ln^2 (n+d+2)}{\eta \varepsilon/3} \right)}.\]
We further fill in the quantities in the bound.
We bound $s_\mathrm{min}^{-1} \leq d^2\kappa$ as before and fill in $\eps=\frac{1}{1+2\sqrt{d}s_\mathrm{min}^{-1}}$ and $\eta = \frac{\feastol}{4\ln(d+n)}$ and $\gamma=2\ln(d+n)$.
This gets an asymptotic running time of
\[O\left(d^{1.5} \ln(n) \sqrt{\frac{\mathcal{M}}{\feastol} \ln\left(\frac{d\norm{\veb}_\infty\ln (n)\kappa}{\feastol} \right)}\right).\]

The next shadow path traverses the feasible set \eqref{eq:phaseI}, now from the random objective $\bar\vetheta^\T(\vex,s)$ to the fixed objective $-s$.
We stop at the intermediate objective $-s+\eps\bar\vetheta^\T(\vex,s)$ with $\eps=\frac{1}{1+2\sqrt{d}s_\mathrm{min}^{-1}}$.
As shown in \Cref{lem:phase1opt}, this yields a solution with $s=0$. We see in \Cref{lem:phasetransition} that this solution satisfies the properties claimed in the final point of this theorem.
The running time bound for this final path of Phase I is identically obtained as the previous path.
\end{proof}

\subsection{Phase II}
After the end of Phase I, the current basis of the algorithm is optimal for
\begin{align*}
    \operatorname{maximize} \quad & \vetheta^\T \vex \\
    \operatorname{subject~to} \quad & A \vex \leq \hat \veb \\
    & \vex \geq \hat \veo.
\end{align*}
As such, we can follow the shadow path starting at the auxiliary objective $\vetheta$ in order to optimize the linear program
\begin{align*}
    \operatorname{maximize} \quad & \vecc^\T \vex \\
    \operatorname{subject~to} \quad & A \vex \leq \hat \veb \\
    & \vex \geq \hat \veo.
\end{align*} up to optimality tolerance $\opttol > 0$ as follows.
We terminate upon finding the optimal basic feasible solution $\vex^* \in \R^d$ to the intermediate linear program
\begin{align}
    \operatorname{maximize} \quad & (\vecc + \opttol\cdot \vetheta)^\T \vex \label{eq:intermediatelp}\\
    \operatorname{subject~to} \quad &  A \vex \leq \hat \veb \nonumber\\
    & \vex \geq \hat \veo. \nonumber
\end{align}
Here, the conditions of \Cref{cor:easytouse} are satisfied: the objective and constraint matrix are fixed, all rows of the constraint matrix have norm no greater than $1$, the right-hand side and bound vector and the objective are sampled appropriately.
Our Phase II algorithm returns its solution after an expected number of pivot steps no more than
\[
(1-\frac{1}{n})^{-1}(1-\frac{1}{d+n})^{-1}\cdot\left(102 + 316 d^{1.5}\ln(n+d)\sqrt{\frac{M}{\feastol}\ln\left(\frac{6742d^3N\ln^3(n+d)}{\opttol\cdot\feastol}\right)}\right),
\]
where $N$ is defined in \Cref{lem:N} and indicates the expected maximum absolute objective value of any feasible solution.

Based on our observations, which can be found in \Cref{sec:bythebooksimplexfindings}, reasonable values for the remaining parameters are a mean width of $M \leq 1000$ and tolerances of $\feastol = \opttol = 10^{-6}$. The expected maximum absolute objective value, which appears in the logarithm, should easily be bounded as $N \leq 10^9$. 

Besides the optimal basic feasible solution to the primal \eqref{eq:perturbed-LP}, the algorithm will yield the optimal basic feasible solution $(\vey^*, \vet^*) \in \R^{n + n + d}$ to the intermediate dual linear program
\begin{align*}
    \operatorname{minimize} \quad & \hat \veb^\T \vey - \hat \veo^\T t \\
    \operatorname{subject~to} \quad & A^\T \vey - \vet = \vecc + \opttol\cdot \vetheta \\
    & \vey,\vet \geq \veo.
\end{align*}
We can now verify that the primal dual solution pair satisfy the approximate feasibility and complementary slackness conditions described in \Cref{sub:introbythebookproperties}.
The primal feasibility conditions
\begin{align*}
    A\vex^* &\leq \veb + \feastol\cdot \ve1_n \\
    -\feastol\cdot\ve1_d &\leq \vex^*
\end{align*}
follow immediately from \eqref{eq:perturbationswithintols}.
The dual feasibility conditions
\begin{align*}
    A^\T \vey^* &\geq \vecc - \opttol\ve1_d \\
    \vey^* &\geq \veo
\end{align*}
follow because $\vetheta\geq-\ve1_d$.
Note that in the optimal solution we have $\vet^* = A^\T \vey^* - \vecc - \opttol \cdot \vetheta$.
To certify the approximate optimality of the pair $\vex^*, \vey^*$ we derive approximate complementary slackness conditions.
Exact complementary slackness for \eqref{eq:intermediatelp}, along with \eqref{eq:perturbationswithintols}, gives us that
\begin{align*}
    \text{~for all~} i \in [n] &\text{,~if~} \vey^*_i > 0 \text{~then~} (A\vex^*)_i = \hat \veb_i \geq \veb_i.  \\
    \text{~for all~} j \in [d] &\text{,~if~} (A^\T \vey^*)_j > \vecc_j + \opttol \cdot \vetheta_j \text{~then~} \vex_j^* = \hat \veo_j \leq 0.
\end{align*}
Note that $\norm{\vetheta}=1$, which means that $\vecc_j < (A^\T \vey^*)_j - \opttol$ implies $\vecc_j < (A^\T \vey^*)_j - \opttol\vetheta_j$, which in turn implies $\vex_j^* \leq 0$.
It follows that the returned solution pair $\vex^*,\vey^*$ satisfies the following approximate complementary slackness conditions:
\begin{align*}
    \text{~for all~} i \in [n] &\text{,~if~} \vey^*_i > 0 \text{~then~} (A\vex^*)_i \geq \veb_i.  \\
    \text{~for all~} j \in [d] &\text{,~if~} (A^\T \vey^*)_j > c_j + \opttol \text{~then~} \vex_j^* \leq 0.
\end{align*}
This is what we promised the algorithm output would satisfy: feasibility and optimality for \eqref{eq:input-LP} up to tolerances as in  \eqref{eq:apx-feas} and \eqref{eq:apx-opt}.

%primal
%max (c+opttoltheta)^T x
%st Ax <= hat b
%-x <= - hat 0

%dual
%min hat b^T y + (-hat 0)^T t
%st A^T y + (-I)^T t = c+opttoltheta
%y,t >= 0

%complementary slackness:
%y_i > 0 implies (Ax)_i = hat b_i
%t_j > 0 implies -x_j = - hat 0_j

%t = A^T y - c - opttoltheta
%(A^T y - c - opttoltheta)_j > 0 implies -x_j = - hat 0_j
%(A^T y)_j > c_j + opttoltheta implies x_j = hat 0_j

%make the antecedent stronger
%(A^T y)_j > c_j + opttol implies x_j = hat 0_j

Note that we do not require to relax the dual nonnegativity constraint $\vey \geq \veo$.
In real solvers this is likely relaxed to $\vey \geq -\feastol\ve1$ because it is obtained through the (inexact) solving of a system of linear equations.
Since we assume that our algorithm may solve such systems exactly, we are able to return solutions with $\vey \geq \veo$.

\section{Running Time Bound}\label{sec:mathysection}
The proof will proceed as follows. In \Cref{sec:goodslacks} we use bound perturbations to prove that bases are either unlikely to be feasible or have large slacks on all non-tight constraints with constant probability when feasible.
In \Cref{sub:largeredcosts} we will use the semi-random shadow vertex method to prove that if bases are dual feasible for some objective, then with constant probability all of their multipliers will be large and positive for some (potentially distinct) intermediate objective.
Finally in \Cref{sec:shadowpath} we use these two properties to prove an upper bound on the length of the semi-random shadow path with bound perturbations, which is the property stated in \Cref{thm:bound_p_slack-epoch}.

\subsection{Bases with Large Slack}\label{sec:goodslacks}
In this section we study the effect of the perturbations defined in \Cref{def:ourperturbations}.
The bound perturbations will have the effect that vertices which are feasible will likely have large slacks on the non-tight constraints.
We first prove this fact for fixed (not necessarily basic) solutions.

\begin{lemma}\label{lem:exists-slack-helper}
    Let $A \in \R^{k \times d}$ be arbitrary and $\hat\veb \in \R^k$ be $(\veb,\eta,\gamma)$-exponentially distributed.
    Let $\omega > 0$ and $\vex \in \mathbb{R}^{d}$ be fixed.
    Let $\ves = \hat\veb - A\vex \in \mathbb{R}^{k}$ denote the slack vector. Then we have
    \begin{align*}
        \Pr\big[ \exists j \in [k] \; : \; \ves_j \leq \omega \big| \ves \geq \veo \big]  \leq 62 \frac{\omega}{\eta}  \ln(\frac{1}{\Pr[\ves \geq \veo]}).
    \end{align*}
\end{lemma}
\begin{proof}
First, by independence of the random variables $\ves_j$ we find that
 \begin{align}\label{align:exists}
        \Pr\big[ \exists j \in [k] \; : \; \ves_j \leq \omega \big| \ves \geq \veo \big]
        &\leq \sum_{j \in [k]} \Pr\big[ \ves_j \leq \omega  \big | \ves_j \geq 0 \big].
    \end{align}

Second, 
let $V \coloneqq \{j \colon \ves_j <\veo \} $ be the set of violated constraints.
Since the numbers $\ves_j$ are independent,
the Chernoff bound (see, e.g., Theorem 4.5 in \cite{probabilityandcomputing})gives us that $\Pr[\ves \geq \veo] = \Pr\big[\abs{V} = 0\big] \leq \exp(-\E[\abs{V}]/2)$. Now we need to relate \eqref{align:exists} with the set of violated constraints. 
Therefore we use \Cref{lem:general-condition-reversing} (see \cite[Lemma 31 in ArXiv v1/v2]{optimal_smoothed_analysis})
which we can apply to the entries of the $(\overline{\veb},\eta,\gamma)$-exponential distribution directly and obtain that $\Pr[\ves_{j} \leq \omega | \ves_{j} \geq 0] \leq 31 \frac{\omega}{\eta} \Pr[\ves_{j} < 0]$. 
Putting everything together gives us 
 \begin{align*}
        \Pr\big[ \exists j \in [k] \; : \; \ves_j \leq \omega \big| \ves \geq \veo \big]
        &\leq \sum_{j \in [k]} \Pr\big[ \ves_j \leq \omega  \big | \ves_{j} \geq 0 \big] \\
        & \leq \sum_{j \in [k]} 31 \frac{\omega}{\eta} \Pr\big[  \ves_{j} < 0 \big] \\
        & = 31 \frac{\omega}{\eta}  \E[\abs{V}] \\
        & \leq 62 \frac{\omega}{\eta} \ln(\frac{1}{\Pr[\ves \geq \veo]}).  
    \end{align*}
\end{proof}

We now extend the above result to the slacks on non-tight constraints for basic feasible solutions.
\begin{lemma}\label{lem:exists-slack}
    Let $A \in \R^{n \times d}$ be arbitrary and $\hat\veb \in \R^n$ be $(\veb,\eta,\gamma)$-exponentially distributed, and let $\omega > 0$.
    Assume $B \subset[n]$ is a basis and $\vex^B = A_B^{-1}\hat\veb_{B}$.
    If $\vex^B \in \R^d$ satisfies $\Pr[A\vex^B \leq \hat\veb] \geq n^{-d}$, then \begin{align*}
        \Pr\big[ \exists j \in [n]\setminus B \; : \; A_j \vex^B \geq \hat\veb_j - \omega \big| A \vex^B \leq \hat\veb \big]  \leq 124 d \frac{\omega}{\eta}  \ln(n) + n^{-d}.
    \end{align*}
\end{lemma}
\begin{proof}
    Assume without loss of generality that $B = \{n-d+1,\dots,n\}$.
    Write $k = n-d$ and $\ves \in \R^k$ with $\ves = \hat\veb_{[n]\setminus B} - A_{[n]\setminus B}\vex^B$.
    This allows us to write
    \[
    \Pr\big[ \exists j \in [n]\setminus B \; : \; A_j \vex^B \geq \veb_j - \omega \big| A \vex^B \leq \hat\veb \big]
    =
    \Pr\big[ \exists j \in [k] \; : \; \ves_j \leq \omega \big| \ves \geq \veo \big].
    \]
    We separate the randomness in $\hat\veb$ into the part on $\hat\veb_B$ and the part on $\hat\veb_{[n]\setminus B}$ and get
    \begin{equation}\label{eq:splitrandomness}
        \Pr\big[ \exists j \in [k] \; : \; \ves_j \leq \omega \text{~and~} \ves \geq \veo \big] =
        \E_{\hat\veb_B}\Big[\Pr_{\hat\veb_{[n]\setminus B}}\big[ \exists j \in [k] \; : \; \ves_j \leq \omega  \text{~and~} \ves \geq \veo \big]\Big].
    \end{equation}
    In other words, by the principle of deferred decisions, we may condition on the value of $\hat\veb_{B}$ and therefore assume it to be fixed. In particular, for any fixed value of $\hat\veb_B$, we will prove an upper bound on
    $\Pr_{\hat\veb_{[n]\setminus B}}\big[ \exists j \in [k] \; : \; \ves_j \leq \omega \big| \ves \geq \veo \big]$.
    If $\hat\veb_B$ satisfies $\Pr_{\hat\veb_{[n]\setminus B}}[\ves\geq \veo] \geq n^{-2d}$ then we apply \Cref{lem:exists-slack-helper} to find

    \begin{align*}
    \Pr_{\hat\veb_{[n]\setminus B}}\big[ \exists j \in [k] \; : \; \ves_j \leq \omega \text{~and~} \ves \geq \veo \big]&\leq 62 \frac{\omega}{\eta} \ln\left(\frac{1}{\Pr[\ves \geq \veo]}\right)\Pr[\ves \geq \veo]\\
    &\leq 124 d \frac{\omega}{\eta} \ln(n)\Pr[s \geq 0].
    \end{align*}
    If $\hat\veb_B$ satisfies $\Pr_{\hat\veb_{[n]\setminus B}}[\ves\geq \veo] \leq n^{-2d}$ then we upper bound
    \[
    \Pr_{\hat\veb_{[n]\setminus B}}\big[ \exists j \in [k] \; : \; \ves_j \leq \omega \text{~and~} \ves \geq \veo \big]\leq \Pr_{\hat\veb_{[n]\setminus B}}\big[ \ves \geq \veo \big] \leq n^{-2d}.
    \]
    Combining the two cases, we find that for every value of $\hat\veb_B$ it holds that
    \[
    \Pr_{\hat\veb_{[n]\setminus B}}\big[ \exists j \in [k] \; : \; \ves_j \leq \omega \text{~and~} \ves \geq \veo \big]\leq124 d \frac{\omega}{\eta} \ln(n)\Pr_{\hat\veb_{[n]\setminus B}}[\ves \geq \veo] + n^{-2d}.
    \]
    In particular, by taking expectations over $\hat\veb_{B}$ of both sides of this inequality, this pointwise upper bound can now be applied to \eqref{eq:splitrandomness} to find
    \[
    \Pr\big[ \exists j \in [k] \; : \; \ves_j \leq \omega \text{~and~} \ves \geq \veo \big]
    \leq
    124 d \frac{\omega}{\eta} \ln(n)\Pr[\ves \geq \veo] + n^{-2d}
    \]

    Now, by dividing both sides by $\Pr[\ves \geq \veo]$ and using $\Pr[\ves \geq \veo] \geq n^{-d}$, we conclude that
    \[
    \Pr\big[ \exists j \in [k] \; : \; \ves_j \leq \omega \big| \ves \geq \veo \big]
    \leq
    124 d \frac{\omega}{\eta} \ln(n) + n^{-d},
    \]
    which finishes the proof.
\end{proof}

The lemmas above allow us to prove our main technical result about bound perturbations:
for any fixed set of bases, the expected number of feasible bases among them is comparable to the expected number of feasible bases with large slacks.

\begin{definition}\label{def:feasible-gap-bases}
    Given arbitrary $A \in \R^{n \times d}$ and $\veb \in \R^n$, as well as a value $\eta > 0$,
    define the set of feasible bases $B$ as
    \[
        \mathcal{F}(A) = \left\{B \in \binom{[n]}{d} : \text{$A_B$ invertible and~} A \vex^B \leq \veb  \right\}. 
    \]
    When $A$ is clear from context or arbitrary, we may refer to $\mathcal{F}(A)$ as just $\mathcal F$.

    We define $\omega(\eta,n,d) = \frac{\eta}{1240 d \ln(n)}$, and we define the set of \emph{feasible bases with $\eta$-good slack} as
    \[
        \mathcal{F}^+(A,\eta) = \left\{B \in \mathcal{F} \colon A_{[n]\setminus B} \vex_B \leq \veb_{[n]\setminus B} - \omega(\eta,n,d) \right\}.
    \]
      When $\eta,n$ and $d$ are clear from context, we may simply refer to $\omega(\eta,n,d)$ as $\omega$.
      When $A$ and $\eta$ are clear from context or arbitrary, we may refer to $\mathcal{F}^+(A,\eta$) as just $\mathcal{F}^+$, and we may refer to it simply as the \emph{feasible bases with good slack}. 
\end{definition}

\begin{theorem}(Slacks are large)\label{thm:slacks}
    Let $A \in \R^{n \times d}$ be arbitrary and $\hat\veb \in \R^n$ be $(\veb,\eta,\gamma)$-exponentially distributed.
    Let $\mathcal S \subseteq \binom{[n]}{d}$ be arbitrary.
    Then
    \[
        \E[\abs{\mathcal S \cap \mathcal{F}^+}] \geq
        0.9 \E[\abs{\mathcal S \cap \mathcal F}] - 2.
    \]
\end{theorem}
\begin{proof}
Let $\mathcal S' := \{B \in \mathcal S : \Pr[B \in \mathcal F] \geq n^{-d}\}$
denote those bases in $\mathcal S$ which have non-negligible probability of being feasible.
Observe that
\[
    \E[\abs{\mathcal S \cap \mathcal{F}^+}] \geq \E[\abs{\mathcal S' \cap \mathcal{F}^+}]
\]
follows directly from the fact that $\mathcal S' \subseteq \mathcal S$.

Let $B \in \mathcal S'$ index a basis and let $\hat\vex^B = A_B^{-1}\hat\veb_B$.
Recall that $\ves \coloneqq \hat\veb_{[n] \setminus B} - A_{[n] \setminus B} \vex^B \in \R^{n-d}$. 
Because $B \in \mathcal S'$ we have $\Pr[B \in \mathcal F] \geq n^{-d}$,
so we may apply \Cref{lem:exists-slack} to
find for $\omega = \omega(\eta,n,d)$ that
      \begin{align*}
	    \Pr[B \notin \mathcal{F}^+ \mid B \in \mathcal{F}]
            &= \Pr[\exists j \in [n]\setminus B \; : \; \ves_j \leq \omega \big| \ves \geq \veo \big] \\
            &\leq 124 \frac{\omega}{\eta} d \ln(n) + n^{-d} \leq 0.1 + n^{-d}.
        \end{align*}      
Hence, $\Pr[B \in \mathcal{F}^+ \mid B \in \mathcal{F}] \geq 0.9 - n^{-d}$.
The result follows by summing up over all $B \in \mathcal S$:
\begin{align*}
    \E[\abs{\mathcal S' \cap \mathcal{F}^+}] &= \sum_{B \in \mathcal S'} \Pr[B \in \mathcal{F}^+] \\
    &= \sum_{B \in \mathcal S'}  \Pr[B \in \mathcal{F}^+| B \in \mathcal F]\Pr[B \in \mathcal F] \\
    &\geq \sum_{B \in \mathcal S'}  (0.9 - n^{-d})\Pr[B \in \mathcal F] \\
    &= (0.9-n^{-d}) \E[\abs{\mathcal S' \cap \mathcal F}] \\
    &\geq 0.9 \E[\abs{\mathcal S' \cap \mathcal F}] - 1.
\end{align*}
Putting the above together, we find
$\E[\abs{\mathcal S \cap \mathcal{F}^+}] \geq 0.9 \E[\abs{\mathcal S' \cap \mathcal F}] - 1$.
Finally, we note that
\[
\E[\abs{(\mathcal S\setminus \mathcal S') \cap \mathcal F}] = \sum_{B \in \mathcal S\setminus \mathcal S'} \Pr[B \in \mathcal F] \leq 1
\]
holds by construction, which implies that
$\E[\abs{\mathcal S \cap \mathcal{F}^+}] \geq 0.9 \E[\abs{\mathcal S \cap \mathcal F}] - 2$
    as required.
\end{proof}

\subsection{Bases with Large Multipliers}\label{sub:largeredcosts}

In this subsection we prove one main ingredient for our proof of \Cref{cor:easytouse}. We show that for any (normal) cone we have large multipliers with high probability.  

\begin{definition}
    A function $f : \R^d \to \R_+$ is $L$-log-Lipschitz if for any $\vex,\vex' \in \R^d$ it holds that
    $f(\vex)/f(\vex') \leq \exp(L\norm{\vex-\vex'})$.
    A probability distribution is $L$-log-Lipschitz if it admits a probability density function with respect to the Lebesgue measure
    and if this probability density function is $L$-log-Lipschitz.
\end{definition}

\begin{definition} \label{def:lexponential}
    Let $L > 0$.
    A random variable $X \in \R^d$ is $L$-exponentially distributed on $\R^d$ if there is a constant $C_L$ such that $$\Pr[X \in S] = \int_{S} C_L e^{-L\norm{x}} \dd \vex $$
    for every measurable set $S \subset \R^d$
\end{definition}

The following fact is standard, we adapt the proof from \cite{dadushhahnle}.
\begin{lemma}\label{fact:norm_exp_distr}
    The normalizing constant $C_L$ of the $L$-exponential distribution is $C_L = \frac{L^d}{d! \vol_d(\ball^d)}$.
    For $X$ exponentially distributed on $\R^d$, the $k$'th moment of $\norm{X}$ is $\E[\norm{X}^k] = \frac{L^{-k} (k+d-1)!}{(d-1)!}$.
\end{lemma}
\begin{proof}
    Without loss of generality, take $L=1$.
    For the normalizing constant,
\begin{align*}
C^{-1} &= \int_{\R^d} e^{-\norm{\vex}} \dd\vex
    = \int_{\R^d} \int_{\norm{\vex}}^\infty e^{-t} {\rm dt} \dd\vex \\
    &= \int_0^\infty e^{-t} \int_{\R^d} 1[\norm{\vex} \leq t] \dd \vex \dd t
    = \int_0^\infty e^{-t} t^d \vol_d(\mathbb{B}^d) \dd t
    = d! \vol_d(\mathbb{B}^d) \text{.}
\end{align*}
To compute the $k$'th moment of the norm,
\begin{align*}
  \E[\|X\|^k ]
    &= C \int_{\R^d} \|\vex\|^k e^{-\|\vex\|} \dd\vex
     = C \int_{\R^d} \|\vex\|^k \int_{\|\vex\|}^\infty e^{-t} {\rm dt} \dd\vex \\
    &= C \int_0^\infty  e^{-t} \int_{t\mathbb{B}^d} \|\vex\|^k
        \dd\vex \dd t
     = C \int_0^\infty  e^{-t} t^{d+k} \dd t
          \int_{\mathbb{B}^d} \|\vex\|^k \dd \vex \\
    &= C \int_0^\infty  e^{-t} t^{d+k} \dd t
          \int_{\mathbb{B}^d} \int_0^1 1[\norm{\vex} \geq s] k s^{k-1} \dd s \dd \vex \\
    &= C \int_0^\infty  e^{-t} t^{d+k} \dd t
         \vol_d(\mathbb{B}^d) \int_0^1 (1-s^d) k s^{k-1} \dd s \\
    &= C (d+k)! \int_0^1 k (1-s^{d+k-1}) \vol_d(\mathbb{B}^d) \dd s \\
 &= C (d+k)! \vol_d(\mathbb{B}^d) \frac{d}{d+k} = \frac{(d+k-1)!}{(d-1)!} \text{.} \qedhere
\end{align*}
\end{proof}
For the exponential distribution we have the following tail bound.

\begin{lemma}\label{lem:tailbound}
    Let $X$ be $L$-exponentially distributed on $\R^d$ with some $L > 0$. Then we have
    $$\Pr[\norm{X} \geq 2ed L^{-1} \ln n] \leq n^{-d}.$$
\end{lemma}
\begin{proof}
Without loss of generality we assume $L=1$.
Using Markov's inequality we know for $k = d \ln n$ and $t=2ed\ln n$ that
    \begin{align*}
        \Pr[\norm{X} > t] &= \Pr[\norm{X}^k > t^k] \\
                          &\leq \frac{\E[\norm{X}^k]}{t^k} \\
                          &\leq \frac{(k+d)^k}{t^k} \\
                          &\leq \frac{(2d\ln n)^{d \ln n}}{(2ed\ln n)^{d \ln n}} = n^{-d}. \qedhere
    \end{align*} 
\end{proof}

For any $L$-log-Lipschitz distribution, and in particular for the $L$-exponential distribution, we will see that a line segment shifted by a random offset will tend to intersect any cone through its center.
\begin{lemma}
\label{lem:goodmultipliers}
Let $M \in \mathbb{R}^{d\times d}$ be an invertible matrix, and let $\vep$ be the sum of the row vectors of $M$. Then we define for any $m \geq 0$, the cone $C_{m} = \{\vey \in \mathbb{R}^{d}: \vey^\T M^{-1} \geq m \ve1\}$. Fix any $\vecc, \vecc' \in \R^{d}$ and $m > 0$. Then for any random vector $\vez \in \R^{d}$ with $L$-log-Lipschitz probability distribution $\mu$,
\begin{align*}\Pr\Big[[\vecc + \vez, \vecc' + \vez] \cap C_{m} \neq \emptyset\Big] \geq e^{-Lm\|\vep\|_{2}}\Pr\Big[[\vecc + \vez, \vecc' + \vez] \cap C_{0} \neq \emptyset\Big] .\end{align*}
\end{lemma}
\begin{proof} 
 Since $\vep$ is the sum of the rows of $M$ and $M$ is invertible, $\vep^\T M^{-1} = \sum_{i=1}^{d} \vece^{i} =  \ve1$. Then $\vey \in C_{0}$ if and only if $\vey^\T M^{-1} \geq \veo$, which is true if and only if $(\vey + m \vep)^\T 
 M^{-1} \geq m \ve1$. This in turn is true if and only if $\vey + m \vep \in C_{m}$.

Note that $\Pr[[\vecc + \vez, \vecc'+\vez] \cap C_{0} \neq \emptyset] = \int_{-[\vecc, \vecc'] + C_{0}} \mu(\vey) d \vex$. Then, since $\vey \in C_{m}$ if and only if $\vey - m\vep \in C_{0}$, 
\begin{align*}\int_{-[\vecc, \vecc'] + C_{0}} \mu(\vey) d \vey = \int_{-[\vecc, \vecc'] + C_{m}} \mu(\vey - m\vep) d\vey.\end{align*}
Then, by applying our assumption that $\mu$ is $L$-log-Lipschitz, we get that
\begin{align*}\int_{-[\vecc, \vecc'] + C_{m}} \mu(\vey -m\vep) \dd \vey \leq \int_{-[\vecc, \vecc'] + C_{m}} \mu(\vey)e^{L m\|\vep\|_{2}} \dd \vey = e^{Lm\|\vep\|_{2}}\Pr[[\vecc + \vez, \vecc' + \vez] \cap C_{m} \neq \emptyset].\end{align*}
\end{proof}

\begin{definition}
    Let $A \in \R^{n \times d}$, let $\vecc, \vez \in \R^d$
    and let $t_2 > t_1 \geq 0$ and $L > 0$ be arbitrary.
    We define
    \[
        \mathcal{M}_{[t_1,t_2]}(A,\vecc,\vez) = \left\{B \in \binom{[n]}{d} : \text{$A_B$ invertible and~} \exists \  t \in [t_1,t_2] \text{ such that } (t \vecc + \vez)^{\T}A_{B^{i}}^{-1} \geq \veo \right\},
    \]
    and we call $\mathcal{M}_{[t_1,t_2]}(A,\vecc,\vez)$ the \emph{set of bases with positive $(\vecc,\vez)$-multipliers on $[t_1,t_2]$}.  When $A$, $\vecc$, $\vez$, $t_1$ and $t_2$ are clear from context or arbitrary, we may refer to $\mathcal{M}_{[t_1,t_2]}(A,\vecc,\vez)$ by just $\mathcal{M}$, and may refer to is simply as the \emph{set of bases with positive multipliers}.
    \\\\
    We define 
    \[
        \mathcal{M}_{[t_1,t_2]}^+(A,\vecc,\vez,L) = \left\{B \in \binom{[n]}{d} \colon \exists \  t \in [t_1,t_2] \text{ such that } (t \vecc + \vez)^{\T}A_{B^{i}}^{-1} \geq \frac{\ln(1/0.9)}{dL}\ve1 \right\},
    \]
     and we call $\mathcal{M}_{[t_1,t_2]}^+(A,\vecc,\vez,L)$ the \emph{set of bases with $L$-good $(\vecc,\vez)$-multipliers on $[t_1,t_2]$}.  When $A$, $\vecc$, $\vez$, $t_1$, $t_2$ and $L$ are clear from context or arbitrary, we may refer to $\mathcal{M}_{[t_1,t_2]}^+(A,\vecc,\vez,L)$ by just $\mathcal{M}^+$, and may refer to it simply as \emph{the set of bases with good multipliers}.
\end{definition}

\begin{theorem}(Multipliers are large)\label{thm:reducedcosts}
    Let $\mathcal S \subseteq \binom{[n]}{d}$ and $t_2 > t_1 \geq 0$ and $L>0$ be arbitrary.
    Let $A \in \R^{n \times d}$ be a matrix with rows of Euclidean norm at most $1$, $\vecc \in \R^d$ be arbitrary and $\vez \in \R^d$ be $L$-log-Lipschitz distributed.
    Then
    \[
        \E[\abs{\mathcal S \cap \mathcal{M}^+}] \geq
        0.9 \E[\abs{\mathcal S \cap \mathcal M}].
    \]
\end{theorem}
\begin{proof}
    Let $C_0^B = \{\vey \in \R^{d}: \vey^\T A_B^{-1} \geq \veo\}$ and $C_{\frac{\ln(1/0.9)}{dL}}^B = \{\vey \in \R^{d}: \vey^\T A_B^{-1} \geq \frac{\ln(1/0.9)}{dL} \ve1\}$. Let further $B \in \mathcal S$ be arbitrary.
    By the construction in the proof of \Cref{lem:goodmultipliers} we have that
    \[
    \Pr[B \in \mathcal{M}^+] = \Pr\Big[[t_1 \vecc + \vez, t_2 \vecc + \vez] \cap C^B_{\frac{\ln(1/0.9)}{dL}} \neq \emptyset\Big]
    \]
    and similarly $\Pr[B \in \mathcal{M}] = \Pr\Big[[t_1 \vecc + \vez, t_2 \vecc + \vez] \cap C^B_0 \neq \emptyset\Big]$.
    By assumption we know that $\norm{p} \leq d$.
    Plugging the values of $p$ and $m$ into \Cref{lem:goodmultipliers}, it follows that
    $\Pr[B \in \mathcal{M}^+] \geq 0.9 \Pr[B \in \mathcal{M}]$.
    Putting everything together, we find
    \begin{align*}
        \E[\abs{\mathcal S \cap \mathcal{M}^+}]
        = \sum_{B \in \mathcal S} \Pr[B \in \mathcal M^+]
        \geq \sum_{B \in \mathcal S} 0.9 \Pr[B \in \mathcal M]
        &= 0.9 \E[\abs{\mathcal S \cap \mathcal M}]
    \end{align*}
    as required.
\end{proof}

\subsection{Bases with Both}\label{sec:shadowpath}

In this subsection we first combine the good slack and good multiplier properties which we have investigated during the last two subsections. Afterwards we use this result for deducing our final \Cref{cor:easytouse}.

\begin{lemma}\label{lem:mostbasesaregood}
    Let $t_2 > t_1 \geq 0$ and $L > 0$,
    let $A \in \R^{n \times d}$ be a matrix with rows of Euclidean norm at most $1$, let $\vecc \in \R^d$, and let $\hat\veb \in \R^n$ be $(\veb,\eta,\gamma)$-exponentially distributed.
    Assume that $\vez \in \R^d$ follows an $L$-log-Lipschitz probability distribution and that is independent from $\hat\veb$.
    We have
    \[
    \E[\abs{\mathcal M^+ \cap \mathcal{F}^+}] \geq 0.8\E[\abs{\mathcal M \cap \mathcal{F}}] - 2
    \]
\end{lemma}
\begin{proof}
    We use the principle of deferred decisions to separate the randomness over $\hat\veb$ and over $\vez$ as follows
    \[
    \E[\abs{\mathcal M^+ \cap \mathcal{F}^+}] =
    \E_{\vez}\big[\E_{\hat\veb}[\abs{\mathcal M^+ \cap \mathcal{F}^+}] \big].
    \]
    Observe that $\mathcal M^+$ does not depend on the value of $\hat\veb$.
    Using \Cref{thm:slacks} with $\mathcal S = \mathcal M^+$ we find for every value of $\vez$ that
    \[
    \E_{\hat\veb}[\abs{\mathcal M^+ \cap \mathcal{F}^+}] \geq 0.9\E_{\hat\veb}[\abs{\mathcal M^+ \cap \mathcal F}] - 2.
    \]
    We can once again reorder expectations and find
    \begin{align*}
        \E[\abs{\mathcal M^+ \cap \mathcal{F}^+}] &\geq \E_{\vez}\big[0.9\E_{\hat\veb}[\abs{\mathcal M^+ \cap \mathcal F}] - 2\big] \\
        &= 0.9\E_{\hat\veb}\big[\E_{\vez}[\abs{\mathcal M \cap \mathcal F}]\big] - 2
    \end{align*}
    For any fixed value of $\hat\veb$, we can observe that $\mathcal F$ is fixed.
    As such, with $\mathcal S =\mathcal F$ we can apply \Cref{thm:reducedcosts} to show
    \[
        \E_{\hat\veb}\big[\E_{\vez}[\abs{\mathcal M^+ \cap \mathcal F}]\big]
        \geq \E_{\hat\veb}\big[0.9 \E_{\vez}[\abs{\mathcal M^+ \cap \mathcal F}]\big]
    \]
    In total this gets us $\E[\abs{\mathcal M^+ \cap \mathcal{F}^+}] \geq 0.81 \E[\abs{\mathcal M \cap \mathcal{F}}] - 2$, which implies the lemma.
\end{proof}

To upper bound the total size of $\mathcal M \cap \mathcal F$, we will need to make reference to those bases which are both in $\mathcal M^+$ and in $\mathcal F^+$.
The argument upper bounding the length of the shadow path examines what happens when two such bases are adjacent.
Since the inclusion criteria for neighboring bases are not independent, we require a small graph theoretical lemma that we simplify and adapt from \cite{optimal_smoothed_analysis}.

\begin{definition}
    Given a graph $(V,E)$, for $v \in V$ we let $\operatorname{deg}(v)$ denote the degree of $v$.
    For subsets $S_1, S_2 \subseteq V$ we define $E[S_1,S_2] = \big\{\{u,v\} \in E : u \in S_1, v \in S_2\big\}$.
\end{definition}

\begin{lemma}\label{lem:redvsblue}
    Let $P = (V,E)$ denote a path
    and $V = R \cup B$ denote a partition of $P$.
    Then $\abs{V} \leq \abs{E[R,R]} + 2\abs{B} + 1$.
\end{lemma}
\begin{proof}
    Since $P$ is a path, a simple extension of the handshake lemma gives us
    \[2 \abs{R} \leq 2 + \sum_{v \in R}\operatorname{deg}(v) = 2 + \abs{E[R,B]} + 2\abs{E[R,R]}
    \]
    Then, since the maximum degree of a vertex in $B$ is at most $2$, we get that $\abs{E[R,B]} \leq 2\abs{B}$.
    Combining the two inequalities, we find that
    \[\abs{R} \leq \frac{2\abs{E[R,R]} + \abs{E[R,B]}+ 2}{2} \leq \abs{E[R,R]} + \abs{B} + 1.\]
    Finally, the lemma follows from the above combined with $\abs{V}=\abs{R}+\abs{B}$.
\end{proof}

\begin{lemma}\label{lem:tuples}
    Let $U$ be a fixed set and let $R \subseteq V \subseteq U$ denote random sets
    satisfying $\E[\abs{R}] \geq p \E[\abs{V}]$ for some $p > 1/2$.
    Let $P = (V,E)$ be an arbitrary path graph.
    Then we have
    \[
        \E[\abs{V}] \leq \frac{\E[\abs{E[R,R]}] + 1}{2p-1}.
    \]
\end{lemma}
\begin{proof}
Let $B = V \setminus R$.
From our assumption we know that $\E[\abs{B}] \leq (1-p)\E[\abs{V}]$.
By applying \Cref{lem:redvsblue} we learn
\[
\E[\abs{V}] \leq \E[\abs{E[R,R]}] + 2\E[\abs{B}] + 1 \leq \E[\abs{E[R,R]}] + 2(1-p)\E[\abs{V}]+1.
\]
Rearranging gives us our conclusion that
$\E[\abs{V}] \leq \frac{\E[\abs{E[R,R]}] + 1}{2p-1}$.
\end{proof}

We now have the correct concepts in place to upper bound the length of the shadow path.
We can do this in two ways: either by using the objective value gap or by using the expected mean width of the feasible region.
Both bounds will be used when proving \Cref{thm:bound_p_slack-epoch}.
\begin{lemma}\label{lem:objective-gap-bound}
    Let $\rho > 0$, $t_2 > t_1\geq 0$ and $L>0$.
    Let $A \in \R^{n \times d}$ be a matrix with rows of Euclidean norm at most $1$, let $\vecc \in \R^d$, and let $\hat\veb \in \R^n$ be $(\veb,\eta,\gamma)$-exponentially distributed.
    Let $\vez \in \R^d$ be $L$-log-Lipschitz distributed. Note that throughout $\vez$ is sampled independently of $\hat \veb$.  
    Let $\mathcal P = (B^1,\dots,B^k)$ denote the shadow path from $t_1 \vecc + \vez$ to $t_2 \vecc + \vez$ of
    \[
    \max\{\vecc^\T\vex:A\vex\leq\hat\veb\}.
    \]
    Then we have
    \begin{align*}
        \E[\abs{\mathcal{P}}] \leq 40 + \frac{2\E[\vecc^\T (\vex^{B^k} - \vex^{B^1})]}{\rho\omega(\eta,n,d)} + \frac{2dL\rho(t_2-t_1)}{\ln(1/0.9)}
    \end{align*}
\end{lemma}
\begin{proof}
In \Cref{sec:prelims} we characterized the shadow path as $\mathcal P = \mathcal M \cap \mathcal F$.
If $\E[\abs{\mathcal P}] \leq 40$ then the result is trivial, so we may assume for the remainder of this proof that $\E[\abs{\mathcal P}] > 40$.
Then $.05\E[|\mathcal{P}|] \geq 2$, so
by \Cref{lem:mostbasesaregood}, we have
\[
\E[\abs{\mathcal M^+ \cap \mathcal{F}^+}] \geq 0.8\E[\abs{\mathcal P}] - 2 \geq 0.75\E[\abs{\mathcal P}].
\]
Now define $T = \{i \in [k-1] : \{B^i, B^{i+1}\} \subseteq \mathcal M^+ \cap \mathcal{F}^+\}.$

We take the path graph $V = \mathcal{P}$ and $E = \{(B^i, B^{i+1}) : i \in [k-1]\}$.
Taking $R = \mathcal{F}^+ \cap \mathcal M^+$ makes $\abs{E[R,R]} = \abs{T}$ and we can apply \Cref{lem:tuples} with $p=0.75$ and find
\[
\E[\abs{\mathcal P}] \leq 2\E[\abs{T}] + 2.
\]

The remainder of this proof is dedicated to showing that the inequality
\begin{align}
\abs{T} &\leq \frac{\vecc^\T (\vex^{B^k} - \vex^{B^1})}{\rho\omega} + \frac{dL\rho(t_2-t_1)}{\ln(0.9)} \label{eq:earlypath}
\end{align}
always holds.
Consider an arbitrary $i \in T$.
Let $p \in B^i$ denote the unique index for which $p \notin B^{i+1}$.
Write $\vex^i = A_{B^i}^{-1} \hat\veb_{B^i}$ and $\vex^{i+1} = A_{B^{i+1}}^{-1} \hat\veb_{B^{i+1}}$.
Finally, define
\begin{align*}
t_{enter}^i = \min \{t \in [t_1,t_2] : (t\vecc + \vez)^\T A_{B^i}^{-1} \geq \veo\} \\
t_{exit}^i = \max \{t \in [t_1,t_2] : (t\vecc + \vez)^\T A_{B^i}^{-1} \geq \veo\}.
\end{align*}

Observe that $t \mapsto \big((t\vecc + \vez)^\T A_{B^i}^{-1}\big)_p$ is linear,
non-negative on $[t_{enter}^i, t_{exit}^i]$ and zero at $t=t_{exit}^i$.
It follows that $(\vecc^\T A_{B^i}^{-1})_p \leq 0$.
We make a case distinction on its value.
If $(\vecc^\T A_{B^i}^{-1})_p \leq -\rho$, then we consider the improvement in objective value along the edge and find that
\[
\vecc^\T(\vex^{i+1} - \vex^i) = (\vecc^\T A_{B^i}^{-1}) (A_{B^i}(\vex^{i+1} - \vex^i)) \geq \rho\omega,
\]
because $A_{B^i}(\vex^{i+1} - \vex^i)$ is $0$ on every coordinate except $p$, where it is at most $-\omega$.
Thus in this case we have
\begin{equation*}%\label{eq:redcostbig}
    1 \leq \frac{\vecc^\T(\vex^{i+1} - \vex^i)}{\rho\omega}.
\end{equation*}

If $(\vecc^\T A_{B^i}^{-1})_p > -\rho$, we consider the increase in the parameter $t$ at the vertex.
Because $i \in \mathcal M^+$, there exists $t_{middle}^i \in [t_{enter}^i,t_{exit}^i]$ for which we have $(t_{middle}^i \vecc + \vez)^\T A_{B^i}^{-1} \geq \frac{\ln(1/0.9)}{dL}$.
Using the fact that $(\vecc^\T A_{B^i}^{-1})_p > -\rho$, we find that
for any $\gamma > 0$ we have
\[
\Big(\big((t^i_{middle}+\gamma)\vecc + \vez\big)^\T A_{B^i}^{-1}\Big)_p
> \frac{ \ln(1/0.9)}{dL} - \gamma \rho.
\]
 For any $t \in [t_{enter}^i, t_{middle}^i+\frac{\ln(1/0.9)}{dL\rho}]$
we have $\big((t\vecc + \vez)^\T A_{B^i}^{-1}\big)_p > 0$.
By construction we pivot out of $B^i$ by relaxing constraint $p$, meaning that the multiplier on $p$ is the first to reach $0$ as $t$ increases. It follows that the entire vector of multipliers satisfies
$(t\vecc + \vez)A_{B^i}^{-1} \geq \veo$.
It follows that $t_{exit}^i > t_{middle}^i + \frac{\ln(1/0.9)}{dL\rho}$ and so
$t_{exit}^i - t_{enter}^i >  t_{exit}^i - t_{middle}^i > \frac{ \ln(1/0.9)}{dL\rho}$.
We conclude for this case that
\[
1 < \frac{dL\rho(t_{exit}^i - t_{enter}^i)}{\ln(1/0.9)}.
\]
Thus, in either case we have
\[
1 \leq \frac{\vecc^\T(\vex^{i+1} - \vex^i)}{\rho\omega} + \frac{dL\rho(t_{exit}^i - t_{enter}^i)}{\ln(1/0.9)}
\]
We finish the proof of \eqref{eq:earlypath} by summing over all $i\in T$ and telescoping the two terms:
\begin{align*}
    \abs{T} = \sum_{i \in T} 1 &\leq \sum_{i \in T}\frac{\vecc^\T(\vex^{i+1} - \vex^i)}{\rho\omega} + \frac{dL\rho(t_{exit}^i - t_{enter}^i)}{\ln(1/0.9)}\\
    &\leq \frac{\vecc^\T(\vex^{k} - \vex^1)}{\rho\omega} + \frac{dL\rho(t_2 - t_1)}{\ln(1/0.9)}
\end{align*}
The first summand is a random variable, the second is not.
The lemma's statement for the case $\E[\abs{\mathcal P}] \geq 40$ now follows as
$\E[\abs{\mathcal{P}}] \leq 2\E[\abs{T}]+2 \leq 2\frac{\E[\vecc^\T (\vex^{B^k} - \vex^{B^1})]}{\rho\omega(\eta,n,d)} + 2\frac{dL\rho(t_2-t_1)}{\ln(1/0.9)} + 2$
\end{proof}

\begin{lemma}\label{lem:mean-width-bound}
    Let $\rho > 0$, $t_2 > t_1\geq 0$ and $L>0$.
    Let $A \in \R^{n \times d}$ be a matrix with rows of Euclidean norm at most $1$, let $\vecc \in \R^d$, and let $\hat\veb \in \R^n$ be $(\veb,\eta,\gamma)$-exponentially distributed.
    Let $\vez \in \R^d$ be $L$-log-Lipschitz distributed.
    Let $\mathcal P = (B^1,\dots,B^k)$ denote the shadow path from $t_1 \vecc + \vez$ to $t_2 \vecc + \vez$ of
    \[
    \max\{\vecc^\T\vex:A\vex\leq\hat\veb\}.
    \]
    Then we have that
    \[
        \E[\abs{\mathcal P}] \leq 40 + \frac{2\E[\vez^\T (\vex^{B^1} - \vex^{B^k})]}{\rho\omega(\eta,n,d)} + \frac{2dL\rho\ln(t_2/t_1)}{\ln(1/0.9)}.
    \]
\end{lemma}
\begin{proof}
In \Cref{sec:prelims} we characterized the shadow path as $\mathcal P = \mathcal M \cap \mathcal F$.
If $\E[\abs{\mathcal P}] \leq 40$ then the result is trivial, so we may assume for the remainder of this proof that $\E[\abs{\mathcal P}] > 40$.
By \Cref{lem:mostbasesaregood}, we have
\[
\E[\abs{\mathcal M^+ \cap \mathcal{F}^+}] \geq 0.8\E[\abs{\mathcal M \cap \mathcal F}] - 1 \geq 0.75\E[\abs{\mathcal P}].
\]
Now define $T = \{i \in [k-1] : \{B^i, B^{i+1}\} \subseteq \mathcal M^+ \cap \mathcal{F}^+\}.$

We take the path graph $V = \mathcal{P}$ and $E = \{(B^i, B^{i+1}) : i \in [k-1]\}$.
Taking $R = \mathcal{F}^+ \cap \mathcal M^+$ makes $\abs{E[R,R]} = \abs{T}$ and we can apply \Cref{lem:tuples} with $p=0.75$ and find
\[
\E[\abs{\mathcal P}] \leq 2\E[\abs{T}] + 2.
\]

The biggest part of the remainder of this proof is dedicated to showing the deterministic statement
\begin{align}
\abs{T} &\leq \frac{\vez^\T (\vex^{B^k} - \vex^{B^1})}{\rho\omega} + \frac{dL\rho\ln(t_2/t_1)}{\ln(0.9)}\label{eq:latepath}
\end{align}
holds.
Consider an arbitrary $i \in T$.
Let $p \in B^i$ denote the unique index for which $p \notin B^{i+1}$.
Write $\vex^i = A_{B^i}^{-1} \hat\veb_{B^i}$ and $\vex^{i+1} = A_{B^{i+1}}^{-1} \hat\veb_{B^{i+1}}$.
Finally, define
\begin{align*}
t_{enter}^i = \min \{t \in [t_1,t_2] : (t\vecc + \vez)^\T A_{B^i}^{-1} \geq \veo\} \\
t_{exit}^i = \max \{t \in [t_1,t_2] : (t\vecc + \vez)^\T A_{B^i}^{-1} \geq \veo\}.
\end{align*}

Recall that, as before, $t \mapsto \big((t\vecc + \vez)^\T A_{B^i}^{-1}\big)_p$ is linear,
non-negative on $[t_{enter}^i, t_{exit}^i]$ and zero at $t=t_{exit}^i$.
It follows that $(\vez^\T A_{B^i}^{-1})_p \geq \veo$.
We make a case distinction on its value.
In case of $(\vez^\T A_{B^i}^{-1})_p \geq \rho$, we consider the decrease in auxiliary objective value along the edge and find that
\[
\vez^\T(\vex^{i+1} - \vex^i) = (\vez^\T A_{B^i}^{-1}) (A_{B^i}(\vex^{i+1} - \vex^i)) \leq -\rho\omega,
\]
because $A_{B^i}(\vex^{i+1} - \vex^i)$ is $0$ on every coordinate except $p$, where it is at most $-\omega$.
Thus in this case we have
\begin{equation*}%\label{eq:redcostbig}
    1 \leq \frac{\vez^\T(\vex^i - \vex^{i+1})}{\rho\omega}.
\end{equation*}
In case of $(\vez^\T A_{B^i}^{-1})_p \leq \rho$, we once again look at the parameter increase from $t_{exit}^i$ to $t_{enter}^i$. We aim to lower bound $t_{exit}^i/ t_{enter}^i$. 
Since $i \in \mathcal M^+$, there exists $t \in [t_{enter}^i,t_{exit}^i]$ be such that $(t \vecc + \vez)^\T A_{B^i}^{-1} \geq \frac{\ln(1/0.9)}{dL}$. Then, since $\left(\vez^{\T} A_{B^{i}}^{-1}\right)_{p} \leq \rho $, by assumption, $-\left(\vez^{\intercal} A_{B^{i}}^{-1}\right)_p \geq -\rho$ meaning that
\begin{align*}
    \left(\left(t \vecc + \vez - \frac{\ln(1/0.9)}{d L \rho} \vez\right)^{\T} A_{B^{i}}^{-1} \right)_{p} &= \left(\left(t \vecc + \vez)\right)^{\T} A_{B^{i}}^{-1}\right)_{p} - \left(\frac{\ln(1/0.9)}{d L \rho} \vez^{\intercal} A_{B^{i}}^{-1}\right)_{p} \\
    &\geq \frac{\ln(1/.9)}{dL} -\frac{\ln(1/.9)}{d L \rho} \rho  \\
    &= 0.
\end{align*}
By construction, since $p$ is the variable leaving the basis when pivoting from $B^{i}$, $p$ is the first coordinate to decrease to $0$ when decrementing $\vez$. In particular, $t \vecc + \vez - \frac{\ln(1/.9)}{d L \rho} \vez$ and $t \vecc + \vez$ are both in the normal cone given by the non-negative row span of $A_{B^{i}}$. That is 
\begin{equation}\label{eq:movez}
(t\vecc + \vez - \frac{\ln(1/0.9)}{dL\rho}\vez)^\T A_{B^i}^{-1} \geq \veo.
\end{equation}
By convexity of the normal cone, if $1 - \frac{\ln(1/.9)}{d L \rho} \leq 0$ then $t \vecc$ would also be in that cone, but by assumption, $B^{i}$ is not the final basis of the path. Therefore, we must also have $1 - \frac{\ln(1/.9)}{d L \rho} > 0$.

As a consequence, we may rescale by a factor $(1-\frac{\ln(1/0.9)}{dL\rho})^{-1}$ and observe that \eqref{eq:movez} is equivalent to the assertion that
$(t'\vecc + \vez)^\T A_{B^i}^{-1} \geq \veo$
for $t' = \frac{t}{1-\frac{\ln(1/0.9)}{dL\rho}}$.
This implies
\[
t_{enter}^i \leq t = t'(1-\frac{\ln(1/0.9)}{dL\rho}) \leq t_{exit}^i(1-\frac{\ln(1/0.9)}{dL\rho}),
\]
which we can rewrite to obtain
$t_{exit}^i/t_{enter}^i \geq \frac{1}{1-\frac{\ln(1/0.9)}{dL\rho}}$
and then, since $\ln(1/(1-x)) \geq x$ for $x \in [0,1)$,
\[
    \ln(t_{exit}^i/t_{enter}^i) \geq \ln(\frac{1}{1-\frac{\ln(1/0.9)}{dL\rho}}) \geq \frac{\ln(1/0.9)}{dL\rho}.
\]

This once again brings us to an upper bound for the unit
\[
    1 \leq \frac{dL\rho\ln(t_{exit}^i/t_{enter}^i)}{\ln(1/0.9)}.
\]
As for the other inequality, we telescope the sums to find
\begin{align*}
    \abs{T} = \sum_{i \in T} 1 &\leq \sum_{i \in T} \frac{\vez^\T(\vex^i - \vex^{i+1})}{\rho\omega} + \frac{dL\rho\ln(t_{exit}^i/t_{enter}^i)}{\ln(1/0.9)}\\
    &\leq \frac{\vez^\T(\vex^1 - \vex^k)}{\rho\omega} + \frac{dL\rho\ln(t_2/t_1)}{\ln(1/0.9)}.
\end{align*}
This finishes the proof of \eqref{eq:latepath}.
The first summand is a random variable, the second is not.
The lemma's statement for the case $\E[\abs{\mathcal P}] \geq 40$ now follows as
$\E[\abs{\mathcal{P}}] \leq 2\E[\abs{T}]+2 \leq 2\frac{\E[\vez^\T (\vex^{B^1} - \vex^{B^k})]}{\rho\omega(\eta,n,d)} + 2\frac{dL\rho\ln(t_2/t_1)}{\ln(1/0.9)} + 2$.
\end{proof}

In order to state our main result, we need two measurements of the feasible set. The first is the mean width:
\begin{definition}\label{def:meanwidth}
    Let $A \in \R^{n \times d}$, $\vecc \in \R^d$ and let $\hat\veb \in \R^n$ be $(\veb,\eta,\gamma)$-exponentially distributed.
    Assume that $P=\{\vex \in \R^d : A\vex \leq \hat\veb\}$ is bounded.
    For random $\vez$ distributed spherically symmetrically,
    let $\vex^1$ denote the vertex of $P$ maximizing $\vez$, or $\vex^1=\veo$ when $P=\emptyset$.
    We define the expected \emph{mean width} to be the quantity
    \[
        M = 2\E_{\hat\veb,\vez}\bigg[\frac{\vez^\T\vex^1}{\norm{\vez}}\bigg] = \frac{\E_{\hat\veb,\vez}[\vez^\T\vex^1]}{\E_{\vez}\big[\norm{\vez}\big].}
    \]
    The number $M$ will only be used when $A$ and $\hat\veb$ are clear from context.
\end{definition}

Note that the mean width will require the random objective to be spherically symmetrically distributed. There are indeed distributions which are both spherically symmetric and $L$-log-Lipschitz:
\begin{remark}
The $L$-exponential distribution (\Cref{def:lexponential}) is both spherically symmetric and $L$-log-Lipschitz.  
\end{remark}

%0 is smaller than 1, so $B^1$ is the maximizer and $B^0$ is the minimizer
\noindent The second measurement we require is for the width in the objective direction:
\begin{lemma}\label{lem:N}
    Let $A \in \R^{n \times d}$, $\vecc \in \R^d$ and let $\hat\veb \in \R^n$ be $(\veb,\eta,\gamma)$-exponentially distributed.
    Let $B^0$ be a basis for which $\vecc^\T A_{B^0}^{-1} \leq \veo$
    and $B^1$ be a basis for which $\vecc^\T A_{B^1}^{-1} \geq \veo$.
    Define $N_{B^0,B^1} = \vecc^\T(A_{B^1}^{-1}\veb_{B^1} - A_{B^0}^{-1} \veb_{B^0}) + \gamma\eta(\norm{\vecc^\T A_{B^{1}}^{-1}}_{1} + \norm{\vecc^\T A_{B^{0}}^{-1}}_1)$.
    Define $N = N(A,\veb,\vecc,\gamma,\eta) = \min_{B^0,B^1} N_{B^0,B^1}$,
    where the minimum is taken over all valid choices of $B^0, B^1$.
    
    Let $\vex^{*,1}$ and $\vex^{*,0}$ denote optimal vertex solutions to respectively
\begin{align*}
    \operatorname{maximize} \quad & \vecc^\T \vex & \text{and}&& \operatorname{minimize} \quad & \vecc^\T \vex \\
    \operatorname{subject~to} \quad &  A\vex \leq \hat \veb &&& \operatorname{subject~to} \quad &  A\vex \leq \hat \veb %\\
    %& \vex \geq \veo &&& & \vex \geq \veo.
\end{align*}
    Then $\E[\vecc^\T(\vex^{*,1}-\vex^{*,0})] \leq N$.
\end{lemma}
\begin{proof}
    We prove the inequality for any valid choice of $B^0, B^1$.
    For any realization of $\hat\veb$ we have $\vecc^\T A_{B^1}^{-1}\hat \veb_{B^1} \geq \vecc^\T \vex^{*,1}$ and $\vecc^\T A_{B^0}^{-1}\hat \veb_{B^0} \leq \vecc^\T \vex^{*,0}$ by weak duality.
    This yields
    \begin{align*}
        \E[\vecc^\T\vex^{*,1}] &\leq \E[\vecc^\T A_{B^1}^{-1}\hat \veb_{B^1}] \\
        &= \vecc^\T A_{B^1}^{-1}\E[\hat \veb_{B^1}] \\
        &= \vecc^\T A_{B^1}^{-1}(\veb_{B^1} + \gamma\eta\ve1) \\
        &\leq \vecc^\T A_{B^1}^{-1}\veb_{B^1} + \gamma\eta\norm{\vecc^\T A_{B^1}^{-1}}_1,
    \end{align*}
    where the last inequality follows because $B^1$ is optimal and hence $\vecc^\T A_{B^1}^{-1} \geq \veo$.
    Similarly we have $\E[\vecc^\T\vex^{*,0}] \geq \vecc^\T A_{B^1}^{-1}\veb_{B^1} - \gamma\eta\norm{\vecc^\T A_{B^1}^{-1}}_1$.
    Taking the difference we find $\E[\vecc^\T\vex^{*,1}] - \E[\vecc^\T\vex^{*,0}] \leq N$ as needed.
\end{proof}

\begin{theorem}
\label{thm:bound_p_slack-epoch}
Let $\eps > 0$ satisfy $\eps^{-1} \geq \frac{1}{1000 d^2 L \gamma \norm{\vecc}_2}$.
    Let $A \in \R^{n \times d}$ be a matrix with rows of Euclidean norm at most $1$, let $\vecc \in \R^d$, and let $\hat\veb \in \R^n$ be $(\veb,\eta,\gamma)$-exponentially distributed.
    Let $\vez \in \R^d$ be $L$-log-Lipschitz and spherically symmetrically distributed.
If $\hat\veb$ and $\vez$ are independent then the shadow path of 
\[
\max\{\vecc^\T\vex:A\vex\leq\hat\veb\}
\]
from a vertex maximizing $\vez$ to a vertex
maximizing $\vecc + \eps \vez$ has expected length at most
\[101+ 9\sqrt{\frac{\E[\| \vez\|] M dL}{\omega} \ln\left( \frac{dLN}{\omega \varepsilon} \right) },\]
where $\omega = \omega(\eta,n,d)$ was defined as $\omega(\eta,n,d) = \frac{\eta}{1240 d \ln(n)}$, $M$ is the expected mean width,
and $N$ is as defined in \Cref{lem:N}.
\end{theorem}
\begin{proof}
    We separate the path from $\vez$ to $\vecc + \eps \vez$ into two subpaths.
    For $0 < t' \leq \eps^{-1}$ to be determined, the first path goes from $\vez$ to $t'\vecc + \vez$ and
    the second goes from $t'\vecc + \vez$ to $\eps^{-1}\vecc + \vez$.
    Let $\vex^1$ be the vertex maximizing $\vez$, let $\vex^2$ be the vertex maximizing $t'\vecc + \vez$, let $\vex^3$ be the vertex maximizing $\eps^{-1}\vecc + \vez$, and finally let $\vex^{\ast}$ be the vertex on the shadow path maximizing $\vecc$. 
    
    By applying \Cref{lem:objective-gap-bound} with $t_0=0, t_1=t'$ and choosing $\rho = \frac{\ln(1/0.9)}{dLt'}$, we know that the first path has expected length at most
    \begin{align*}
    40 + 2\frac{\E[\vecc^\T(\vex^2 - \vex^1)]\cdot dLt'}{\ln(1/0.9)\omega}
     + 2 &\leq
    42 + 19\frac{\E[\vecc^\T(\vex^2 - \vex^1)] \cdot dLt'}{\omega} \\
    &\leq
    42 + 19\frac{\E[\vecc^\T(\vex^* - \vex^1)] \cdot dLt'}{\omega} \\
    & \leq 42 + 19 \frac{NdLt'}{\omega},
    \end{align*}
    where the last inequality followed from \Cref{lem:N}.
    We choose $t' = \frac{ \omega}{dLN}$ to make sure that the first sub-path has expected length at most $42 + 19 = 61$.
    
    Note that $t' = \frac{ \omega}{dLN} = \frac{\eta}{1240 d^2LN\ln(n)} \leq \frac{\eta}{1000 d^2L \cdot \gamma\eta\norm{\vecc^\T A_{B^1}^{-1}}_1}$,
    where $B^1$ is the optimal feasible basis for maximizing $\vecc^\T\vex$.
    The rows of $A$ reach have 2-norm at most $1$,
    which makes that $\norm{\vecc^\T A_{B^1}^{-1}}_1 \geq \norm{\vecc}_2$ by the triangle inequality.
    Hence $t' \leq \frac{1}{1000d^2 L \gamma \norm{\vecc}_2} \leq \eps^{-1}$ as promised.

   Next, note that the shadow path from $\vez$ to $\vecc$ is monotonically decreasing in $\vez$.
   In particular, since $\vex^{\ast}$ is at the end of the shadow path and $\vex^{1}$ is at the start, 
   \[\vez^{\T} \vex^{1} \geq \vez^{\T} \vex^{2} \geq \vez^{\T} \vex^{3} \geq \vez^{\T} \vex^{\ast}.\]
   In particular this gives $0 \leq \E[\vez^\T(\vex^2 - \vex^3)] \leq \E[\vez^\T(\vex^1-\vex^\ast)]$.
   For the second sub-path, which goes from $t'\vecc + \vez$ to $\eps^{-1}\vecc + \vez$,
    we use \Cref{lem:mean-width-bound} and find it to have expected length at most
    \[
    40 + \frac{2\E[\vez^\T (\vex^2 -\vex^3)]}{\rho\omega} + \frac{2dL\rho\ln(1/(t'\eps))}{\ln(1/0.9)} \leq
    40 + 2\frac{\E[\vez^\T \vex^1] - \E[\vez^\T \vex^*]}{\rho\omega} + \frac{2dL\rho\ln(1/(t'\eps))}{\ln(1/0.9)}.
    \]
    Note that $\vex^*$ is independent of $\vez$.
    Since $\vez$ is distributed symmetrically around the origin that implies
    $\E[\vez^\T \vex^*] = 0$.
    The vertex $\vex^1$ was defined as the maximizer of $\vez$, meaning that
    $2 \E[\vez^\T \vex^1] = \E[\norm{\vez}] \cdot M$
    This means that the second sub-path has expected length at most
    \[
    40 + \frac{\E\big[\norm{\vez}\big]\cdot M}{\rho\omega} + \frac{2dL\rho\ln\big(\frac{dLN}{\omega\eps}\big)}{\ln(1/0.9)}.
    \]
    Now we choose $\rho =\sqrt{\frac{2\E\big[\norm{\vez}\big] \cdot M \ln(1/.9)}{\omega dL \ln\left(\frac{dLN}{\omega \eps}\right)}}$
    to minimize this sum. Then the resulting minimizer is 
    \[40 + 2\sqrt{\frac{2\E[\| \vez\|] M dL \ln\left(\frac{dLN}{\omega \varepsilon} \right)}{ \omega \ln(1/.9)} } \leq 40 + 9\sqrt{\frac{\E[\| \vez\|] M dL \ln\left(\frac{dLN}{\omega \varepsilon} \right)}{ \omega } }.\]
    Adding this to the expected length $61$ of the first subpath yields the desired result.
\end{proof}

\Cref{thm:bound_p_slack-epoch} is the main technical version of our shadow bound.
We will derive two corollaries from it that will be easier to use.

\begin{theorem}\label{cor:easytouse}
Let $\eps > 0$ satisfy $\eps^{-1} \geq \frac{1}{1000 d^3 \gamma \norm{\vecc}_2}$.
Let $A \in \R^{n \times d}$ be a matrix with rows of Euclidean norm at most $1$, let $\vecc \in \R^d$, and let $\hat\veb \in \R^n$ be $(\veb,\eta,\gamma)$-exponentially distributed.
Let $\vetheta \in \mathbb{S}^d$ be distributed uniformly at random, independent from $\hat\veb$.
If the LP
\[
\max\{\vecc^\T\vex:A\vex\leq\hat\veb\}
\]
has expected mean width $M$, then the shadow path from a vertex maximizing $\vetheta$ to a vertex maximizing $\vecc + \varepsilon \vetheta$ has expected length at most
\[102 + 316\sqrt{\frac{M d^3 \ln(n)}{\eta} \ln\left(\frac{6742 d^{3} N \ln^2(n)}{\eta \varepsilon} \right)},\]
where $N$ is as defined in \Cref{lem:N}.
\end{theorem}
\begin{proof}
Let $\vez$ be an $L$-exponentially distributed random vector,
sampled in such a manner that $\vetheta = \frac{\vez}{\norm{\vez}}$.
Instead of measuring the path from $\vetheta$ to $\vecc + \eps\vetheta$, we instead consider the path from $\vez$ to $\vecc + \eps
\frac{1}{2edL^{-1}\ln(n)}\vez$.
By Theorem \ref{thm:bound_p_slack-epoch}, the expected number of steps on this latter path is at most
\[101 +  9 \sqrt{\frac{\E[\| \vez\|] M dL}{\omega} \ln\left( \frac{dLN (2 e d L^{-1} \ln(n))}{\omega \varepsilon} \right) } = 101 + 9d \sqrt{\frac{\E[\| \vez\|] M dL}{\omega} \ln\left(\frac{2ed^{2}N \ln(n)}{\omega \varepsilon} \right)},\]
where we fill in $\omega = \frac{\eta}{1240 d \ln(n)}$ as specified in \Cref{def:feasible-gap-bases}.
    By \Cref{fact:norm_exp_distr}, we have
    $\E\big[\norm{\vez}\big] \cdot L = d$ which yields the complete bound of    
\[101 + 9\sqrt{\frac{1240 M d^3 \ln(n)}{\eta} \ln\left(\frac{2480ed^{3}N \ln^2(n)}{\eta \varepsilon} \right)} \leq 101 + 316\sqrt{\frac{M d^3 \ln(n)}{\eta} \ln\left(\frac{6742 d^{3}N\ln^2(n)}{\eta \varepsilon} \right)}.\]

Now consider the expected number of bases which are on the former path ($\vetheta$ to $\vecc + \eps\vetheta$) but not on the latter path ($\vetheta$ to $\vecc + \eps\frac{1}{2edL^{-1}\ln(n)}\vez$).
If there are any such bases, then there are at most $\binom{n}{d} \leq n^d$ many.
However, such bases exist only when $\norm{\vez} \geq 2edL^{-1}\ln(n)$.
By Lemma \ref{lem:tailbound}, we have $\Pr[\norm{\vez} \geq 2 e d L^{-1} \ln(n)] \leq n^{-d}$.
Hence the expected number of bases on the former path but not on the latter is at most $1$.
This finishes the proof.
\end{proof}

Finally, we require a variant of \Cref{cor:easytouse} which allows for two non-perturbed variable bounds.
This will allow us to use the semi-random shadow vertex method on LP's of the form
\begin{align}
    \operatorname{maximize} \quad &\vecc^\T \vex \nonumber \\
    \operatorname{subject~to} \quad & A \vex \leq \hat \veb \nonumber \\
    &0 \leq \vex_d \leq U \tag{One Fixed Bound} \label{eq:oneboundfix}
\end{align} 
with $\vecc \in \{\vece_d, -\vece_d\}$.

\begin{theorem}\label{thm:TODO}
Assume $n > d \geq 2$.
Let $\eps > 0$ satisfy $\eps^{-1} \geq \frac{1}{1000 d^3 \gamma \norm{\vecc}_2}$.
Let $A \in \R^{n \times d}$ be a matrix with rows of Euclidean norm at most $1$, let $\vecc \in \{\vece_d, -\vece_d\}$, $U > 0$ and let $\hat\veb \in \R^n$ be $(\veb,\eta,\gamma)$-exponentially distributed.
Let $\bar\vetheta \in \sfe$ be distributed uniformly at random, independent from $\hat\veb$.
Let $(-\hat L, \hat U)$ be $((0,U),\eta,\gamma)$-exponentially distributed, again independent of what came before.
If the feasible set of
\begin{align}
    \operatorname{maximize} \quad &\vecc^\T \vex \nonumber \\
    \operatorname{subject~to} \quad & A \vex \leq \hat \veb \nonumber \\
    &\hat L \leq \vex_d \leq \hat U \label{eq:oneboundperturb}
\end{align} 
has expected mean width $M$, then the shadow path from a vertex maximizing $\bar\vetheta$ to a vertex maximizing $\vecc + \varepsilon \bar\vetheta$ on the feasible set \eqref{eq:oneboundfix} has expected length at most
\[113 + 344d \sqrt{\frac{d \ln(n+2) M}{\eta} \ln\left(\frac{6742d^{3}(U + 2\eta\gamma)\ln^2 (n+2)}{\eta \varepsilon} \right)}.\]
\end{theorem}
\begin{proof}
    Since $\vecc = \vece_d$ or $\vecc = -\vece_d$, the shadow path will leave or enter the faces of the feasible set with $\vex_d=0$ or $\vex_d=U$ at most once each.
    For all other pivot steps, the bases visited are also bases of the relaxed feasible set $\{\vex : A \vex \leq \hat\veb\}$.
    Whenever $\hat 0 \leq 0$ and $U \leq \hat U$, these same bases are also feasible bases on the same shadow path on \eqref{eq:oneboundperturb}.
    By \Cref{lem:perturbationsbounded} this condition happens with probability at least $1-\frac{2}{(n+d)^2}$.
    As such, the expected number of pivot steps on the shadow path from $\bar\vetheta$ to $\vecc + \varepsilon \bar\vetheta$ for this new LP gives an upper bound on the length of the path we seek to bound.
    Using \Cref{cor:easytouse} we find that we can expect at most
    \[
        102 + 316\sqrt{\frac{M d^3 \ln(n+2)}{\eta} \ln\left(\frac{6742 d^{3} N \ln^2(n+2)}{\eta \varepsilon} \right)}
    \]
    bases on the shadow path from $\bar\vetheta$ to $\vecc + \varepsilon \bar \vetheta$ on the feasible set of \eqref{eq:oneboundperturb}.
    Resampling $\hat 0$ and $\hat U$ until $\hat 0 \leq 0$ and $\hat U \leq U$ are satisfied requires an average of at most
    $\frac{1}{1-\frac{2}{(n+d)^2}}$ samples, and Wald's equation tells us that
    our final path has expected length no more than
    \[
        2+\frac{1}{1-\frac{2}{(n+d)^2}}\left(102 + 316\sqrt{\frac{M d^3 \ln(n+2)}{\eta} \ln\left(\frac{6742 d^{3} N \ln^2(n+2)}{\eta \varepsilon} \right)}\right).
    \]
    The extra $2$ is from the pivots to and from the faces $\vex_d=0$ and $\vex_d=U$.
    We compute new constants using $n > d \geq 2$.
    
    The value of $N$ can be computed exactly. Any feasible basis maximizing $\vece_d$ has objective value $U$ on the unperturbed problem and contains the inequality $s \leq \hat U$, which will be the sole constraint to have a non-zero Lagrange multiplier and that multiplier is $1$.
    Similarly for any feasible basis minimizing $\vece_d$ over the unperturbed problem having objective value $0$ and a multiplier of $1$ on the basic constraint $-s \leq 0$.
    Thus $N = U + 2\gamma\eta$.
\end{proof}

\section{Discussion}\label{sec:scorecard}

In this section we evaluate how well the present analysis comports with further observation and how well our overall methodological structure ultimately comported with the general by-the-book analysis framework.  Of course, in any mathematical analysis compromises are inevitable in order to accommodate the analyst's ability to produce proofs.  Since the impetus for performing a by-the-book analysis (and for the by-the-book analysis framework itself) is to produce theoretical results which have greater capacity to explain what is observed in practice, we propose that it is essential to a by-the-book analysis to evaluate how successful the analysis was in comporting with its observations.  This is in part to provide a realistic assessment of the explanatory power of our results, but for this work, in which we propose the by-the-book analysis framework, it is also to differentiate the strengths and weaknesses of the by-the-book framework as a general algorithm analysis framework from the strengths and weaknesses of the particular by-the-book analysis we performed.

\subsection{Input Assumptions}

Our bound is written as polynomial in:
\[d,n,\eta,M, \ln(N), \ln(\kappa), \feastol, \log(\frac{1}{\opttol}).\]
The standard input size measures for any linear program are $d$ and $n$, the number of variables and number of inequalities respectively. For weakly polynomial run-time bounds, $\ln(\kappa)$ and $\ln(N)$ are also standard to include as part of the measure of the input-size as they are bounded by the bit-complexity of the linear program.
The remaining parameters $\eta$, $\text{optTol}$, and $\text{feasTol}$ are all directly related to our modeling assumptions.  We contend that, as supported by our research and computational experiments, it is reasonable to assume the chosen bounds on these numbers in a real-world computational setting. 
Likewise, our assumption that the rows of the input constraint matrix have Euclidean norm at most 1 is reasonable and unobtrusive. This is due to the fact that common LPs are hyper sparse, often containing only a constant number of non-zero elements in every inequality constraint, and each of these entries being bounded in magnitude. The mean width $M$ was already discussed in detail in \Cref{par:MIPLIB}.

\subsection{Model Assumptions}

Our analysis has strong requirements on the pivot rule and ratio test, needing variants that are considered inefficient in practice and needing both to be performed exactly, without any loss of numerical precision.

We require that bounds are perturbed once and never shifted again.
In particular, this requires a numerically exact implementation of Dantzig's ratio test, ruling out the possibility of using Harris' ratio test.
Since any state of the art implementation of the simplex method benefits from Harris' ratio test \cite[p.~179]{marosbook}, this is a notable limitation of the current analysis.

We use a priori perturbations of the bounds and right-hand side. In our observations, we found that HiGHS does a priori perturbations not in the primal simplex method, but only to the cost vector in the dual simplex method.
A more realistic model of the simplex method would allow for pivot steps which are slightly infeasible, and would allow to add and subtract perturbations throughout the running time.
%In these respects, our choice of a priori perturbations hopes to \textit{model} the effect of the online purturbations performed in practice, but does not constitute an accurate one-to-one description of what is done in practice.

When perturbing, our analysis requires a particular choice of probability density function for the individual perturbations.
This is different from current leading implementations, which sample uniformly from an interval.
We do not currently know how our probability distribution compares to the uniform distribution in practice.
As such, we can not evaluate whether this is an unrealistic modeling choice or a possible improvement over the state-of-the-art.

For the pivot rule, we require that the semi-random shadow vertex path is followed. For our analysis, the randomness in the perturbations and pivot rule should furthermore be independent of each other and of the input data.
This rules out a wide range of sophisticated pricing strategies, including steepest edge pivoting, partial pricing and multiple pricing (see \cite[p.~112-114]{orchardhaysbook} and \cite[p.~186-188]{marosbook}), as well as long-step pivoting rules used in Phase I of the simplex method \cite[p.~116-117]{orchardhaysbook}.
In the dual, it rules out the bound-flipping ratio test.

The shadow vertex rule has been criticized for the difficulty of following it with appropriate pivot tolerances in a numerical setting \cite[p.~130]{orchardhaysbook}.
As such, it does not see much use. The pivot rule used most commonly in practice is the steepest-edge rule (enhanced with partial pricing in the primal or bound-flipping ratio test in the dual).
An analysis similar to the one presented here which uses steepest-edge would be highly desirable.  Unfortunately, theoretical analysis is currently hamstrung by a broad lack of theoretical tools for analyzing the progress of steepest-edge outside of bounding step-by-step optimality gap reduction, which is not sufficient for proving strong upper bounds on the running time since any one pivot may reduce the optimality gap by an arbitrarily small amount.  General improvements in our ability to theoretically analyze steepest-edge would open up the possibility of improved by-the-book analyses of simplex.

There are various other techniques in active use among high-quality simplex implementations that we did not bring into consideration for this paper.
Most of these pertain to sparsity, numerical error, and cheap update formulas for weights or factorizations.

Automatic row/column rescaling, while often explained as being performed for the purpose of numerical stability, is likely to be an important component to include in future models.
Rescaling affects the mean width of the feasible set, impacts the pivoting choices made by steepest-edge and semi-random shadow vertex rules, and influences the geometry of the feasible set when tolerances are accepted.

%Finally, a typical simplex method found in practice alternates between the primal and dual, while our analysis only pertains to a single run of the primal, and our Phase I procedure is not one commonly used in implementations. For all these above reasons, there is space to qualitatively improve over the analysis presented in this paper.

\subsection{Methodology} Although we broadly succeeded in following the proposed methodology of by-the-book analysis, we did lightly deviate from this structure during the third phase.  We originally planned to let an assumption of bounded geometric diameter act as a model and proxy for the input being well-scaled.  It was only after beginning mathematical analysis in the third phase that we discovered the mean width may be a better parameter: it would yield stronger bounds, and moreover is much easier to measure. In one sense this is in conflict with the principle that our assumptions should be informed solely by a priori observations, and not retroactively informed by the proofs.  

On the other hand, we believe that it is natural and even inevitable to discover over the course of the analysis that early choices of assumptions for modeling one's observations can be improved by better ones.  In this case, we did not so much introduce a new assumption as refine the parameter we were using to model an existing assumption (in this case, the assumption that our LP was well-scaled in some concrete way). Moreover, in the spirit of by-the-book analysis, we were only content to allow our bounds to depend on $M$ provided we could verify that it is bounded on real instances.
This is actually a strength of our shift from geometric diameter to mean width: the mean width is more easily measured (or rather, accurately estimated) than the geometric diameter.
However, it remains to be seen whether the simplex method's performance truly has a meaningful dependence on mean width; as we note in \Cref{sec:discussion}, this requires further investigation.

\subsection{Outcomes} Although our results are a significant step forward in the study of the simplex method, they do not constitute a complete quantitative explanation of its real-world performance.
We can evaluate the analysis and the theorem statement both quantitatively and qualitatively.

\paragraph{Scaling with problem size}

Our bound contains a large constant factor dependence, much bigger than the scaling with problem size observed in practice.
In the course of this work, only limited effort was expended to reduce this constant factor.

Our running time bound scales faster with the problem size $d$ than the nearly linear growth observed in \cite{randibmmanual, dantzigbook,sha87, andrei2004complexity,makhoringlpk, xpress}.
At the moment of writing it is unclear which losses happen due to inefficiencies in the proof, and which happen due to insufficient accounting of the geometry of the feasible set.

\paragraph{Size of perturbations.}
Larger and more aggressive perturbations are known to improve performance \cite{BixbySurvey}.
A common technique in solvers, used when the simplex method stalls with repeated zero-length pivot steps, is to add increasingly large perturbations until stalling resolves.
Although these perturbations do help the simplex method make progress, it comes at a cost: the perturbations must later be removed in order to return a solution that is feasible up to the desired tolerance.

The running time bounds in our theorems contain a one-over-square-root dependence on the feasibility tolerance.
Interpreted naively, this would suggest that changing the feasibility tolerance by a factor $100$ would change the running time by as much as a factor $10$.
To the best of our knowledge, no effect of such magnitude has been reported.
The likely source of this disparity is that, in our analysis, step length is governed only by the perturbations: we did not use the geometry of the feasible set.

In order to improve the running time bounds, and to have a bound which does not scale as strongly with the tolerance, taking into account this geometry is necessary.
Otherwise, a running time dependence of $\sqrt{M/\eta}$ is tight, matching a lower bound construction up to polynomial factors in $d$ and poly-logarithmic factors in $M/\eta$.

This lower bound is obtained by taking $n=(c/\eta)^{d/2}$ constraint vectors, for some absolute constant $c > 1$, roughly equally spaced on the unit sphere.
With a right-hand side equal to $\ve1$, the feasible set will approximate a Euclidean ball.
Any simplex method navigating from the minimizer to the maximizer of a given objective will require at least $\Omega(\eta^{-1/2})$ pivot steps, provided that every solution visited satisfies the relaxed constraints $A \vex \leq (1+\eta)\ve1$ and lies outside the interior of $A\vex \leq \ve1$.
This holds even for $d=2$.
For more on such constructions and analyses, see \cite{optimal_smoothed_analysis, bdghl21, DGGT16}.

\paragraph{Length and curvature}
In our mathematical analysis, we prove that vertices visited by the semi-random shadow vertex method tend to take up space in the primal, by having at least one long shadow edge incident, or take up space in the dual, by creating curvature of the shadow polygon.
Furthermore, the analysis suggests that this effect is strongest when close to the maximizer of the random objective, and gets weaker clo
Without randomizing the shadow plane with respect to the data, the theory would not predict this effect to happen.

This prediction aligns with findings from Fischer. In his Diplomarbeit \cite{fischer}, he used the shadow vertex method to draw projections of the Netlib LP problems.
On the horizontal axis he plotted the value of the objective $\vecc^T \vex$.
For the vertical axis he either used fixed objectives $\vex_0, \vex_1,$ or $\vea_i^\T \vex$, or random objectives $\vez^\T \vex$ with $\vez \in [-1,1]^d$ uniformly random.
Four consecutive shadow paths were followed $\vecc \to \vez \to -\vecc \to -\vez \to \vecc$ to find the complete shadow image.
In cases where the shadow is unbounded, only the vertices are drawn.
Upon request, we obtained the images from G. M. Ziegler.

\begin{figure}
    \centering
    \begin{subfigure}{0.45\textwidth}
    \includegraphics[width=\textwidth]{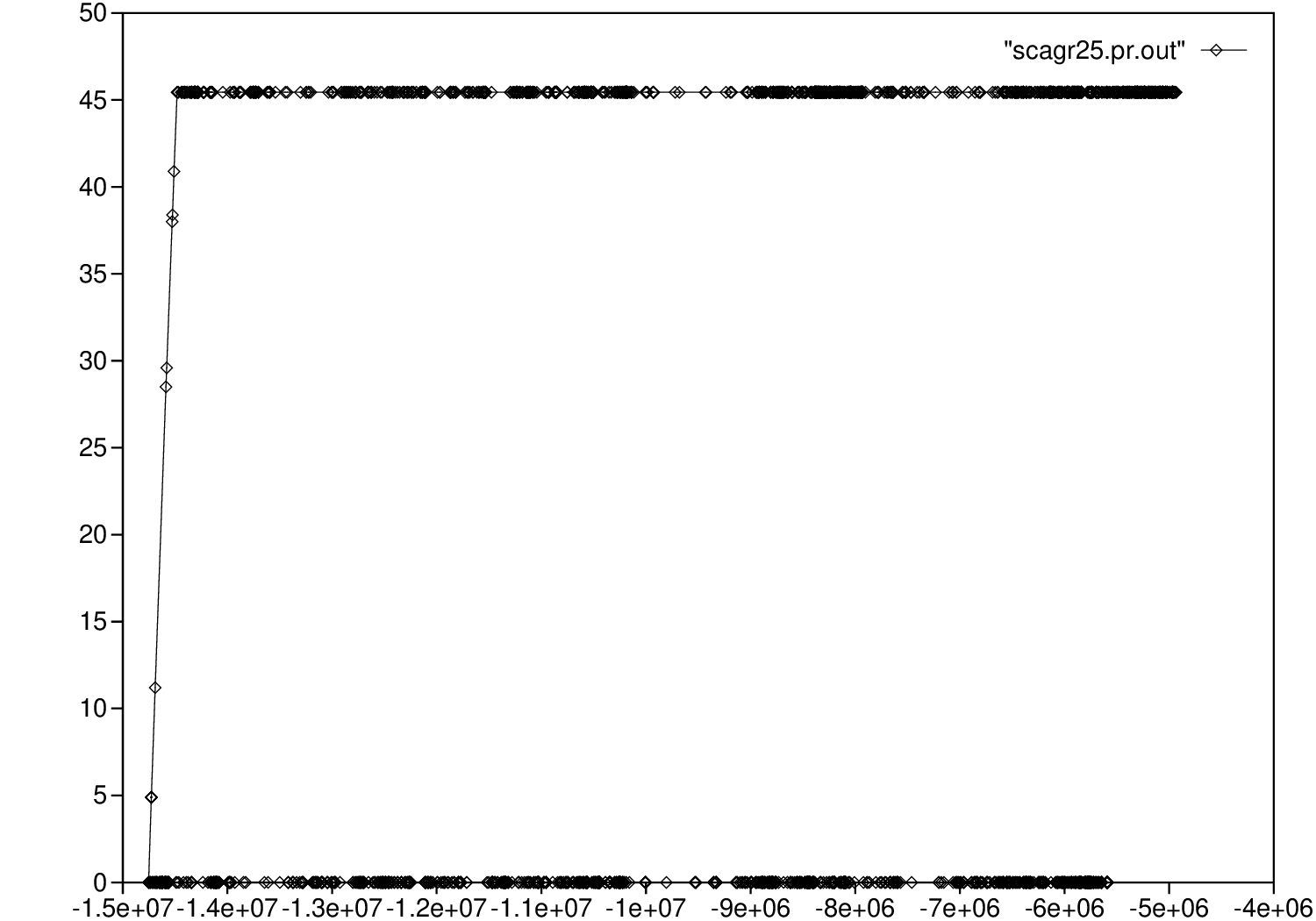}
    \caption{$\vex_0$ for vertical axis}
    \end{subfigure}
    \begin{subfigure}{0.45\textwidth}
    \includegraphics[width=\textwidth]{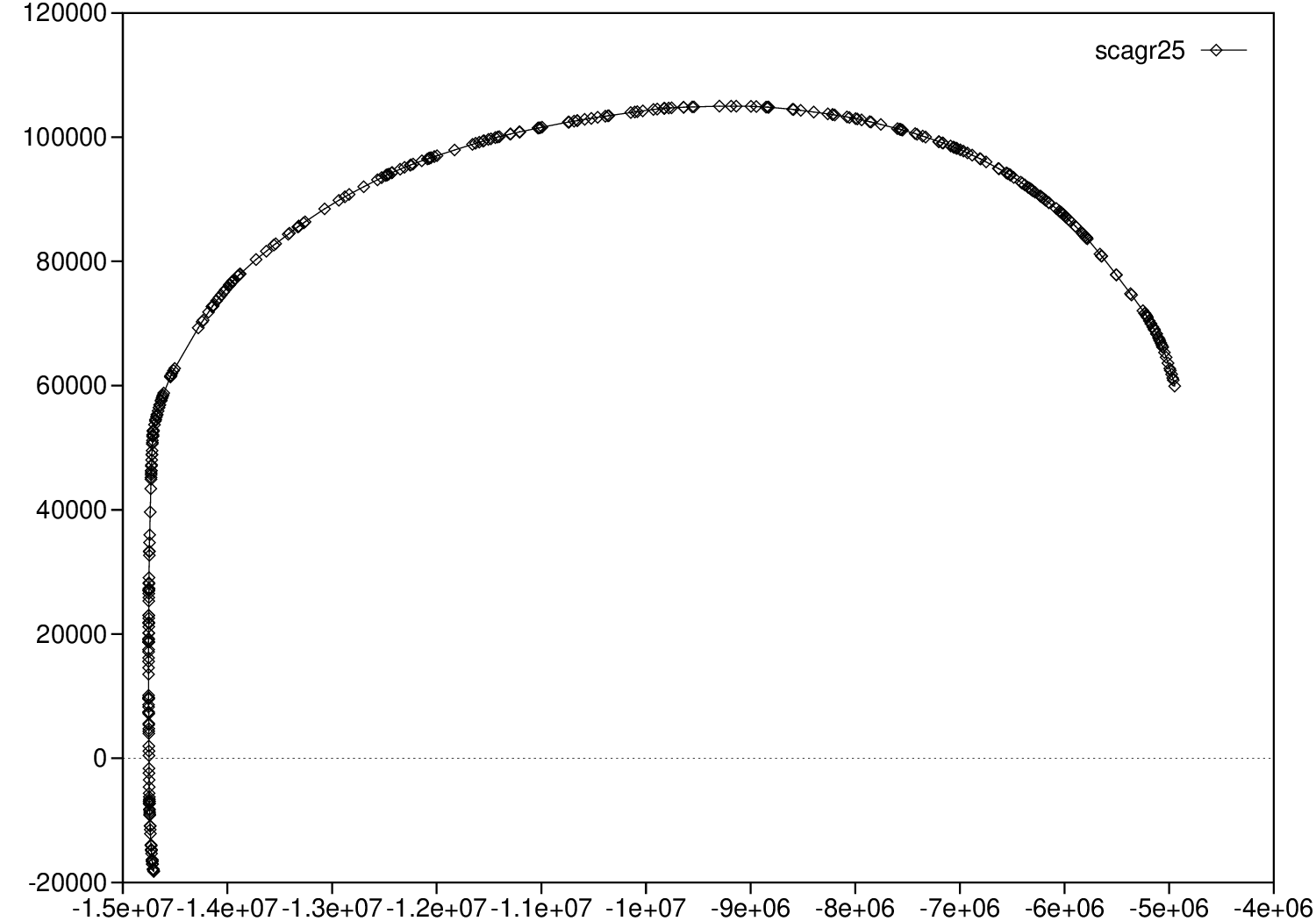}
    \caption{Random objective for vertical axis}
    \end{subfigure}
    \caption{Two-dimensional projections of \texttt{scagr25}, LP objective for horizontal axis \cite{fischer}.}
    \label{fig:fischer-scagr25}
\end{figure}

For most instances we received only one plot for a single auxiliary objective vector.
However for the instance \texttt{scagr25} Fischer drew two plots: one with a deterministic shadow $(\vex_0, \vecc^\T \vex)$ and one with a semi-random shadow $(\vez^\T \vex, \vecc^\T \vex)$.
These can be seen in \Cref{fig:fischer-scagr25} which depict the two projections.
%For the deterministic shadow, we observe only a weak correspondence between the shape of the shadow and the location of the shadow vertices.
For the semi-random shadow in \Cref{fig:fischer-scagr25} (b), we see a strong correspondence: when a given part of the boundary with constant length has more vertices (and thus, shorter edges), then it is more curved (taking up dual space).
For the deterministic shadow in \Cref{fig:fischer-scagr25} (a), we do not see the same correspondence.
At least on the qualitative level, Fischer's plots agree with the prediction from our analysis.

\section{Conclusion}\label{sec:discussion}
In this paper we introduced the framework of by-the-book analysis and used it to study the performance of the simplex method.
We have presented a running time analysis for a simplex method with bound perturbations, obtaining an upper bound on the expected running time of 
\[O\left(d^{1.5}\ln(n) \sqrt{\frac{ M}{\feastol} \ln\left(\frac{d\kappa N\ln(n)}{\feastol \cdot \opttol} \right)}\right)\]
pivot steps required to solve an LP up to given feasibility and optimality tolerance.
This result demonstrates how incorporating bound perturbations can provably ensure the simplex method requires only few pivot steps.

Although this work is predominantly dedicated to introducing and performing a by-the-book analysis of the simplex method, we note that this bound on the expected running time holds for arbitrary inputs and is mathematically sound even outside of the context of by-the-book analysis.
We highlight the two most significant improvements over earlier work.

First, from a purely mathematical lens, our analysis works for sparse constraint data, in contrast to smoothed and average-case analysis.
Since all practical LPs are sparse (see, e.g., \cite[p.~64]{orchardhaysbook}, \cite[p.~50]{marosbook} or \cite{gay1985electronic,miplib}), this is an important feature.
Second, due to the fact that we followed the by-the-book framework proposed herein, the assumptions we make are grounded in observations from simplex method implementations, LP modeling best practices, and measurements from practical benchmark instances. 
Our running time primarily depends on the LP being well scaled and the simplex method featuring bound perturbations.
Both of these properties are essential to any large-scale LP solvers \cite[p.~110,~241]{marosbook}.
In this manner, the results of this by-the-book analysis correspond well with established knowledge on the simplex method as it is used in practice. This is an advantage of our bound, even when it is interpreted purely as a highly parameterized theoretical result; the parameters themselves were chosen not out of mathematical convenience but due to the fact that our research and experiments indicate that these parameters are bounded in practice, and may indeed play a genuine role in the behavior of the simplex method.

The exact role of scaling for linear programs has previously not been well understood, but is known to reduce the number of pivot steps \cite[p.~110-111]{marosbook} \cite[p.~98]{panbook}.
There is no broad agreement on how to define when an LP is well scaled \cite{Tomlin1975}.
We take our theorems and experiments to indicate that the mean width of the feasible set is bounded for ``well-scaled'' instances.
A more in-depth investigation is in order to substantiate or refute that belief.
Are the instances with small mean width indeed the ones that would be subjectively judged as well-scaled?
Are instances with large mean width typically harder to solve for a simplex method?
These are only some of the questions that need to be answered in order to confidently speak about whether small mean width is indeed a desirable property of linear programs.

Despite the remaining limitations highlighted in \Cref{sec:scorecard}, we believe that the present work is a significant step forward towards a theoretical understanding of the simplex method's real-world performance.
Moreover, we expect our proposed framework of by-the-book analysis to be applicable to many more algorithms.
Any algorithm could be the subject of a by-the-book analysis, provided that there is enough computational experience and/or high-quality open-source code available. A by-the-book analysis may be especially desirable in situations where existing theoretical results struggle to match the known practical efficiency of an algorithm or require unrealistic assumptions.

\section{Acknowledgements}
This work is supported by ANR JCJC grant ANR-24-CE48-2762.
We thank François Lamothe for feedback on the initial manuscript and for helpful discussions about Phase 1 of the simplex algorithm.
Our deepest gratitude goes out to the many developers who made time to explain us the inner workings of their simplex codes, including from COPT, FICO, Google, Gurobi, Hexaly, HiGHS, MOSEK, OMP, and SCIP, as well as Laurent Poirrier.
Any flaws in our descriptions are ours alone.

% 7: HiGHS, OMP, Google, Porrier, FICO, SoPlex, Gurobi. The other 3 we spoke to: MOSEK, Hexaly and COPT were interviewed only later in the process

\appendix

\section{Additional Proofs}\label{app:additionalproofs}

The following lemma from \cite{optimal_smoothed_analysis} is reproduced to make the present document self-contained.
\begin{lemma}[General condition-reversing interval lemma]\label{lem:general-condition-reversing}
    Let $L > 0$ be arbitrary.
    Suppose $s \in \R$ is a continuous random variable whose probability density function $f : \R \to \R$ satisfies the following limited log-Lipschitz property:
    For every $x_1, x_2 \in [t-2/L,t+2/L]$ we have $f(x_1)/f(x_2) \leq e^{L\abs{x_1-x_2}}$.
    Then for any $\eps \in [0,1/L]$ we have
    \[
        \Pr[s \geq t-\eps \mid s \leq t] \leq 31 \eps L \cdot \Pr[s \geq t].
    \]
\end{lemma}
\begin{proof}
    We start by proving that
    $\Pr\big[s \in [t-\eps, t]\big] \leq e^2 \eps L \cdot \Pr\big[s \in [t, t+1/L]\big]$.
    In the final paragraphs of this proof we will extend this statement into the desired conclusion.

    Since $\eps \leq 1/L$ we are within the log-Lipschitzness range to
    bound our intended left-hand side as
    \[
        \Pr\big[s \in [t-\eps, t]\big]
        = \int_{t-\eps}^t f(x) \dd x
        \leq \int_{t-\eps}^t f(t) e^{L\abs{x-t}} \dd x
        \leq e \eps f(t).
    \]
    Similarly, we may use this log-Lipschitzness property to lower bound the probability in our intended right-hand side and find
    \begin{align*}
        \Pr\big[ s \in [t, t + 1/L]\big]
                      &= \int_t^{t + 1/L} f(x) \dd x \\
                      &\geq \int_t^{t + 1/L} f(t) e^{-L\cdot\abs{x-t}} \dd x \\
                      &\geq e^{-1} f(t) L^{-1}.
    \end{align*}
    Putting these two inequalities together, we find
    \begin{equation}\label{eq:numeratorbound}
        \Pr\big[s \in [t-\eps, t]\big] \leq e^2 \eps L \Pr\big[s \in [t,t+1/L] \big].
    \end{equation}
    This is the initial statement mentioned at the start of this proof.
    We now define the affine transformation $T : \R \to \R$ to satisfy
    $T(t-2/L) = t-2/L$ and $T(t+1/L) = t$.
    For this transformation we observe that the Jacobian is $\frac{2}{3}$,
    and for any point $y \in [t-2/L, t+1/L]$ we have $\abs{y-T(y)} \leq 1/L$.
    We use the transformation to give a change of variables and find
    \begin{align}
        \Pr[s \leq t] &= \int_{-\infty}^t f(x) \dd x \nonumber\\
                      &= \int_{-\infty}^{t-2/L} f(x) \dd x + \int_{t-2/L}^t f(x) \dd x \nonumber\\
                      &= \int_{-\infty}^{t-2/L} f(x) \dd x +
                      \frac{2}{3}\int_{t-2/L}^{t+1/L} f( T(y)   ) \dd y \nonumber\\
                      &\geq \int_{-\infty}^{t-2/L} f(x) \dd x + \frac{2}{3e}\int_{t-2/L}^{t+1/L} f(y) \dd y \nonumber\\
                      &\geq \frac{2}{3e} \int_{-\infty}^{t+1/L} f(y) \dd y \nonumber\\
                      &= \frac{2}{3e} \Pr[s \leq t+1/L]\nonumber\\
                      &\geq \frac{2}{3e}\Pr[s\leq t+1/L].\label{eq:denominatorbound}
    \end{align}
    Using the above two inequalities in order to bound the numerator and the denominator, we
    can now prove the lemma as follows
    \begin{align*}
        \Pr[s \geq t-\eps \mid s \leq t] &= \frac{\Pr\big[s \in [t-\eps, t]\big]}{\Pr[s \leq t]} \\
                                         &\leq \frac{3e\Pr\big[s \in [t-\eps, t]\big]}{2\Pr[s \leq t+1/L]} \tag{by \eqref{eq:denominatorbound}} \\
                                         &\leq \frac{3e^3 \eps L \cdot \Pr\big[s \in [t, t+1/L]\big]}{2\Pr[s \leq t+1/L]} \tag{by \eqref{eq:numeratorbound}} \\
                                         &\leq \frac{31 \eps L \cdot \Pr\big[s \in [t, t+1/L]\big]}{\Pr[s \leq t+1/L]} \\
                                         &= 31\eps L \cdot \Pr[s \geq t \mid s \leq t + 1/L].
    \end{align*}
    In order to establish our final inequality, we use $\Pr[s \geq t \mid s \leq t+1/L] \leq 1$ to find
	\begin{align*}
		\Pr[s \geq t] &= \Pr[s > t+1/L] + \Pr[s \geq t \mid s \leq t+1/L] \Pr[s \leq t+1/L]\\
		&\geq \Pr[s \geq t \mid s \leq t+1/L]\cdot\Big(\Pr[s > t+1/L] + \Pr[s \leq t+1/L]\Big)\\
		&= \Pr[s \geq t \mid s \leq t+1/L].
	\end{align*}
    These two inequalities, then, prove the lemma as
    \[
    \Pr[s \geq t-\eps \mid s\leq t] \leq 31\eps L\cdot \Pr[s \geq t \mid s \leq t + 1/L] \leq 31\eps\Pr[s \geq t].
    \]
\end{proof}

\addcontentsline{toc}{section}{References}
\printbibliography

\end{document}